\documentclass[12pt]{article}

\usepackage[margin=1in]{geometry}
\usepackage{setspace}
\usepackage{enumitem}
\usepackage{mathtools}
\usepackage{amsthm,amsmath}
\usepackage{dsfont}
\usepackage{amssymb}
\usepackage{dsfont}
\usepackage{bm}
\usepackage{breqn}
\usepackage{color}
\usepackage{xcolor}
\usepackage{pgfplots,pgfplotstable}
\pgfplotsset{compat=1.11} 
\usepackage{tikz}
\usetikzlibrary{patterns,decorations.pathreplacing}
\usepackage[authoryear]{natbib}
\usepackage[hidelinks]{hyperref}

\usepackage[flushleft]{threeparttable}

\allowdisplaybreaks

\newtheorem{lemma}{Lemma}
\newtheorem{theorem}{Theorem}
\newtheorem{assumption}{Assumption}
\newtheorem{example}{Example}[section]
\newtheorem{corollary}{Corollary}

\newtheoremstyle{named}{}{}{\itshape}{}{\bfseries}{.}{.5em}{#1\thmnote{ #3}}
\theoremstyle{named}

\newcommand{\inv}{^{-1}}
\newcommand{\isum}{\sum_{i=1}^n}

\newcommand{\E}{\mathbb{E}}

\newcommand{\R}{\mathbb{R}}

\newcommand{\Z}{\mathcal{Z}}
\newcommand{\argmin}{\text{argmin}}

\begin{document}

\title{Simple Adaptive Size-Exact Testing for Full-Vector and Subvector Inference in Moment Inequality Models\footnote{We acknowledge helpful feedback from Donald Andrews, Xiaohong Chen, Whitney Newey, Adam Rosen, Matthew Shum, J\"{o}rg Stoye, the participants of the 2nd Econometrics Jamboree at the UC Berkeley, the 2nd CEMMAP UCL/Vanderbilt Joint Conference, and econometrics seminars at Columbia University, the National University of Singapore, the UCSD, the UCLA, and the University of Wisconsin at Madison.}}
\author{Gregory Cox\footnote{Department of Economics, National University of Singapore (\href{mailto:ecsgfc@nus.edu.sg}{ecsgfc@nus.edu.sg})} 
\and Xiaoxia Shi\footnote{Department of Economics, University of Wisconsin-Madison (\href{mailto:xshi@ssc.wisc.edu}{xshi@ssc.wisc.edu})}}
\date{\today}

\maketitle

\begin{abstract}
We propose a simple test for moment inequalities that has exact size in normal models with known variance and has uniformly asymptotically exact size more generally. 
The test compares the quasi-likelihood ratio statistic to a chi-squared critical value, where the degree of freedom is the rank of the inequalities that are active in finite samples. 
The test requires no simulation and thus is computationally fast and especially suitable for constructing confidence sets for parameters by test inversion. 
It uses no tuning parameter for moment selection and yet still adapts to the slackness of the moment inequalities. 
Furthermore, we show how the test can be easily adapted for inference on subvectors for the common empirical setting of conditional moment inequalities with nuisance parameters entering linearly.
\end{abstract}

\mbox{}\\
\mbox{}\\


\pagebreak
\section{Introduction}

In the past decade or so, inequality testing has become a mainstream  inference method used for models where standard maximum likelihood or method of moments are difficult to use, for reasons including multiple equilibria, incomplete data, or complicated dynamic patterns. 
In such models, inequalities can often be derived from equilibrium conditions and rational decision making. 
Inference can then be conducted by inverting tests for these inequalities at each given parameter value.\footnote{An incomplete list of applications that use inequalities as estimation restrictions includes \cite{Tamer2003}, \cite{Uhlig2005}, \cite{BajariBenkardLevin2007}, \cite{BGIM2007}, \cite{CilibertoTamer2009}, \cite{BMM2011}, \cite{Holmes2011}, \cite{BIWY2012}, \cite{Chetty2012}, \cite{NevoRosen2012}, \cite{KawaiWatanabe2013}, \cite{Eizenberg2014}, \cite{HuberMellace2015}, \cite{PPHI2015}, \cite{MagnolfiRoncoroni2016}, \cite{Sheng2016}, \cite{Sullivan2017}, \cite{He2017}, \cite{ISS2018},
\cite{Wollman2018}, \cite{FackGrenetHe2019}, \cite{MSZ2019}. For recent overview of the literature, see for example \cite{HoRosen2017}, \cite{CanayShaikh2017}, and \cite{Molinari2020}.} 

Although conceptually simple, conducting inference via test inversion poses considerable computational challenges to practitioners. 
This is because in order to get an accurate calculation of the confidence set, one needs to test the inequalities at a set of parameter values that is dense enough in the parameter space. 
Depending on the application, the number of values needed to be tested can be astronomical and increases exponentially with the dimension of the parameter space. 
Moreover, existing tests often require simulated critical values that are nontrivial to compute even for a single value of the parameter, let alone repeated for a large number of parameter values.\footnote{Existing tests for general moment inequalities with simulated critical values include \cite{CHT2007}, \cite{RomanoShaikh2008}, \cite{AndrewsGuggenberger2009}, \cite{AndrewsSoares2010},  \cite{Bugni2010}, \cite{Canay2010},  
\cite{RomanoShaikh2012}, and \cite{RomanoShaikhWolf2014}. See \cite{CanayShaikh2017} and \cite{Molinari2020} for more references.}

Besides computational challenges, most existing methods for moment inequality models involve tuning parameter sequences that are required to diverge at a certain rate as the sample size increases. 
The threshold in the generalized moment selection procedures (e.g. \cite{Rosen2008} and  \cite{AndrewsSoares2010}) and the subsample size in subsampling-based methods (e.g. \cite{CHT2007} and \cite{RomanoShaikh2012}) are notable examples.\footnote{Arguably, the size of a first stage confidence set or the number of simulation/bootstrap draws are also tuning parameters commonly used to test moment inequalities.} 
Appropriate choices often depend on data in complicated ways, and an inappropriate choice can threaten the validity of the test. 

Clearly, there are two ways to ease the computational burden: one is to make the inequality test easier for each parameter value, and the other is to reduce the number of parameters that need to be tested. 
We contribute to the literature in both. 
First, we propose a simple test for general moment inequalities that requires no simulation. 
It simply uses the $\text{(quasi-)}$ likelihood ratio statistic ($T_n$) and a chi-squared critical value, where the data-dependent degrees of freedom come as a by-product of computing $T_n$. 
By not requiring simulation, the test saves computation time hundreds-fold compared to tests involving simulated critical values, where a critical value statistic needs to be computed for each simulated sample. 

Second, we then specialize to a conditional moment inequality model with nuisance parameters ($\delta$) entering linearly, a common empirical setting, and propose a confidence set for a subvector parameter of interest ($\theta$). 
The subvector confidence set is based on our full-vector test applied to the nuisance-parameter-eliminated model. 
By focusing on the parameter of interest, one only needs to consider a grid on the space of $\theta$ which can be much lower dimensional than the space of $(\theta',\delta')'$. 
Thus, the number of parameter values that need to be tested is drastically reduced. 
Also, as an added benefit, the subvector test is less conservative and more powerful than the projection of the full vector test, as demonstrated in our second Monte Carlo experiment.

In both contexts, our test is simulation and tuning parameter free. 
Its critical value is simply the chi-squared critical value with degrees of freedom equal to the rank of the active moment inequalities, where we call a moment inequality {\em active} if it holds with equality at the restricted estimator of the moments.\footnote{Active inequalities are the sample counterpart of binding inequalities, which hold with equality at the population expectation of the moments.} 
The test is shown to have exact finite sample size in a normal model with known variance and to be uniformly asymptotically valid more generally. 
Moreover, it automatically adapts to the slackness of the moment inequalities despite the absence of a deliberate moment selection step. 
In particular, when all but one inequality get increasingly slack, the test asymptotes to one that ignores all the slack inequalities, which coincides with the uniformly most powerful test for the limiting model.

The idea of simple chi-squared critical values for testing inequalities appeared as early as in \cite{Bartholomew1961} and \cite{Rogers1986} for testing one-sided alternatives against a {\em simple} null, but was only recently proved to be valid for a composite null in \cite{MohamadGoemanZwet2020} in a normal model. 
We move beyond \cite{MohamadGoemanZwet2020} in three ways: (a) we allow an intercept in the inequalities defining the null and thus generalize the null space from a cone to a polyhedron. 
This is important for moment inequality models as the limiting experiment may not be a cone in general when we allow local-to-binding inequalities; (b) we design a simple but novel refinement to make the test size-exact; and (c) we prove the uniform asymptotic validity of the test for moment inequality models. 

The idea of eliminating nuisance parameters from linear moment inequalities is first suggested in \cite{GuggenbergerHahnKim2008}, where they introduce the Fourier-Motzkin elimination to the literature and propose a Wald-type test on the resulting inequalities. 
Yet two main difficulties hinder the application of this idea: (a) numerical calculation of the Fourier-Motzkin elimination in general is an NP-hard computational problem, and (b) the estimated slopes in front of the nuisance parameters enter the inequalities via a non-differentiable function, and thus undermine the validity of testing procedures applied directly to them. 
The first difficulty is not present for our test because its special structure only requires us to calculate the {\em rank} of the active inequalities, avoiding the full-blown elimination procedure. 
The second difficulty is circumvented by considering the conditional distribution of the moments given a vector of instrumental variables.

\cite{AndrewsRothPakes2019} has the closest setting with our paper. 
They propose a test based on the max statistic. 
In the most basic version, their test uses a conditional critical value from a truncated normal distribution. 
This basic version involves no simulation or tuning parameter and as a result is easy to compute. 
However, the basic version has poor power properties that prompt them to recommend a hybrid test. 
The hybrid test uses a simulated critical value as well as a tuning parameter that determines the size of a first-stage confidence set. 

There are a few papers in the literature that propose methods to mitigate the computational challenges described above. 
\cite{KMS2019} cast the problem of finding the bounds of the projection confidence interval of each parameter into a nonlinear nonconvex constrained optimization problem, and provide a novel algorithm to solve this optimization problem more efficiently. 
Our simple inequality test is complementary to \cite{KMS2019}'s algorithm in that we make testing for each value hundreds-fold easier while their algorithm reduces the number of values that need to be tested. 
\cite{BugniCanayShi2017} propose a profiling method that simplifies computation in the same way as the subvector confidence set proposed in this paper, by reducing the search from the space of the whole parameter vector to that of a low dimensional subvector. 
The difference is that our subvector test, by taking advantage of the linearity of the model, is much easier to compute than \cite{BugniCanayShi2017}'s test, which applies more generally. 
\cite{CCT2018} propose a quasi-Bayesian method to subvector inference which has similar computational cost as \cite{BugniCanayShi2017} when applied to moment inequality models.\footnote{More details about the comparison of the computational aspect of these papers can be found in Section 4.3 of \cite{HsiehShiShum2020}.} 
\cite{CCT2018} also propose a simple test for a scalar parameter of interest that uses a chi-squared critical value that is valid under an additional assumption. 

A couple of other papers aim to reduce the sensitivity of testing to tuning parameters. 
\cite{AndrewsBarwick2012} (AB, hereafter) refines the procedure of \cite{AndrewsSoares2010} by computing an optimal moment selection threshold that maximizes a weight average power and a size correction. 
Using the optimal threshold and the size correction provided in that paper, one no longer needs to choose a tuning parameter. 
Computationally, it is the same as \cite{AndrewsSoares2010} if one has 10 or fewer moment inequalities and can use the tables of optimal tuning and size correction values in the paper. 
It is much more computationally demanding otherwise. 
\cite{RomanoShaikhWolf2014} (RSW, hereafter) replace the moment selection step of the previous literature with a confidence set for the slackness parameter and employ a Bonferroni correction to take into account the error rate of this confidence set. 
There is still a tuning parameter, the confidence level of the first step, but this tuning parameter no longer affects the asymptotic size of the test. 
Computationally, using the same number of critical value simulations, it is slightly more costly than \cite{AndrewsSoares2010} due to the first-step confidence set construction. 
AB and RSW are our points of comparison in our first set of Monte Carlo experiments, where we show that our simple test saves computational cost hundreds-fold, while comparing favorably to AB and RSW in terms of size and power.

The remainder of this paper proceeds as follows. 
Section 2 covers full-vector inference in moment inequality models. 
Section 3 covers the extension to subvector inference in conditional moment inequality models with nuisance parameters entering linearly.
Section 4 reports the simulation results. 
Section 5 concludes. 
An appendix contains the proofs and additional results. 

\section{Full-Vector Inference}\label{sec:momineq}

We consider a moment inequality model of the form 
\begin{align}
A\E_{F}\overline{m}_n(\theta)\le b, \label{momineq} 
\end{align}
where $A$ is a $d_A\times d_m$ matrix, $b$ is a $d_A\times 1$ vector, and $\overline{m}_n(\theta)=n\inv\isum m(W_i, \theta)$ for a $d_m$-dimensional moment function $m(\cdot,\theta)$ known up to the parameter $\theta$ and the data $\{W_i\}_{i=1}^n$ with joint distribution $F$. 
Let $\Theta$ be the parameter space for $\theta$. 
The quantities $A$ and $b$ may depend on $\theta$ and the sample size $n$, a dependence that we keep implicit for simplicity unless otherwise needed. 
The moment inequality model identifies the true parameter value up to the identified set,\footnote{If $A$ and $b$ depend on $\theta$, the formula for $\Theta_0(F)$ becomes $\{\theta\in \Theta: A(\theta)\E_{F}\overline{m}_n(\theta)\leq b(\theta)\}$.} 
\begin{equation}
\Theta_0(F) = \{\theta\in \Theta: A\E_{F}\overline{m}_n(\theta)\leq b\}.\label{IDset}
\end{equation} 

Our specification of a moment inequality model slightly differs from that in Andrews and Soares (2010, AS hereafter) by including a coefficient matrix $A$ and an intercept $b$. 
The model reduces to the setup of AS when we  let $b=0$ and 
\begin{align}
A=\begin{pmatrix}-I_p & \mathbf{0}_{p\times v}\\\mathbf{0}_{v\times p}&-I_v\\\mathbf{0}_{v\times p}&I_v\end{pmatrix}, \label{Standard A}
\end{align}
where $d_A=p+2v$, the first $p$ moments are inequalities, and the last $v$ moments are equalities. 
Introducing $A$ and $b$ is useful because it allows us to maintain an invertible variance matrix assumption while succinctly covering both equalities and inequalities, as well as accommodating upper and lower bounds with a gap between bounds that is deterministic.\footnote{For example, $\E[\bar{W}_n]-1\leq \theta\leq \E[\bar{W}_n]$ can be written in our notation with $m(w,\theta) = \theta-w$, $A = \left(\begin{smallmatrix}1\\-1\end{smallmatrix}\right)$ and $b  = \left(\begin{smallmatrix}0\\1\end{smallmatrix}\right)$.}

Like most papers in the literature, including AS, AB, and RSW, we conduct inference for the true parameter $\theta_0$ by test inversion. 
That is, for a given significance level $\alpha\in (0,1)$, one constructs a test $\phi_n(\theta,\alpha)$ for $H_0:\theta=\theta_0$, and obtains the confidence set for $\theta_0$ by calculating 
\begin{align}
CS_n(1-\alpha) = \{\theta\in\Theta:\phi_n(\theta,\alpha)=0\}.
\end{align}

\subsection{Test Construction}\label{sub:Normal}
We introduce two simple tests, one being a refinement of the other. 
Both are easy to compute, requiring no tuning parameters or simulations. 
Both use the (quasi-) likelihood ratio statistic, 
\begin{align}
T_n(\theta) = \min_{\mu: A\mu\leq b}n(\overline{m}_n(\theta) - \mu)'{\widehat\Sigma_n}(\theta)\inv(\overline{m}_n(\theta) - \mu), \label{Tn}
\end{align}
where $\widehat\Sigma_n(\theta)$ denotes an estimator of $\textup{Var}_F(\sqrt{n}\overline{m}_n(\theta))$. 

Both tests use data-dependent critical values that are based on the rank of the rows of $A$ corresponding to the inequalities that are active in finite sample. 
To define them rigorously, let $\hat{\mu}$ be the solution to the minimization problem in (\ref{Tn}). 
This is the restricted estimator for the moments.
Let $a'_j$ denote the $j$th row of $A$ and let $b_j$ denote the $j$th element of $b$ for $j = 1,2,\dots,d_A$. 
Let 
\begin{align}
\widehat{J} = \{j\in\{1,2,\dots,d_A\}: a'_j \hat \mu=b_j\}, \label{Jhat}
\end{align}
which is the set of indices for the active inequalities. 
For a set $J\subseteq\{1,2,\dots,d_A\}$, let $A_J$ be the submatrix of $A$ formed by the rows of $A$ corresponding to the elements in $J$. 
Let $\textup{rk}(A_J)$ denote the rank of $A_J$. 
Let $\hat r = \textup{rk}(A_{\widehat J})$.\footnote{$\hat\mu$, $\hat J$, and $\hat r$ depend on $\theta$, a dependence that we keep implicit for simplicity.} 

The critical value of the first simple test is the $100(1-\alpha)\%$ quantile of $\chi^2_{\hat r}$, the chi-squared distribution with $\hat r$ degrees of freedom. 
Thus, the first simple test is
\begin{align}
\phi^{\textup{CC}}_n(\theta,\alpha) = 1\left\{T_n(\theta)>\chi^2_{\hat r,1-\alpha}\right\},\label{testCC}
\end{align}
where CC stands for ``conditional chi-squared'' indicating that the test uses the chi-squared critical value conditional on the active inequalities. 
We show the validity of the CC test below. 
The intuition is that $T_n(\theta)$ follows the $\chi^2_{\hat r}$ distribution conditional on $\hat r$ when all inequalities are binding (that is, $A\E_{F}\overline{m}_n(\theta) = b$), and is stochastically dominated by the $\chi^2_{\hat r}$ distribution when some of the inequalities are slack. 

The CC test can be somewhat conservative because when $\hat r=0$ (effectively no inequality is active in finite sample), $T_n(\theta)$ is a point mass at zero and the conditional rejection probability is zero instead of $\alpha$. 
For this reason,  the null rejection probability of the CC test can be as low as $(1-\Pr(\hat r=0))\alpha$, which lies in the interval $[\alpha/2,\alpha]$. 

We propose a second simple test that eliminates the conservativeness. 
We call this the RCC (refined CC) test. 
We define the RCC test by adjusting the quantile of the $\chi^2_1$ distribution when $\hat r=1$. 
Instead of the $100(1-\alpha)\%$ quantile, the RCC test uses a $100(1-\hat\beta)\%$ quantile, where $\hat\beta$ varies between $\alpha$ and $2\alpha$ depending on how far from active the additional (inactive) inequalities are. 
We construct $\hat\beta$ carefully so that the refinement exactly restores the size of the test. When $\hat r=1$, suppose without loss of generality that the first inequality is active and satisfies $a_1\neq 0$.\footnote{In this case, other inequalities may be active too because we have not ruled out the possibility that $A$ contains redundant rows or zero rows. But this is possible only if the other active inequalities are collinear with $a_1$. } 
Now define a measure of how far from being active the other inequalities are. 
For each $j=2,...,d_A$, let 
\begin{align}
\hat\tau_{j}=\left\{\begin{array}{ll}\frac{\sqrt{n}\|a_1\|_{\widehat\Sigma_n(\theta)}\left(b_j-a'_j\hat{\mu}\right)}{\|a_1\|_{\widehat\Sigma_n(\theta)} \|a_j\|_{\widehat\Sigma_n(\theta)} - a'_1\widehat\Sigma_n(\theta)a_j}&\text{ if }\|a_1\|_{\widehat\Sigma_n(\theta)} \|a_j\|_{\widehat\Sigma_n(\theta)} \neq a'_1\widehat\Sigma_n(\theta)a_j\\
\infty &{\text{ otherwise}}\end{array}\right., \label{zj}
\end{align}
where $\|a\|_\Sigma = (a'\Sigma a)^{1/2}$. 
This is equal to zero when the $j$th inequality is active, and positive when it is inactive. 
It is scaled using the angle between $\widehat\Sigma^{1/2}_n(\theta)a_1$ and $\widehat\Sigma^{1/2}_n(\theta)a_j$.\footnote{Note that  $a'_1\widehat\Sigma_n(\theta)a_j=\|a_1\|_{\widehat\Sigma_n(\theta)} \|a_j\|_{\widehat\Sigma_n(\theta)} \cos\varphi$, where $\varphi$ stands for the angle.} 
Then let 
\begin{align}
\hat\tau=\inf_{j\in\{2,...,d_A\}} \hat\tau_j. \label{zz}
\end{align}
This quantity is easy to compute and has a nice geometric interpretation that is illustrated in Example \ref{ex:illustration2} below. 

Now we can define 
\begin{align}
\hat\beta = \left\{\begin{array}{ll}2\alpha \Phi(\hat\tau)&\text{ if }\hat r=1 \\
\alpha &\text{ otherwise}\end{array}\right., \label{beta}
\end{align}
where $\Phi(\cdot)$ is the standard normal cumulative distribution function (cdf).\footnote{$\hat\beta$, $\hat\tau_j$, and $\hat\tau$ depend on $\theta$, a dependence that we keep implicit for simplicity.} 
When a second inequality is close to being active, $\hat \tau$ is close to $0$ and then $\hat\beta$ is close to $\alpha$. 
When all the other inequalities are far from active, then $\hat\tau$ is very large and $\hat\beta$ is close to $2\alpha$. 
We define the RCC test for $H_0:\theta=\theta_0$ to be 
\begin{align}
\phi_n^{\text{RCC}}(\theta,\alpha) = 1\{T_n(\theta)>\chi^2_{\hat r,1-\hat\beta}\}.\label{RCC}
\end{align}

Since $\hat\tau\in[0,\infty]$, $\hat\beta\in[\alpha,2\alpha]$. 
Thus we have the following comparison of the CC and the RCC tests: 
\begin{align}
\phi_n^{\textup{RCC}}(\theta,\alpha/2)\leq \phi_n^{\textup{CC}}(\theta,\alpha)\leq \phi_n^{\textup{RCC}}(\theta,\alpha).\label{CCRCC_bound}
\end{align}
Moreover, when an equality is being tested, at least two inequalities are always active, in which case we have $\hat\beta=\alpha$, and the RCC test reduces to the CC test. 

We also define a reduced test that only uses a subset of the inequalities. 
For $J\subseteq\{1,...,d_A\}$, let $\phi^{\textup{RCC}}_{n,J}(\theta,\alpha)$ denote the RCC test defined with $A_J$ and $b_J$ instead of $A$ and $b$, where $b_J$ denotes the subvector of $b$ formed by the elements of $b$ corresponding to the indices in $J$. 
This test is a useful point of comparison when the inequalities not in $J$ are very slack. 

\subsection{Finite Sample Properties}\label{sub:normal-RCC}
When the moments are normally distributed with known variance, the following theorem states the finite sample properties of the RCC test. 

\begin{theorem}\label{thm:normal-rcc} 
Suppose $\Sigma_n(\theta)$ is a positive definite matrix such that $\sqrt{n}(\overline{m}_n(\theta) -\E_{F}\overline{m}_n(\theta)) \sim N(0,\Sigma_n(\theta))$ and $\widehat\Sigma_n(\theta)=\Sigma_n(\theta)$ a.s. for all $\theta\in\Theta_0(F)$. 
Then the following hold. 
\begin{enumerate}
\item[\textup{(a)}] For any $\theta\in\Theta_0(F)$, $\E_{F}\phi_n^{\textup{RCC}}(\theta,\alpha) \leq \alpha$. 
\item[\textup{(b)}] If $A\E_{F}\overline{m}_n(\theta) = b$ and $A\neq \mathbf{0}$, then $\E_{F}\phi_n^{\textup{RCC}}(\theta,\alpha) = \alpha$. 
\item[\textup{(c)}] If $J\subseteq \{1,...,d_A\}$ and  $\{\theta_s\}_{s=1}^\infty$ is a sequence such that $\theta_s\in\Theta_0(F)$ for all $s$, and for all $j\notin J$, $a_j\neq 0$ and $(a'_j\E_{F}\overline{m}_n(\theta_s) - b_j)/\|a_j\|\rightarrow-\infty$ as $s\to\infty$, then 
$$\lim_{s\to\infty}\textup{Pr}_{F}\left(\phi_n^{\textup{RCC}}(\theta_s,\alpha)\neq\phi^{\textup{RCC}}_{n,J}(\theta_s,\alpha)\right)=0.$$
\end{enumerate}
\end{theorem}

\noindent\noindent\textbf{Remarks.} (1) Part (a) shows the finite sample validity of the RCC test when the moments are normally distributed with known variance. 
Part (b) shows that the RCC test is size-exact when all the inequalities bind. 
Using (\ref{CCRCC_bound}), part (b) also implies that the finite sample size of the CC test is between $\alpha/2$ and $\alpha$ if all the inequalities may bind.
Part (c) shows that, in the presence of very slack inequalities, the RCC test reduces to the test that only uses the not-very-slack inequalities. 
Put another way, the RCC test adapts to the slackness of the inequalities. 
We call this property ``irrelevance of distant inequalities'' or IDI. 
This is especially useful if all but one inequality is very slack, because the reduced test is the one-sided t-test for the sole binding inequality, which is uniformly most powerful. 

(2) Several papers, including \cite{Kudo1963} and \cite{Wolak1987}, propose a classical test for inequalities that can be applied here. The classical test is based on the $1-\alpha$ quantile of the least favorable distribution of $T_n(\theta)$, which is a mixture of $\chi^2_0, \chi^2_1,\dots,\chi^2_{d_A}$ distributions. 
This test also has exact size, but lacks the IDI property. 
When many inequalities tested are slack, the power of this test can be very low. 
Besides, the critical value typically requires simulation, which makes it computationally less attractive than the RCC test. 

(3) AS introduced a generalized moment selection procedure that achieves an asymptotic version of the IDI property via a sequence of tuning parameters. 
AB's asymptotic normality-based test has finite sample exact size. 
It nearly has the IDI property, but not exactly. 
The size correction the test uses causes it to respond to very slack inequalities, albeit to a lesser extent than the classical test. 

(4) The only other moment inequality test with exact size and the IDI property is the non-hybrid test in \cite{AndrewsRothPakes2019}. 
This test is based on the conditional distribution of the maximum standardized element of $\sqrt{n}(A\overline{m}_n(\theta)-b)$ given the second-largest maximum. 
The test also asymptotes to the one-sided t-test when all but one inequality get increasingly slack. 
However, the test has undesirable power when multiple inequalities are not well separated, which prompts them to recommend a hybrid test instead. 

(5) Theorem 1 and the other results in this paper are stated in terms of hypothesis tests. 
However, they can be extended to results on the coverage probability of confidence sets defined by test inversion in a standard way. 
Specifically, under the conditions of Theorem 1(a), we have for all $\theta_0\in\Theta_0(F)$, 
\begin{equation}
\text{Pr}_F\left(\theta_0\in CS^{\text{RCC}}_n(1-\alpha)\right)\ge 1-\alpha,
\end{equation}
where $CS_n^{\text{RCC}}(1-\alpha)=\{\theta\in\Theta: \phi^{\text{RCC}}_n(\theta,\alpha)=0\}$ is the confidence set formed by inverting the RCC test. 

(6) The proof of Theorem 1 is challenging. 
It relies on a partition of the space of realizations of the moments according to which inequalities are active. 
It then uses a bound on probabilities of translations of sets to bound the rejection probability conditional on each set in the partition. 
\cite{MohamadGoemanZwet2020} prove a special case of part (a) for the CC test when the inequalities define a cone. 
We extend the result to the RCC test and allow the inequalities to define an arbitrary polyhedron, which are important extensions for moment inequality models. \qed

\subsection{A Simple Example}
It helps to demonstrate the CC and RCC tests in a simple two-inequality example. 
\begin{example}\label{ex:illustration2} 
Consider an example where  $d_m=2$, $A=I$, $b=\mathbf{0}$, and $\Sigma_n(\theta)=I$. 
We omit $\theta$ from the notation for ease of exposition. 
Thus, we are testing $H_0:\E_{F}\overline{m}_n\leq 0$ using the statistic $\sqrt{n}\overline{m}_n$, which follows a bivariate standard normal distribution. 

On the space of $\sqrt{n}\overline{m}_n$, the rejection region for the CC test is illustrated by the shaded region in Figure \ref{RejectionRegionEx2}. 
In this example, the likelihood ratio statistic is the squared distance between $\sqrt{n}\overline{m}_n$ and the third quadrant of the plane. 
If $\sqrt{n}\overline{m}_n$ lies in the second or fourth quadrants of the plane, one inequality is active and the $\chi^2_1$ quantile is used. 
If $\sqrt{n}\overline{m}_n$ lies in the first quadrant of the plane, two inequalities are active and the $\chi^2_2$ quantile is used. 
The critical values for the RCC test are illustrated using a dashed line where they deviate from the CC test.\footnote{The discontinuity in the critical value illustrated in Figure 1 is similar to a discontinuity in the generalized moment selection procedure proposed by AS that occurs whenever a moment is at the threshold of being selected.}

\begin{figure}
\begin{center}
\scalebox{0.75}{
\begin{tikzpicture}
\draw[pattern=north west lines, pattern color=gray!70] (0,0) rectangle (4.9,4.9);
\fill[white] (0,0) circle (4);
\draw[thick] (4,0) arc (0:90:4);
\draw[pattern=north west lines, pattern color=gray!70] (-4.9,3.2029) rectangle (0,4.9);
\draw[pattern=north west lines, pattern color=gray!70] (3.2029,-4.9) rectangle (4.9,0);
\draw[ultra thick,color=white] (-4.9,3.2029) -- (-4.9,4.9) -- (4.9,4.9) -- (4.9,-4.9) -- (3.2029,-4.9);
\draw[thick] (-4.9,3.2029) -- (0,3.2029);
\draw[thick] (3.2029,-4.9) -- (3.2029,0);
\draw[dotted, thick] (-4.9,-4.9) -- (0,0);
\fill[black] (2,-3.5) circle (0.1);
\draw (1.9,-3.9) node {\large{$\sqrt{n}\overline{m}_n$}};
\fill[black] (0,-3.5) circle (0.1);
\draw (-0.65,-3.15) node {\large{$\sqrt{n}\hat{\mu}$}};
\fill[black] (-3.5,-3.5) circle (0.1);
\draw (-3.5,-3.9) node {\large{$y$}};
\draw [decorate, decoration={brace,amplitude=10pt},xshift=-0pt,yshift=-5pt,thick] (-0.2,-3.5) -- (-3.3,-3.5);
\draw (-1.85,-4.4) node {\large{$\hat\tau$}};
\draw[dashed, thick] (-4.9,2.6890) -- (-3.5,2.7008) --(-3,2.7146)-- (-2,2.7796) -- (-1.5,2.8415) -- (-1,2.9301) -- (-0.75,2.9857)--(-0.5,3.0496) -- (-0.25,3.1219) -- (-0.1,3.1695) to[out=0,in=-155] (0,3.2029); 
\draw[dashed, thick] (2.6890,-4.9) --  (2.7008,-3.5) --(2.7146,-3)-- (2.7796,-2) -- (2.8415,-1.5)-- (2.9301,-1) -- (2.9857,-0.75)--(3.0496,-0.5)--(3.1219,-0.25) -- (3.1695,-0.1) to[out=90,in=-115] (3.2029,0);
\draw[->,thick] (-3.5,-3.5)--(5,-3.5);
\draw[->,ultra thick] (-5,0)--(5,0) node[right]{\Large{$\sqrt{n}\overline{m}_{n1}$}};
\draw[->,ultra thick] (0,-5)--(0,5) node[above]{\Large{$\sqrt{n}\overline{m}_{n2}$}};
\end{tikzpicture}
}
\end{center}
\caption{Geometric representation of the CC test (shaded) and the RCC test (dashed) in Example \ref{ex:illustration2}. 
}
\label{RejectionRegionEx2}
\end{figure}
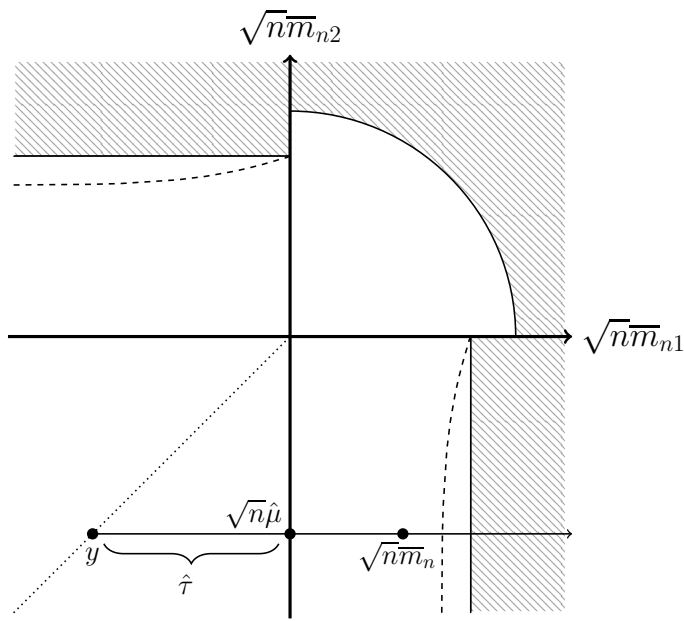
From the figure, we can see that the RCC test deviates from the CC test only when the number of active inequalities is one (in the second and fourth quadrants of the plane). 
In that case, a smaller critical value is used that depends on how far from active the other inequality is, measured using $\hat\tau$. 
The quantity $\hat\tau$ has the following geometric interpretation: the point $\sqrt{n}\hat{\mu}$ is the projection of $\sqrt{n}\overline{m}_n$ onto a face of the polyhedron defined by the inequalities. 
Continue that line into the interior of the polyhedron until you reach a point, $y$, that is equidistant between two inequalities. 
In the figure, the set of points that are equidistant between two inequalities is represented by the dotted line, which is the 45-degree line. 
Then $\hat\tau$ is the distance between $\sqrt{n}\hat{\mu}$ and $y$. 
This geometric interpretation extends to more complicated examples with more inequalities or non-orthogonal inequalities. 

The reason the refinement still controls size is that we condition on the event that $\sqrt{n}\overline{m}_n$ belongs to the ray that starts at $y$ and emanates through $\sqrt{n}\hat\mu$ and $\sqrt{n}\overline{m}_n$. 
It is sufficient to control the conditional rejection probability for every such ray. 
By conditioning on the ray, the denominator of the conditional rejection probability is $\Phi(\hat\tau)$, which allows us to adjust $\alpha$ up to $\hat\beta$. 
\end{example}

\subsection{Asymptotics}

Now we turn to a model without normality or a known variance matrix. 
For expositional purposes, we focus on the independent and identically distributed (i.i.d.) data case here, while results in Appendix \ref{sec:asyRCC} cover more general cases.

With i.i.d. data, we can estimate the variance matrix $\Sigma_n(\theta)$ with the usual sample variance matrix:
\begin{equation}
\widehat{\Sigma}_n(\theta) = n^{-1}\sum_{i=1}^n(m(W_i,\theta) - \overline{m}_n(\theta))(m(W_i,\theta) - \overline{m}_n(\theta))'.
\end{equation}
We show that the RCC test has correct asymptotic size uniformly over a large class of data generating processes. 

The following assumption defines the set of data generating processes allowed. 
Here $|\cdot|$ denotes the matrix determinant, and $\epsilon$ and $M$ are fixed positive constants that do not dependent on $F$ or $\theta$. 
This set of assumptions is weak and also seen in \cite{AndrewsGuggenberger2009} or AS. 
\begin{assumption}\label{assu:iidasy} 
For all $F\in\mathcal{F}$ and $\theta\in\Theta_0(F)$, the following hold. 
\begin{enumerate}
\item[\textup{(a)}] $\{W_i\}_{i=1}^n$ are i.i.d. under $F$. 
\item[\textup{(b)}] $\sigma_{F,j}^2(\theta):=\textup{Var}_{F}(m_j(W_i,\theta))>0$ for $j=1,\dots,d_m$.
\item[\textup{(c)}] $|\textup{Corr}_{F}(m(W_i,\theta))|>\epsilon$, where $\textup{Corr}_{F}(m(W_i,\theta))$ is the correlation matrix of the random vector $m(W_i,\theta)$ under $F$.
\item[\textup{(d)}] $\E_{F}|m_j(W_i,\theta)/\sigma_{F,j}(\theta)|^{2+\epsilon}\leq M\text{ for }j=1,\dots,d_m$.
\end{enumerate}
\end{assumption}
Let $D_F(\theta)$ denote the diagonal matrix formed by $\sigma_{F,j}^2(\theta):j=1,\dots,d_m$. 
For $J\subseteq\{1,\dots,d_A\}$, let $I_J$ denote the rows of the identity matrix corresponding to the indices in $J$.\footnote{Note that $I_JA$ is an alternate notation for $A_J$.} 
The following theorem states the asymptotic properties of the RCC test. 
\begin{theorem}\label{thm:asysize-iid}
Suppose Assumption \textup{\ref{assu:iidasy}} holds. 
\begin{enumerate}
\item[\textup{(a)}] Then 
\[
\underset{n\to\infty}{\textup{limsup}}\sup_{F\in {\cal F}}\sup_{\theta\in\Theta_0(F)}\E_{F}\phi^{\textup{RCC}}_n(\theta,\alpha)\leq \alpha.
\]
And for a sequence  $\{(F_n,\theta_n):F_n\in{\cal F},\theta_n\in \Theta_0(F_n)\}_{n=1}^\infty$  such that $A(\theta_n)D_{F_n}(\theta_n)\to A_\infty$ and for all $J\subseteq{1,\dots,d_A}, \textup{rk}(I_JA(\theta_n)D_{F_n}(\theta_n)) = \textup{rk}(I_JA_\infty)$,

\item[\textup{(b)}] if $A_\infty\neq \mathbf{0}$ and for all $j\in\{1,...,d_A\}$, $\sqrt{n}(a'_j \E_{F_n}\overline{m}_n(\theta_n)-b_j)\rightarrow 0$, then 
\[
\lim_{n\rightarrow\infty}\E_{F_n}\phi^{\textup{RCC}}_n(\theta_n,\alpha)=\alpha, \text{ and}
\]

\item[\textup{(c)}] if instead there is a ${J}\subseteq\{1,\dots,d_A\}$ such that  for all $j\notin {J}$, $\sqrt{n}(a'_j \E_{F_n}\overline{m}_n(\theta_n)-b_j)\rightarrow-\infty$ as $n\rightarrow\infty$, then
\[
\lim_{n\rightarrow\infty} \textup{Pr}_{F_n}\left(\phi^{\textup{RCC}}_n(\theta_n,\alpha)\neq\phi^{\textup{RCC}}_{n,{J}}(\theta_n,\alpha) \right)=0. 
\]
\end{enumerate}
\end{theorem}
\noindent\textbf{Remarks.}
(1) Part (a) shows that the RCC test is asymptotically uniformly valid. 
Part (b) shows that when all the inequalities bind, or are sufficiently close to binding, the RCC test does not under-reject asymptotically. 
Part (c) shows an asymptotic IDI property of the RCC test: if some of the inequalities are very slack, the test reduces to the one based only on the not-very-slack inequalities. 
Parts (b) and (c) can be combined to show that the RCC test has exact asymptotic size (and hence is asymptotically non-conservative) when there exists at least one fixed triple $(F,\theta,j)$ such that $\theta\in\Theta_0(F)$, $a_j\neq \mathbf{0}$ and $a_j'\mathbb{E}_Fm(W_i,\theta) = b_j$. 

(2) Theorem \ref{thm:asysize-iid} combines with (\ref{CCRCC_bound}) to imply that the CC test is asymptotically uniformly valid, and when the RCC test is asymptotically non-conservative, it can only be conservative to a limited extent:
\begin{equation}
\alpha/2\leq \underset{n\to\infty}{\textup{limsup}}\sup_{F\in {\cal F}}\sup_{\theta\in\Theta_0(F)}\E_{F}\phi^{\textup{CC}}_n(\theta,\alpha)\leq \alpha. 
\end{equation}

(3) The outline of the proof of Theorem \ref{thm:asysize-iid}(a) is conceptually simple. 
The almost sure representation theorem is invoked on the convergence of the moments, and then Theorem \ref{thm:normal-rcc} is invoked on the limiting experiment. 
However, the details are quite complicated. 
A technical complication that arises is that the rank of the inequalities can be lower in the limit than in the finite sample. 
This is handled by adding additional inequalities so the limit is an appropriate approximation to the finite sample (see Lemma \ref{BExistence} in the appendix). 

(4) Theorem \ref{thm:asysize-iid} is stated for finitely many inequalities, but the test may also work well in high dimensions. 
The biggest challenge in high dimensions is the covariance matrix estimation. 
If a good covariance matrix estimator can be found and the moments are approximately normal, then one can appeal to Theorem \ref{thm:normal-rcc} as a good approximation. 
One way to improve covariance matrix estimation is to assume sparsity or use shrinkage as in \cite{LedoitWolf2012}. 
\qed

\section{Subvector Inference}
\label{sec:sub}
In this section, we apply the CC test and the RCC test to a subvector inference problem in a conditional moment inequality model: 
\begin{align}
\E_{F_Z}[B_Z\overline{m}_n(\theta) -C_Z\delta|Z]\leq d_Z, ~a.s.\label{condmom}
\end{align}
where $Z=\{Z_i\}_{i=1}^n$ is a sample of instrumental variables (each $Z_i$ is taken to be a subvector of $W_i$ without loss of generality), $B_Z$, $C_Z$, and $d_Z$ are $k\times d_m$, $k\times p$, and $k\times 1$ matrices, $\delta$ is an unknown nuisance parameter, $\theta$ is the unknown parameter of interest, and $F_Z$ denotes the conditional distribution of $\{W_i\}_{i=1}^n$ given $\{Z_i\}_{i=1}^n$. 
The subscript $Z$ is used to denote dependence on $Z_1,...,Z_n$. 
The quantities $B_Z$, $C_Z$, and $d_Z$ may also depend on $\theta$ and the sample size $n$, a dependence that we keep implicit for simplicity. 
Similar to the full-vector case, this model can succinctly cover both equalities and inequalities by an appropriate choice of $B_Z$, $C_Z$, and $d_Z$, as well as accommodating upper and lower bounds with a gap between the bounds that depends on $\{Z_i\}_{i=1}^n$. 

The model (\ref{condmom}) is similar to that considered in \cite{AndrewsRothPakes2019}, and both are special cases of conditional moment inequality models compared to the  setup of \cite{AndrewsShi2013}. 
It has two key features: (a) the nuisance parameter $\delta$ enters linearly, and (b) the coefficients on $\delta$ depend only on the exogenous variables $\{Z_i\}_{i=1}^n$. 
These two features allow us to develop a simple subvector test that is also tuning parameter and simulation free.\footnote{\cite{AndrewsRothPakes2019} rely on these features as well.} 
Special as it is, these features reflect a recurring theme of many empirical models: exogenous covariates are used to incorporate heterogeneity and/or to control for confounders. 
Here we consider three examples.
\medskip

\begin{example}\label{ex:intreg0} 
\cite{ManskiTamer2002} consider an interval regression model: 
\begin{equation}
Y^\ast_i = X'_i\theta + Z'_{ci}\delta+\varepsilon_i,
\end{equation}
where $Y^\ast_i$ is a dependent variable, $X_i$ is a vector of possibly endogenous variables, $Z_{ci}$ is a vector of exogenous covariates including the constant. 
There is a vector of excluded instrumental variables $Z_{ei}$ that satisfies $\E[\varepsilon_i|Z_i]=0$, where $Z_i = (Z'_{ci},Z'_{ei})'$. 
The outcome $Y^\ast_i$ is not observed. 
Instead, $Y_{Li}$ and $Y_{Ui}$ are observed such that $Y^\ast_i\in [Y_{Li}, Y_{Ui}]$. 
The imperfect observation of $Y^\ast_i$ may be caused by missing data or survey design where respondents are given a few brackets to select from instead of asked to give a precise answer. 

Let $I(Z_i)$ be a finite non-negative vector of instrumental functions. 
Then we have
\begin{align}
E\left[\left.\begin{pmatrix}(Y_{Li}-X'_i\theta_0)I(Z_i)\\-(Y_{Ui}-X'_i\theta_0)I(Z_i)\end{pmatrix} - \begin{pmatrix}I(Z_i)Z'_{ci}\\-I(Z_i)Z'_{ci}\end{pmatrix}\delta_0\right|Z_i\right]\leq 0,
\end{align}
which yields a model of the form (\ref{condmom}) with $B_Z = I$, $W_i = (Y_{Li},Y_{Ui},X'_i,Z'_i)'$, $m(W_i,\theta) = \left(\begin{smallmatrix}(Y_{Li}-X'_i\theta)I(Z_i)\\ -(Y_{Ui} - X'_i\theta)I(Z_i)\end{smallmatrix}\right)$, $C_Z= n^{-1}\sum_{i=1}^n \left(\begin{smallmatrix}I(Z_i)Z'_{ci}\\-I(Z_i)Z'_{ci}\end{smallmatrix}\right)$, and $d_Z=\mathbf{0}$. 
\end{example}

\begin{example}\label{ex:intreg} 
The second example is a generalized interval regression example, where the model is 
\begin{equation}
\psi(Y^\ast_i,X_i,\theta_0) = Z'_{ci}\delta_0+\varepsilon_i, ~~\mathbb{E}[\varepsilon_i|Z_i]=0,\label{momentcondition}
\end{equation}
where $\psi$ is a known function but $Y_i^\ast$ is unobserved. 
Using auxiliary data and model assumptions, one can construct bounds $\psi^U_i(\theta_0)$ and $\psi^L_i(\theta_0)$ such that 
\begin{align}
\E[\psi_i^L(\theta_0)|Z_i]\leq \E[\psi(Y^\ast_i,X_i,\theta_0)|Z_i]\leq \E[\psi_i^U(\theta_0)|Z_i].\label{bounds}
\end{align}
Then, analogously to the previous example, we have
\begin{align}
\E\left[\left.\begin{pmatrix}\psi_i^L(\theta_0)I(Z_i)\\-\psi_i^U(\theta_0)I(Z_i)\end{pmatrix} - \begin{pmatrix}I(Z_i)Z'_{ci}\\-I(Z_i)Z'_{ci}\end{pmatrix}\delta_0\right|Z_i\right]\leq 0,
\end{align}
which also yields a model of the form (\ref{condmom}).

Such a model is constructed in \cite{GandhiLuShi2019} to conduct inference for the aggregate demand function when observed market share has many zero values. 
There $\psi$ is an inverse demand function, $Y^\ast_i$ is the unobserved choice probability of differentiated products. 
In Section \ref{sub:mc-subv}, we consider a Monte Carlo example of this model where we also provide more details of the bound construction. 
In the application of \cite{GandhiLuShi2019}, control variables ($Z_{ci}$) are essential for the validity of the instruments.
\end{example}

\begin{example} 
\cite{Eizenberg2014} studies the portable PC  market to quantify the welfare effect of eliminating a product. 
Central to the question is the fixed cost of providing the product. 
Eizenberg bounds the fixed cost of product $i$, say $\zeta_i$ using the revealed preference approach. 
To describe the bound, let $u_\ell$ be a vector of zeroes and ones of the same length as the number of potential products that firm $\ell$ can provide. 
The value $1$ indicates that the product is provided, and the value $0$ indicates otherwise. 
Let $f(u_\ell)$ denote expected variable profit for providing  the products indicated by $u_\ell$. 
Then, revealed preference of firm $\ell$ implies that
\begin{align}
\zeta_i&\geq f(u_\ell+e_i) - f(u_\ell)\\
\zeta_i&\leq f(u_\ell) - f(u_\ell-e_i),\nonumber
\end{align}
where $e_i$ is a vector of the same length as $u_\ell$ whose $i$th element is equal to 1 and all other elements are zeroes.

Let $Z_i$ be a vector of product characteristics (including the constant). 
One can consider the following conditional moment inequality model:
\begin{align}
\E\left[(f(u_\ell+e_i) - f(u_\ell) - P(Z_i)'\gamma_0)I(Z_i)|Z_i\right]&\leq 0\\
\E\left[(-f(u_\ell) + f(u_\ell-e_i) + P(Z_i)'\gamma_0)I(Z_i)|Z_i\right]&\leq 0,\nonumber
\end{align}
where $P(Z_i)$ is a vector of known functions of $Z_i$ to accommodate nonlinearity and $I(Z_i)$ is a vector of nonnegative instrumental functions. 
The function $P(Z_i)'\gamma_0$ captures the (observed) heterogeneity of fixed cost across products. 
Using our method, one can construct confidence intervals for each element of $\gamma_0$ and any linear combinations of $\gamma_0$ such as the average derivative.
\end{example}

Two additional examples that fit into our framework are \cite{Katz2007} and \cite{Wollman2018} as reviewed in \cite{AndrewsRothPakes2019}. 

To conduct inference on $\theta$, we invert a test for $H_0:\theta=\theta_0$. 
This is equivalent to the following null hypothesis for a given value $\theta_0$:
\begin{align}
H_0: \exists \delta~\text{such that}~~C_Z\delta\geq B_Z\E_{F_Z}[\overline{m}_n(\theta_0)|Z]-d_Z, ~a.s.\label{H00}
\end{align}
In the next few subsections, we proceed to construct a computationally simple, tuning parameter and simulation free test for (\ref{H00}).

\subsection{Eliminating the Nuisance Parameter}\label{sub:elimination}
Directly testing (\ref{H00}) is difficult because it requires checking the validity of the inequality for all possible values of $\delta$. 
Instead, we construct a representation of the hypothesis that does not involve $\delta$. 
The construction makes use of the following lemma that is a corollary of Theorem 4.2 of \cite{Kohler1967}. 
Kohler's result is a refinement of the well-known Fourier-Motzkin method for eliminating variables from a system of linear inequalities. 
The latter was first introduced to the moment inequality literature by \cite{GuggenbergerHahnKim2008}. 
\begin{lemma}\label{lem:dual}
Let $B$ and $C$ be conformable matrices and $d$ be a conformable vector. 
There exists a matrix $A(B,C)$ and a vector $b(C,d)$ such that 
$$
\{\delta: C\delta\geq B\mu - d\}\neq\emptyset~\Leftrightarrow~ A\mu\leq b.
$$
Furthermore, $A(B,C) = H(C)B$ and $b(C,d) = H(C)d$, where $H(C)$ is the matrix with rows formed by the vertices of the polyhedron $\{h\in R^k: h\geq 0, C'h=\mathbf{0},\mathbf{1}'h=1\}$.
\end{lemma}

It is immediately implied by Lemma \ref{lem:dual} that the statement in (\ref{H00}) is equivalent to
\begin{align}
A_Z\E_{F_Z}[\overline{m}_n(\theta_0)|Z]\leq b_Z, \label{eliminated}
\end{align}
where $A_Z = A(B_Z,C_Z)$ and $b_Z=b(C_Z,d_Z)$, which, like $B_Z, C_Z$, and $d_Z$, 
may also depend on $Z_1,...,Z_n$. 
This is the same as model (\ref{momineq}) except that the expectation is conditional on the instrumental variables. 

If $H(C_Z)$ and thus $A_Z$ and $b_Z$ can be obtained, one can apply a conditional version of any inequality testing procedures on (\ref{eliminated}). 
However, a significant obstacle is that calculating $H(C_Z)$ requires enumerating  the vertices of a polyhedron. 
Vertex enumeration is doable in small dimension, but becomes highly nontrivial as the number of conditional moment inequalities increases because the number of vertices can grow exponentially with $k$. 
See for example \cite{SierksmaZwols2015}. 

On the other hand, we shall see below that, if one applies the CC test on the hypothesis (\ref{eliminated}), one only needs to know the {\em rank} of the rows of $A_Z$ corresponding to the active inequalities. 
And this rank can be obtained by solving a relatively small number of linear programming problems without computing $A_Z$ or $b_Z$. 
This is the key insight that makes the subvector version of the CC and RCC tests feasible. 

\subsection{The Subvector CC and RCC Tests}\label{sub:sub-normal}

\subsubsection{The Subvector CC Test}
Let $\widehat\Sigma_n(\theta)$ denote an estimator of the conditional variance matrix of the moments: 
\begin{align}
\Sigma_n(\theta) = \textup{Var}(\sqrt{n}\overline{m}_n(\theta)|Z). 
\end{align}
The subvector CC test for $H_0:\theta=\theta_0$ is the CC test defined in (\ref{Tn})-(\ref{testCC}) for the inequalities (\ref{eliminated}). 
We denote it by $\phi_n^{\textup{sCC}}(\theta,\alpha)$ to distinguish it from the full-vector CC test.

Now we describe how to compute $T_n(\theta)$ and $\hat r$ without computing $A_Z$ or $b_Z$. 
For $T_n(\theta)$, it is immediate from Lemma \ref{lem:dual} that it can be equivalently written as
\begin{align}
T_n(\theta) = \min_{\delta,\mu: C_Z\delta\geq B_Z\mu - d_Z}n(\overline{m}_n(\theta)-\mu)'\widehat\Sigma_n(\theta)^{-1}(\overline{m}_n(\theta)-\mu).\label{Tnsub}
\end{align}
Let $\hat{\mu}$ denote the minimizer. 
For $\hat r$, we first give the following lemma. 
In the lemma, ${\cal H} = \{h\in R^k:h\geq 0, C'_Zh=\mathbf{0},h'(B_Z\hat{\mu} - d_Z) = 0\}$, $B'_Z{\cal H} = \{B'_Zh:h\in{\cal H}\}$, and $\text{rk}(S)$ is the maximum number of linearly independent vectors in $S$. 
\begin{lemma}\label{lem:rank} 
$\hat r=\textup{rk}(B'_Z{\cal H}).$ 
\end{lemma}
The rank of the polyhedral cone $B'{\cal H}$ for some matrix $B$ is the dimension of $\textup{span}(B'{\cal H})$, the linear span of $B'{\cal H}$. 
The key to calculate it is to find the set of linear equalities that define $\textup{span}(B'{\cal H})$. 
This can be done by finding the parametric form of ${\cal H}$ as described in \cite{HuynhLassezLassez1992}. 
That is, a representation of ${\cal H}$ of the form
\begin{align}
{\cal H} = \{h\in \mathbb{R}^k: h_{J_{co}^c} = G_1h_{J_{co}}, G_0h_{J_{co}}\geq0\},\label{paraform}
\end{align}
where $J_{co}$ is a subset of $\{1,\dots,k\}$, $J_{co}^c$ is its complement, $h_J$ is the subvector of $h$ selected by the index set $J$, and $G_1$ and $G_0$ are two matrices constructed so that (\ref{paraform}) holds and $\{h_{J_{co}}:G_0h_{J_{co}}\geq 0\}$ is a full-dimensional polyhedral cone. 
In other words, the parametric form divides $h$ into a core subvector $h_{J_{co}}$ and a non-core subvector $h_{J_{co}^c}$, such that $h_{J_{co}^c}$ is the largest subvector that can be written as a linear combination of $h_{J_{co}}$. 
Based on (\ref{paraform}), it is clear that 
\begin{align}
\textup{span}({\cal H}) = \{h\in \mathbb{R}^k: h_{J_{co}^c }= G_1 h_{J_{co}}\}.
\end{align}
Accordingly, algebra shows that
\begin{align}
\textup{span}(B'{\cal H}) = B'\textup{span}({\cal H}) = \{(B_{J_{co}^c}'G_1+B_{J_{co}}')h_{J_{co}}:h_{J_{co}}\in \mathbb{R}^{k_{co}}\},
\end{align}
where $B_{J}$ is the matrix formed by the rows of $B$ corresponding to the indices $J$, and $k_{co}$ is the number of elements in $J_{co}$. 
This implies that $\text{rk}(B'{\cal H}) = \text{rk}(G_1'B_{J_{co}^c}+B_{J_{co}})$. 
Finally, invoking Lemma \ref{lem:rank}, we have,
\begin{equation}
\hat r = \text{rk}(G_1'B_{J_{co}^c}+B_{J_{co}}).
\end{equation}

It remains to find the parametric form (\ref{paraform}). 
In particular we need to find the matrix $G_1$.\footnote{The matrix $G_0$ needs not be computed since it does not enter the subsequent rank calculation.} 
\cite{HuynhLassezLassez1992} present a way to do this by solving $k$ linear programming (LP) problems. 
The first step is to determine if the constraint $h_j\geq 0$ is an implicit (or implied) equality, which holds if 
\begin{align}
{\cal H} = \{h\in \R^k:h_j=0,h_{-j}\geq 0,C'_Zh= 0, (B_Z\hat{\mu}-d_Z)'h = 0\}, \label{eq:Bset}
\end{align}
where $h_j$ is the $j$th element of $h$ and $h_{-j}$ is $h$ with its $j$th element removed. 
To determine if $h_j\geq 0$ is an implicit equality, one simply solves the LP problem:
\begin{align}
\min_{h} - e_j'h&\label{LPj}\\
s.t.~~~~~ h&\geq 0\nonumber\\
C'_Zh&=\mathbf{0}\nonumber\\
(B_Z\hat{\mu}-d_Z)'h&=0,\nonumber
\end{align}
where $e_j$ is the $j$th column of $I_k$. 
Then $h_j\geq 0$ is an implicit equality if and only if the minimum value of this linear programming problem is zero.  

Let $J_{eq}$ denote the set of all $j\in \{1,\dots,k\}$ such that $h_j\geq 0$ is an implicit equality. 
Let $I_{J}$ denote the submatrix of $I_{k}$ formed by rows of $I_{k}$ corresponding to indices in $J$. 
Then ${\cal H}$ is defined by $h\geq 0$ and the following linear equation system
\begin{align}
\left(\begin{smallmatrix}I_{J_{eq}}\\C'_Z\\(B_Z\hat{\mu}-d_Z)'\end{smallmatrix}\right)h =0.\label{equalities}
\end{align}
Finding the matrix $G_1$ amounts to solving for the $k-k_{co}$ redundant variables from this equation system where $k_{co} = k-\text{rk}\left(\begin{smallmatrix}I_{J_{eq}}\\C'_Z\\(B_Z\hat{\mu}-d_Z)'\end{smallmatrix}\right)$. 
This can be done easily via Gauss-Jordan reduction.

In fact, in the special case that $\text{rk}(B_Z) = k$, we have $\text{rk}(B'_Z{\cal H}) = \text{rk}({\cal H})$, and there is no need to even calculate $G_1$. 
This is because $\text{rk}({\cal H}) = k_{co} = k-\text{rk}\left(\begin{smallmatrix}I_{J_{eq}}\\C'_Z\\(B_Z\hat{\mu}-d_Z)'\end{smallmatrix}\right)$ by the definition of the parametric form. \medskip

To sum up, we propose the following procedure to implement the sCC test.
\begin{enumerate}
\item[Step 1.] Calculate $T_n$ according to (\ref{Tnsub}). If $T_n=0$, let $\hat r=0$. 
Otherwise, proceed to the next step.
\item[Step 2.] For each $j$, solve the linear programming (LP) problem (\ref{LPj}) and collect the $j$'s such that the minimum value of the LP problem is 0 in $J_{eq}$.
\item[Step 3.] Calculate $k_{co} = k-\text{rk}\left(\begin{smallmatrix}I_{J_{eq}}\\C'_Z\\(B_Z\hat{\mu}-d_Z)'\end{smallmatrix}\right)$. 
If $\text{rk}(B_Z) = k$, let $\hat r= k_{co}.$
Otherwise proceed to the next step.
\item[Step 4.] Apply Gauss-Jordan reduction on (\ref{equalities}) to find $G_1$ and $J_{co}$. Then let $\hat r= \text{rk}(G_1'B_{J_{co}^c}+B_{J_{co}}).$
\end{enumerate}

In this procedure, the most time consuming step is Step 2 where we solve $k$ LP problems. 
Fast and accurate algorithms for LP problems are well-developed and widely available, which makes the CC test feasible even for a large number of inequalities and nuisance parameters.

\subsubsection{The Subvector RCC Test}\label{sec:sRCC}
The subvector RCC test for $H_0:\theta=\theta_0$ is the RCC test defined in (\ref{RCC}) for the inequalities (\ref{eliminated}). 
We denote it by $\phi_n^{\textup{sRCC}}(\theta,\alpha)$ to distinguish it from the full-vector RCC test. 

Note that it is only when  $\hat r = 1$ {\em and} $T_n(\theta)\in [\chi^2_{1,1-2\alpha},\chi^2_{1,1-\alpha}]$ that the RCC test can possibly differ from the CC test. Thus, the RCC test can be performed via the following steps. 

\begin{enumerate}
\item[Step R1.] Compute $T_n(\theta)$ and $\hat r$ using the approach described above for the sCC test. 

If $\hat r\neq 1$ or $T_n(\theta)\notin [\chi^2_{1,1-2\alpha},\chi^2_{1,1-\alpha}]$, stop and let
$$
\phi_n^{\textup{sRCC}}(\theta,\alpha)  =\phi_n^{\textup{sCC}}(\theta,\alpha).
$$
Otherwise, proceed to the next step. 

\item[Step R2.] Compute $H(C_Z)$ using a polyhedron vertex enumeration algorithm. 
For example, we use {\bf {con2vert.m}} in Matlab.\footnote{The Matlab function {\bf con2vert.m} requires the polyhedron to be bounded and full-dimensional. Thus, we apply this function after reducing the bounded polyhedron $\{h\in R^k:h\geq 0, C_Z'h=0, \mathbf{1}'h=1\}$ to a lower, full-dimensional polyhedron by finding the parametric form using the method in Huynh et al. (1992).} 
Let $A_Z = H(C_Z)B_Z$ and $b_Z = H(C_Z)d_Z$. 
Then define $\hat\beta$ by the formulae in (\ref{zj})-(\ref{beta}). 
Finally let 
\begin{align}
\phi_n^{\textup{sRCC}}(\theta,\alpha)=1\{T_n(\theta)>\chi^2_{1,1-\hat\beta}\}.\nonumber
\end{align}
\end{enumerate}
The vertex enumeration part of Step R2 can be difficult for large $k$ and $p$. 
However, notice that Step R2 is only needed when $\hat r=1$ and $T_n(\theta)\in [\chi^2_{1,1-2\alpha},\chi^2_{1,1-\alpha}]$ in Step R1, which is infrequent for large $k$. 

\subsection{Finite Sample Validity of the sCC and sRCC Tests}
The following result shows the finite sample properties of the sRCC test assuming normally distributed moments and a known conditional variance matrix. 
The result is a corollary of Theorem \ref{thm:normal-rcc}. 
Let $z$ be a realization of $Z$, and let $\Theta_0(F_z) = \{\theta\in\Theta: \exists\delta\textup{ s.t. }C_z\delta\geq B_z\E_{F_z}[\overline{m}_n(\theta)|z]-d_z\}$. 
Let $A_z=A(B_z,C_z)$ and $b_z=b(C_z,d_z)$ from Lemma \ref{lem:dual}. 
Let $e_j$ denote the $R^{d_A}$-vector with $j$th element one and all other elements zero. 
For any $J\subseteq \{1,\dots,d_A\}$, let $\phi^{\textup{sRCC}}_{n,J}(\theta,\alpha)$ denote the sRCC test defined using $I_J A_z$ and $I_Jb_z$ in place of $A_z$ and $b_z$. 

\begin{corollary}\label{cor:normal} 
Suppose $\Sigma_n(\theta)$ is a positive definite matrix such that the conditional distribution of $\sqrt{n}(\overline{m}_n(\theta) - \E_{F_Z}\overline{m}_n(\theta))$ given $Z=z$ is distributed $N(\mathbf{0},\Sigma_n(\theta))$ and $\widehat\Sigma_n(\theta)=\Sigma_n(\theta)$ a.s. for all $\theta\in\Theta_0(F_z)$. 
Then the following hold. 
\begin{enumerate}
\item[\textup{(a)}] For any $\theta\in \Theta_0(F_z)$, $\E_{F_z}[\phi_n^{\textup{sRCC}}(\theta,\alpha)|z]\leq \alpha$. 
\item[\textup{(b)}] If $A_z\E_{F_z}[\overline{m}_n(\theta)|z]=b_z$ and $A_z\neq\mathbf{0}$, then $\E_{F_z}[\phi_n^{\textup{sRCC}}(\theta,\alpha)|z]= \alpha$. 
\item[\textup{(c)}] If $J\subseteq\{1,...,d_A\}$ and $\{\theta_s\}_{s=1}^\infty$ is a sequence such that $\theta_s\in\Theta_0(F_z)$ for all $s$, and for all $j\notin J$, $e'_jA_z\neq 0$ and $e'_{j}(A_z\E_{F_z} \overline{m}_n(\theta_s)-b_z)/\|e'_jA_z\|\rightarrow-\infty$ as $s\rightarrow\infty$, then 
\[
\lim_{s\rightarrow\infty}\textup{Pr}_{F_z}\left(\phi^{\textup{sRCC}}_{n}(\theta_s,\alpha)\neq\phi^{\textup{sRCC}}_{n,J}(\theta_s,\alpha)|z\right)=0. 
\]
\end{enumerate} 
\end{corollary}

\noindent\noindent\textbf{Remarks.} 
(1) Part (a) shows the finite sample validity of the sRCC test under normality. 
Part (b) states mild conditions under which the sRCC test is not conservative. 
Part (c) states the IDI property of the sRCC test. 

(2) Since the sRCC test rejects whenever the sCC test does, the corollary implies the validity of the sCC test: $\E_{F_z}[\phi_n^{\textup{sCC}}(\theta,\alpha)|z]\leq \alpha$. 

(3) A result for asymptotically uniform size control of the subvector tests is available in the appendix. 
It relies on the asymptotic normality of the moments conditional on $Z_1,...,Z_n$ and a consistent estimator for $\Sigma_n(\theta)$. 
When the data are i.i.d., we provide low level conditions for two cases: discrete $Z_i$ and continuous $Z_i$. 

In the first case, suppose $Z_i$ takes on a finite number of values in a set $\Z$. 
A straightforward estimator of $\textup{Var}_{F_z}(\sqrt{n}\overline{m}_n(\theta))$ is the weighted average of the sample variances of $m(W_i,\theta)$ within each category of $Z_i$: 
\begin{align}
\widehat{\Sigma}_n(\theta) = \sum_{\ell\in \Z}\frac{n_\ell}{n}\frac{1}{n_\ell-1}\sum_{i=1}^n(m(W_i,\theta) - \overline{m}_{n}^\ell(\theta))(m(W_i,\theta) - \overline{m}_{n}^\ell(\theta))'1\{z_i= \ell\}, \label{condVar1}
\end{align}
where $n_\ell = \sum_{i=1}^n 1\{z_i=\ell\}$ and $\overline{m}_n^\ell(\theta) = \frac{1}{n_\ell}\sum_{i=1}^n m(W_i,\theta)1\{z_i=\ell\}$. 
As we show in Appendix \ref{app:iidsub}, sufficient conditions for the consistency of this estimator involve boundedness of the fourth moment of $m(W_i,\theta)$ and the assumption that every $z_i$ value occur twice or more in the sample $\{z_i\}_{i=1}^n$ eventually. 

When $Z_i$ contains continuous random variables, one can use a nearest neighbor matching estimator similar to that used for the standard error of a regression discontinuity estimator in Abadie, Imbens, and Zheng (2014).\footnote{This is also the estimator used in \cite{AndrewsRothPakes2019}.} 
Let $\Sigma_{Z,n} = n^{-1}\sum_{i=1}^n(Z_i-\overline{Z}_n)(Z_i-\overline{Z}_n)'$ where $\overline{Z}_n = n^{-1}\sum_{i=1}^n Z_i$. 
For each $i$, define the nearest neighbor to be
\begin{align}
\ell_Z(i) = \text{argmin}_{j\in\{1,\dots,n\},j\neq i} (Z_i-Z_j)'\Sigma_{Z,n}^{-1}(Z_i-Z_j).
\end{align}
When the argmin is not unique, picking one randomly does not affect the consistency of the resulting estimator. 
The estimator of $\Sigma_n(\theta)$ is then given by
\begin{align}
\widehat{\Sigma}_n(\theta) = \frac{1}{2n}\sum_{i=1}^n(m(W_i,\theta) -m(W_{\ell_Z(i)},\theta))(m(W_i,\theta) - m(W_{\ell_Z(i)},\theta))'.\label{condVar_nn}
\end{align}
As we show in Appendix \ref{app:iidsub}, sufficient conditions for the consistency of this matching estimator involves the boundedness of $\{z_i\}_{i=1}^\infty$ and the Lipschitz continuity of $\textup{Var}(m(W_i,\theta)|z_i)$ in $z_i$. \qed

\section{Monte Carlo Simulations}
We consider two sets of Monte Carlo simulations, one to evaluate the performance of our tests in a general moment inequality model without nuisance parameters, and the second to evaluate the performance of our subvector tests in the interval regression model of \cite{GandhiLuShi2019}. 

\subsection{Full-Vector Simulations}
Our first set of simulations takes the generic moment inequality design from AB. 
This design allows a variety of correlation structures across moments and thus can approximate a wide range of applications.  

We briefly describe the Monte Carlo design here and refer the readers to Section 6 of AB for further details. 
Consider the moment inequality model 
\begin{equation}
E[W_i - \theta]\geq \mathbf{0},
\end{equation}
and the null hypothesis $H_0:\theta=\mathbf{0}$, where $W_i$ is a $p$-dimensional random vector. 
Let the data $\{W_i\}_{i=1}^n$ be i.i.d. with sample size $n$. 
Let $W_i \sim N(\mu,\Omega)$, where $\Omega$ is a correlation matrix and $\mu$ is a mean-vector. 
Three choices of $\Omega$ are considered: $\Omega_{\text{Neg}}$, $\Omega_{\text{Zero}}$, and $\Omega_{\text{Pos}}$, indicating negative, zero, and positive correlation among the moments, respectively. 
The exact numerical specifications of these matrices for different $p$'s are in Section 4 of AB and Section S7.1 of the Supplemental Material of AB. 
Also, three choices of $p$: $2$, $4$, and $10$, are considered. 

For each combination of $\Omega$ and $p$, we approximate the size of the tests using the maximum null rejection probability (MNRP) over a set of $\mu$ values that satisfies $\mu\geq \mathbf{0}$. 
These $\mu$ values are taken from AB whose calculations suggest that these points are capable of approximating the size of the tests. 
We compute a weighted average power (WAP) for easy comparison. 
The WAP is the simple average of a set of carefully chosen points in the alternative space. 
We take these points also from AB who design them to reflect cases with various degrees of violation or slackness for each of the inequalities. 
These $\mu$ values are given in Section 4 of AB and Section S7.1 of the Supplemental Material of AB. 
Besides WAP, we also report size-corrected WAP, which is obtained by adding a (positive or negative) number to the critical value where the number is set to make the size-corrected MNRP equal to the nominal level. 

We report Monte Carlo simulation results to compare the CC and the RCC tests to the recommended tests in AB and RSW, more specifically, the bootstrap-based AQLR (adjusted quasi-likelihood ratio) test in AB and the two-step procedure in RSW.\footnote{We use the AB test for comparison because it is tuning parameter free like ours (in the sense that AB propose and use an optimal choice of the AS tuning parameter), and we use RSW's two-step test for comparison because it should be insensitive to reasonable choices of its tuning parameter.}

\begin{table}[p]
\begin{threeparttable}
{\scriptsize
\caption{Finite Sample Maximum Null Rejection Probabilities and Size-Corrected Average Power of Nominal 5\% Tests (known $\Omega$, $n=100$)}\label{tab:knownO}
\begin{center}
\begin{tabular}{ccccccccccccccc}
\hline\hline
&\multicolumn{4}{c}{$k=10$}&&\multicolumn{4}{c}{$k=4$}&&\multicolumn{4}{c}{$k=2$}\\
\cline{2-5}\cline{7-10}\cline{12-15}\vspace{-0.1cm}\\
{Test}&{MNRP}&{WAP}&{ScWAP}&{ReT}&& {MNRP}&{WAP}&{ScWAP}&{ReT}&&{MNRP}&{WAP}&{ScWAP}&{ReT}\\
\hline\\
\multicolumn{15}{c}{$\Omega = \Omega_{\text{Neg}}$}\\
\hline\\
RCC& $.051$&.61&.61&1&&.051&.62&.62&1&&.049&.62&.62&1\\
CC&$.050$&.61&.61&1&&.050&.60&.61&1&&.045&.58&.60&1\\
AB&.044&.53&.55&368&&.054&.59&.58&369&&.062&.65&.61&367\\
RSW&.094&.61&.47&217&&.068&.61&.57&204&&.054&.63&.62&220\\
\hline\\
\multicolumn{15}{c}{$\Omega = \Omega_{\text{Zero}}$}\\
\hline\\
RCC&$.051$&.64&.63&1&&.051&.66&.66&1&&.049&.69&.69&1\\
CC&$.050$&.62&.62&1&&.046&.62&.63&1&&.037&.61&.66&1\\
AB&.049&.65&.66&367&&.054&.68&.67&375&&.063&.69&.66&367\\
RSW&.044&.60&.60&216&&.047&.63&.64&208&&.051&.65&.65&222\\
\hline\\
\multicolumn{15}{c}{$\Omega = \Omega_{\text{Pos}}$}\\
\hline\\
RCC&$.051$&.76&.76&1&&.050&.75&.75&1&&.049&.72&.72&1\\
CC&$.040$&.72&.75&1&&.034&.68&.73&1&&.033&.62&.69&1\\
AB&.042&.79&.81&351&&.052&.77&.76&375&&.066&.73&.70&369\\
RSW&.044&.76&.77&213&&.047&.72&.74&214&&.048&.67&.67&231\\
\hline
\end{tabular}
\begin{tablenotes}
\item {\em Note:} CC denotes the conditional chi-squared test, RCC denotes the refined CC test, AB denotes the adjusted quasi-likelihood ratio (AQLR) test with bootstrap critical value in AB and RSW denotes the two-step test in RSW. 
MNRP denotes the maximum null rejection probabilities, WAP denotes the weighted average power, ScWAP denotes the size-corrected WAP, and ReT denotes computation time relative to CC for each Monte Carlo repetition. 
The AB test uses 1000 critical value simulations and the RSW test uses 499 critical value simulations. 
The results for the CC and RCC tests are based on $10^5$ simulations, while the results for the AB and RSW tests are based on 2000 simulations for computational reasons. 
\end{tablenotes}
\end{center}
}
\end{threeparttable}
\end{table}

\begin{table}[p]
\begin{threeparttable}
{\scriptsize
\caption{Finite Sample Maximum Null Rejection Probabilities and Size-Corrected Average Power of Nominal 5\% Tests (Estimated $\Omega$, $n=100$)}\label{tab:unknownO}
\begin{center}
\begin{tabular}{ccccccccccccccc}
\hline\hline
&\multicolumn{4}{c}{$k=10$}&&\multicolumn{4}{c}{$k=4$}&&\multicolumn{4}{c}{$k=2$}\\
\cline{2-5}\cline{7-10}\cline{12-15}\vspace{-0.1cm}\\
{Test }&{MNRP}&{WAP}&{ScWAP}&{ReT}&& {MNRP}&{WAP}&{ScWAP}&{ReT}&&{MNRP}&{WAP}&{ScWAP}&{ReT}\\
\hline\\
\multicolumn{15}{c}{$\Omega = \Omega_{\text{Neg}}$}\\
\hline\\
RCC& $.076$&.64&.53&1&&.059&.63&.60&1&&.053&.62&.61&1\\
CC&$.076$&.63&.53&1&&.059&.61&.58&1&&.049&.59&.59&1\\
AB&.048&.51&.53&400&&.055&.58&.56&398&&.057&.64&.62&375\\
RSW&.104&.60&.42&210&&.066&.59&.54&202&&.054&.62&.60&208\\
\hline\\
\multicolumn{15}{c}{$\Omega = \Omega_{\text{Zero}}$}\\
\hline\\
RCC&$.069$&.66&.59&1&&.056&.66&.64&1&&.053&.69&.68&1\\
CC&$.069$&.64&.57&1&&.053&.66&.64&1&&.040&.61&.65&1\\
AB&.046&.62&.64&405&&.052&.66&.66&382&&.060&.68&.65&371\\
RSW&.043&.57&.59&213&&.044&.61&.63&195&&.046&.63&.65&209\\
\hline\\
\multicolumn{15}{c}{$\Omega = \Omega_{\text{Pos}}$}\\
\hline\\
RCC&$.055$&.77&.75&1&&.053&.75&.74&1&&.052&.72&.71&1\\
CC&$.046$&.73&.74&1&&.038&.68&.73&1&&.035&.63&.69&1\\
AB&.039&.78&.81&412&&.049&.76&.76&388&&.059&.72&.69&379\\
RSW&.043&.75&.77&222&&.045&.71&.73&205&&.046&.65&.68&221\\
\hline
\end{tabular}
\begin{tablenotes}
\item {\em Note:} Same as Table \ref{tab:knownO}.
\end{tablenotes}
\end{center}
}
\end{threeparttable}
\end{table}

Two sets of results are reported. 
The first set assumes a known $\Omega$ and is reported in Table \ref{tab:knownO}. 
In this case, the RCC test should have exact size and the CC test should be somewhat under-sized especially with small $k$. 
These theoretical predictions are exactly confirmed in the table. 
The second set of results does not assume a known $\Omega$ and is reported in Table \ref{tab:unknownO}. 
In this case, the RCC test still has very good MNRP at $k=2$, but has some noticeable over-rejection when $k=10$ with $\Omega = \Omega_{\text{Neg}}$ and $\Omega_{\text{Zero}}$. 
This may reflect the difficulty in estimating the large dimensional $\Omega$ with a small sample size ($n=100$). 
In comparison, the AB and the RSW tests exhibit under-rejection in some cases and over-rejection in other cases, even with known $\Omega$. 
This could be the result of the relatively small number of critical value simulations (1000 for AB and 499 for RSW, as recommended therein). 
Even with the small number of critical value simulations, the AB test and the RSW test are 200-400 times more costly than the RCC test, as shown in the ReT column in the tables.  
Increasing this number will increase their computational cost proportionally.

In terms of power, we find that the RCC test compares favorably to the AB and RSW tests. 
The biggest advantage of the RCC test seems to come from the points with a small number of violated inequalities and a large number of mildly slack inequalities while the magnitude of the advantage varies with $\Omega$.\footnote{For example, with $k=10$, $\Omega=\Omega_{\text{Zero}}$ and estimated $\Omega$, at the point $\mu = (-0.268,0.1, 0.1, 0.1,	0.1, \allowbreak0.1,0.1,	0.1,0.1,	0.1)'$, the size-corrected power of the RCC test is 0.55 while that for AB and RSW are respectively 0.28 and 0.22.} 
By the uni-dimensional criterion ScWAP, the RCC test is better than or the same as the RSW test in all cases but the case with $k=10, \Omega=\Omega_{\text{Pos}}$ with unknown $\Omega$ where the RSW test under-rejects and the size-correction adjusts its power up.  
Also according to ScWAP, the RCC test has better or equal power as AB in half of the cases, while the difference in all cases is small. 

Overall, when we consider not only size and power from the simulations, but also computational cost, simplicity, simulation error (or lack thereof) of the critical values, and finite-sample properties, we recommend the RCC test over the other tests.

\subsection{Subvector Inference in Interval Regression}\label{sub:mc-subv}
Now consider a special case of Example \ref{ex:intreg} above. 
More specifically, we consider a model where $Y_i^\ast=s^\ast_i $ is the probability of an event of interest, for example, death by homicide for a random person county $i$, or a product being purchased by a random consumer in market $i$. 
For simplicity, we consider a simple logit model for the probability: $s_i^\ast =  \frac{\exp(X'_i\theta_0+Z'_{ci}\delta_0+\varepsilon_i)}{1+\exp(X'_i\theta_0+Z'_{ci}\delta_0+\varepsilon_i)}$, where $\varepsilon_i$ is the country or market level unobservable that satisfies $\E[\varepsilon_i|Z_i]=0$. 
Then (\ref{momentcondition}) holds with
\begin{align}
\psi(Y^\ast_i,X_i,\theta) = \log(s_i^\ast/(1-s^\ast_i)) - X'_i\theta.
\end{align}

The variable $s^\ast_i$ is unobserved, but we observe $s_{N,i}$, an empirical estimate of $s^\ast_i$ based on $N$ independent chances for the event of interest to happen: $Ns_{N,i}$ follows a binomial distribution with parameters $(N,s^\ast_i)$. 
For example, $N$ could be the population of the county and $s_{N,i}$ is the homicide rate of the county. 
We use the method introduced in \cite{GandhiLuShi2019} to construct $\psi_{i}^L(\theta)$ and $\psi_{i}^U(\theta)$ based on $s_{N,i}$. 
By \cite{GandhiLuShi2019}, for $N\geq 100$, the following construction satisfies (\ref{bounds}):
\begin{align}
\psi_i^U(\theta)  = \log(s_{N,i}+2/N) - \log(1-s_{N,i}+\underline{s})- X_i'\theta\\
\psi_i^L(\theta)= \log(s_{N,i}+\underline{s})-\log(1-s_{N,i}+2/N)-X_i'\theta,\nonumber
\end{align}
where $\underline{s}$ is the smaller of $0.05$ and half of the minimum possible value of $\min(s^\ast_i,1-s^\ast_i)$.\footnote{ These bounds are not necessarily sharp, but that is not important for our purpose, which is to investigate the statistical performance of the sCC and sRCC tests.} 
We assume that $\underline{s}$ is known  and refer the reader to \cite{GandhiLuShi2019} for practical recommendations regarding $\underline{s}$.\footnote{It is worth pointing out that we deviate from \cite{GandhiLuShi2019} by holding $N$ fixed as the sample size (i.e. the number of observations for $(s_{N,i},X_i,Z_{ci})$) increases, and thus do not aim for point identification.}

For this simulation, we consider a scalar endogenous variable $X_i$, a scalar excluded instrument $Z_{ei}$ and a $d_c$-dimensional exogenous covariate $Z_{ci} = (1,Z_{c,2,i},\dots,Z_{c, d_c,i})'$. 
To generate the data, we let $N=100$, $\varepsilon_i\sim \min\{ \max\{-4,N(0,1)\},4\}$. 
We let the non-constant elements of $Z_i$ be mutually independent Bernoulli variables with success probability $0.5$. 
Let $X_i = 1\{Z_{ei}+\varepsilon_i/2>0\}$. 
Also consider $\theta_0=-1$, $\delta_0=(0,-1, \mathbf{0}_{d_c-2}')'$. 
Given this data generating process, the lowest and the highest possible values for $s^\ast_i$ are respectively
$$
\frac{\exp(-6)}{1+\exp(-6)} = 0.0025\text{ and }\frac{\exp(4)}{1+\exp(4)}=0.982.
$$
Thus $\underline{s} = 0.00125$. 

We calculate that the identified set of $\theta_0$ is approximately $[-1.203,-0.757]$. 
Details of the calculation are given in Appendix \ref{sub:IDset}.

We simulate the rejection rate of the tests using 5000 Monte Carlo repetitions. 
In each repetition, we generate an i.i.d. data set $\{s_{N,i},X_i, Z_i\}_{i=1}^n$, for two sample sizes $n=500$ and $n=1000$. 
We consider three cases of $d_c$: $d_c=2,3$ and $4$. 

For instrumental functions, we use $I(Z_i) = (I_{j,\ell}(Z_i))_{j, \ell  = 2,\dots,d_c+1:j\neq \ell}$, where
\begin{align}
I_{j,\ell}(Z_i) = \begin{pmatrix}1(Z_{ji} = 1,Z_{\ell i}=1)\\1(Z_{ji} = 1,Z_{\ell i}=0)\\1(Z_{ji} = 0,Z_{\ell i}=1)\\1(Z_{ji} = 0,Z_{\ell i}=0)\end{pmatrix},
\end{align}
where $Z_{ji} $ is the $j$th element of $Z_i$. 
Thus, when $d_c=2$ (or $3$, $4$), there are $4$ (or $12$, $48$) instrumental functions, which give us $8$ (or $24$, $96$) moment inequalities. 
 
Figure \ref{Fig:subcc} reports the rejection rates of the  sCC test and the sRCC test for $H_0:\theta_0=\theta$ at $\theta$ values in $[-2.5,0.5]$ in the three cases of $d_c$ and two sample sizes. 
The shaded area indicates the identified set for $\theta_0$. 
As we can see, both tests reject $H_0$ at rates lower than $5\%$ for $\theta$'s in the identified set. 
The rejection rates are closer to 5\% on the boundary of the identified set than in the interior and the sRCC test is better than the sCC test in all cases. 
The rejection rate grows steadily toward 1 as $\theta$ moves away from the identified set, and it grows with the sample size as well, as expected. 

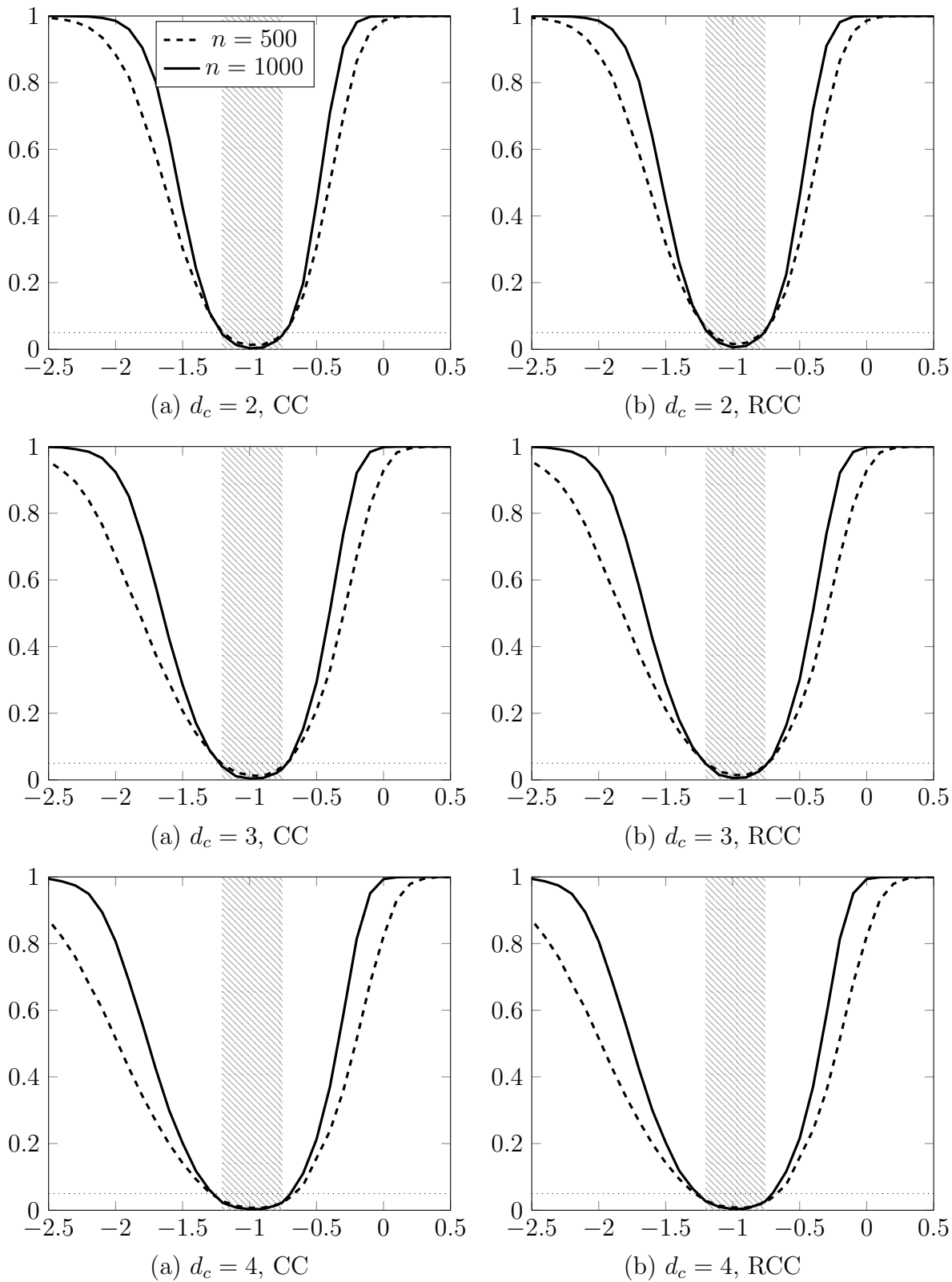
\begin{figure}
\caption{Rejection Rates of the sCC and sRCC Tests for $H_0:\theta=\theta_0$ at a Variety of $\theta$ Values, in 3 Cases of $d_c$ and for 2 Sample Sizes (nominal level = $5\%$)}
\pgfplotsset{width=0.5*\textwidth}
\begin{center}
\begin{tikzpicture}
\begin{axis}
[ymin=0,ymax=1,xmin=-2.5,xmax=0.5,legend style={at={(axis cs:(-1.7,0.985)},anchor=north west}]
\draw [pattern=north west lines, pattern color=gray!70,draw=none] (-1.203,0) rectangle (-0.757,0.996);
\draw[dotted] (-2.5,0.05) -- (0.5,0.05);
\addplot[very thick,dashed] table [x=beta_0, y=ccRej, col sep=semicolon] {IntReg_RejRate_RCCcorrected_S_5000_dx_1_n_500.txt};
\addplot[very thick] table [x=beta_0, y=ccRej, col sep=semicolon] {IntReg_RejRate_RCCcorrected_S_5000_dx_1_n_1000.txt};
\addlegendentry{$n=500$};
\addlegendentry{$n=1000$};
\end{axis}

\node[align=center, below] at (3,-.6) {(a) $d_c=2$, CC};
\end{tikzpicture}
\begin{tikzpicture}
\begin{axis}
[ymin=0,ymax=1,xmin=-2.5,xmax=0.5]
\draw [pattern=north west lines, pattern color=gray!70,draw=none] (-1.203,0) rectangle (-0.757,0.996);
\draw[dotted] (-2.5,0.05) -- (0.5,0.05);
\addplot[very thick,dashed] table [x=beta_0, y=accRej, col sep=semicolon] {IntReg_RejRate_RCCcorrected_S_5000_dx_1_n_500.txt};
\addplot[very thick] table [x=beta_0, y=accRej, col sep=semicolon] {IntReg_RejRate_RCCcorrected_S_5000_dx_1_n_1000.txt};
\end{axis}
\node[align=center, below] at (3,-.6) {(b) $d_c=2$, RCC};
\end{tikzpicture}

\begin{tikzpicture}
\begin{axis}
[ymin=0,ymax=1,xmin=-2.5,xmax=0.5,legend style={at={(axis cs:(-1.7,0.985)},anchor=north west}]
\draw [pattern=north west lines, pattern color=gray!70,draw=none] (-1.203,0) rectangle (-0.757,0.996);
\draw[dotted] (-2.5,0.05) -- (0.5,0.05);
\addplot[very thick,dashed] table [x=beta_0, y=ccRej, col sep=semicolon] {IntReg_RejRate_RCCcorrected_S_5000_dx_2_n_500.txt};
\addplot[very thick] table [x=beta_0, y=ccRej, col sep=semicolon] {IntReg_RejRate_RCCcorrected_S_5000_dx_2_n_1000.txt};
\end{axis}

\node[align=center, below] at (3,-.6) {(a) $d_c=3$, CC};
\end{tikzpicture}
\begin{tikzpicture}
\begin{axis}
[ymin=0,ymax=1,xmin=-2.5,xmax=0.5]
\draw [pattern=north west lines, pattern color=gray!70,draw=none] (-1.203,0) rectangle (-0.757,0.996);
\draw[dotted] (-2.5,0.05) -- (0.5,0.05);
\addplot[very thick,dashed] table [x=beta_0, y=accRej, col sep=semicolon] {IntReg_RejRate_RCCcorrected_S_5000_dx_2_n_500.txt};
\addplot[very thick] table [x=beta_0, y=accRej, col sep=semicolon] {IntReg_RejRate_RCCcorrected_S_5000_dx_2_n_1000.txt};
\end{axis}
\node[align=center, below] at (3,-.6) {(b) $d_c=3$, RCC};
\end{tikzpicture}

\begin{tikzpicture}
\begin{axis}
[ymin=0,ymax=1,xmin=-2.5,xmax=0.5]
\draw [pattern=north west lines, pattern color=gray!70,draw=none] (-1.203,0) rectangle (-0.757,0.996);
\draw[dotted] (-2.5,0.05) -- (0.5,0.05);
\addplot[very thick,dashed] table [x=beta_0, y=ccRej, col sep=semicolon] {IntReg_RejRate_RCCcorrected_S_5000_dx_3_n_500.txt};
\addplot[very thick] table [x=beta_0, y=ccRej, col sep=semicolon] {IntReg_RejRate_RCCcorrected_S_5000_dx_3_n_1000.txt};
\end{axis}
\node[align=center, below] at (3,-.6) {(a) $d_c=4$, CC};
\end{tikzpicture}
\begin{tikzpicture}
\begin{axis}
[ymin=0,ymax=1,xmin=-2.5,xmax=0.5]
\draw [pattern=north west lines, pattern color=gray!70,draw=none] (-1.203,0) rectangle (-0.757,0.996);
\draw[dotted] (-2.5,0.05) -- (0.5,0.05);
\addplot[very thick,dashed] table [x=beta_0, y=accRej, col sep=semicolon] {IntReg_RejRate_RCCcorrected_S_5000_dx_3_n_500.txt};
\addplot[very thick] table [x=beta_0, y=accRej, col sep=semicolon] {IntReg_RejRate_RCCcorrected_S_5000_dx_3_n_1000.txt};
\end{axis}
\node[align=center, below] at (3,-.6) {(b) $d_c=4$, RCC};
\end{tikzpicture}
\end{center}
\label{Fig:subcc}
\end{figure}

On the computational side, fixing the sample size at $n=1000$ and computing each test 5000 times, we document the time needed to compute the rejection rate of each test at one $\theta$ value on a computer with Intel X5680 3.33HZ CPU and 12Gb RAM running Matlab 2018. 
We find that on average, the subvector CC test uses about 13 minutes, 29 minutes, and 43 minutes respectively for $d_c=2$ (8 inequalities and 2 nuisance parameters), $d_c=3$ (24 inequalities and 3 nuisance parameters), and $d_c=4$ (96 inequalities and 4 nuisance parameters). 
The subvector  sRCC test uses a similar amount of time as the sCC test when $d_c=2$ and $3$. 
It uses noticeably more time (2.2 hours) when $d_c=4$, which recall is the time to perform the test 5000 times when there are 96 conditional inequalities to begin with and 4 nuisance parameters to eliminate. 
Thus, for the scale of the job, the computational cost is quite modest.

The advantage of avoiding the computation of $H(C_Z)$ is very large at larger $d_c$. 
If we do not use the procedure described in Section \ref{sec:sRCC}, but instead compute $H(C_Z)$ every time we perform the sRCC test, then the computation time for $d_c=3$ becomes 14 hours, and that for $d_c=4$ becomes a whopping 300 hours.\footnote{The matrix $H(C_Z)$ is found to have around 19, 1600, and 127550 rows, respectively for $d_c=2$, $3$, and $4$.} 
Therefore, we recommend our procedure over performing subvector tests via brute force elimination of the nuisance parameters. 

For comparison, we also compute the rejection rates of two projection-based tests: the projection RCC test with unconditional variance (Proj-U) and the projection RCC test with conditional variance (Proj-C). 
Both tests use the full-vector RCC test. 
That is, the RCC test for the null hypothesis $H_0:(\theta',\delta')' = (\theta'_0,\delta'_0)'$) for the model $\E_{F_z}[\overline{m}_n(\theta) -C_z\delta]\leq 0$.\footnote{In this example, $B_z = I$ and $d_z=\mathbf{0}$.} 
The difference is that the Proj-C test uses the average conditional variance estimator $\widehat{\Sigma}_n(\theta)$ in (\ref{condVar_nn}), while the Proj-U test uses the unconditional variance estimator, $\widehat{\Sigma}_n^{\textup{U}}(\theta,\delta) =$
\begin{equation}
 n^{-1}\sum_{i=1}^n(m(W_i,\theta)-C_{z_i}\delta - \overline{m}_n(\theta)+C_Z\delta)(m(W_i,\theta)-C_{z_i}\delta - \overline{m}_n(\theta)+C_Z\delta)', 
\end{equation}
where $C_{z_i}$ denotes the $C_z$ matrix evaluated with only the $i$th observation for $z$. 
Suppose that the test statistic and the critical value for the full-vector RCC test using $\widehat{\Sigma}_n(\theta)$ are denoted $T_n^{\textup{C}}(\theta,\delta)$ and $\textup{cv}_n^{\textup{C}}(\theta,\delta,1-\alpha)$, and those using $\widehat{\Sigma}_n^{\textup{U}}(\theta,\delta)$ are denoted $T_n^{\textup{U}}(\theta,\delta)$ and ${\textup{cv}}_n^{\textup{U}}(\theta,\delta,1-\alpha)$. 
Then the Proj-C and the Proj-U tests are, respectively, 
\begin{align}
\phi_n^{\textup{Proj-C}}(\theta,\alpha)& =1\{\inf_{\delta}[T_n^{\textup{C}}(\theta,\delta)-\textup{cv}_n^{\textup{C}}(\theta,\delta,1-\alpha)] >0\},\textup{ and }\\
\phi_n^{\textup{Proj-U}}(\theta,\alpha) &= 1\{\inf_{\delta}[T_n^{\textup{U}}(\theta,\delta)-\textup{cv}_n^{\textup{U}}(\theta,\delta,1-\alpha)] >0\}.\nonumber
\end{align}

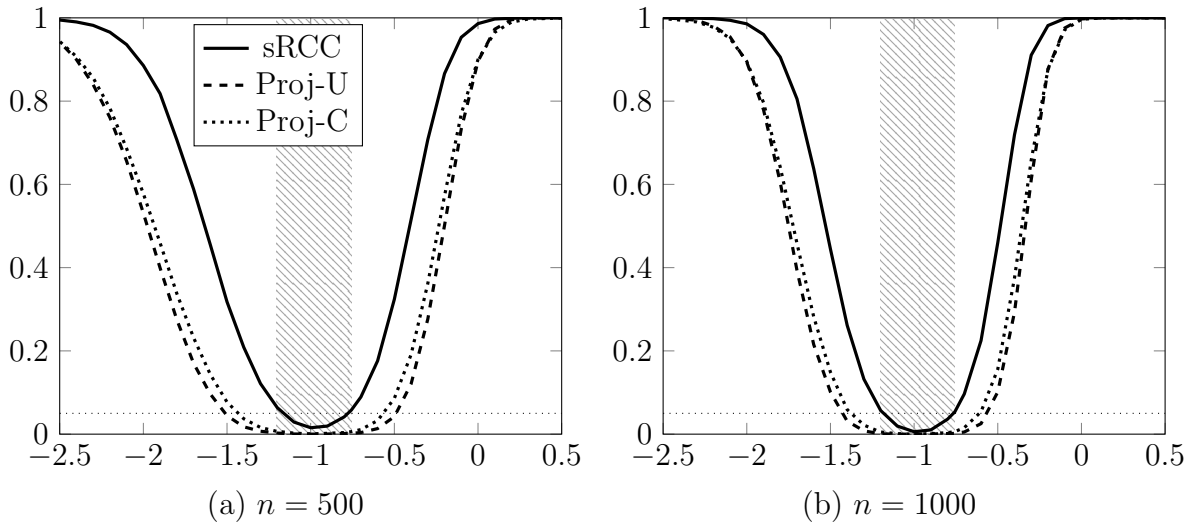
\begin{figure}
\caption{Rejection Rates of the sRCC Test versus the Proj-U and the Proj-C tests ($d_c=2$, nominal level=5\%).}
\pgfplotsset{width=0.5*\textwidth}
\begin{center}
\begin{tikzpicture}
\begin{axis}
[ymin=0,ymax=1,xmin=-2.5,xmax=0.5,legend style={at={(axis cs:(-1.7,0.985)},anchor=north west}]
\draw [pattern=north west lines, pattern color=gray!70,draw=none] (-1.203,0) rectangle (-0.757,0.996);
\draw[dotted] (-2.5,0.05) -- (0.5,0.05);
\addplot[very thick] table [x=beta_0, y=accRej, col sep=semicolon] {IntReg_RejRate_RCCcorrected_S_5000_dx_1_n_500.txt};
\addplot[very thick,dashed] table [x=beta_0, y=rccRej, col sep=semicolon] {IntReg_RejRate_RCCproj_S_5000_maxStart_5_dx_1_n_500.txt};
\addplot[very thick,dotted] table [x=beta_0, y=rccRej, col sep=semicolon] {IntReg_RejRate_condRCCproj_S_5000_maxStart_5_dx_1_n_500.txt};
\addlegendentry{sRCC};
\addlegendentry{Proj-U};
\addlegendentry{Proj-C};
\end{axis}
\node[align=center, below] at (3,-.6) {(a) $n=500$};
\end{tikzpicture}
\begin{tikzpicture}
\begin{axis}
[ymin=0,ymax=1,xmin=-2.5,xmax=0.5,legend style={at={(axis cs:(-1.7,0.985)},anchor=north west}]
\draw [pattern=north west lines, pattern color=gray!70,draw=none] (-1.203,0) rectangle (-0.757,0.996);
\draw[dotted] (-2.5,0.05) -- (0.5,0.05);
\addplot[very thick] table [x=beta_0, y=accRej, col sep=semicolon] {IntReg_RejRate_RCCcorrected_S_5000_dx_1_n_1000.txt};
\addplot[very thick,dashed] table [x=beta_0, y=rccRej, col sep=semicolon] {IntReg_RejRate_RCCproj_S_5000_maxStart_5_dx_1_n_1000.txt};
\addplot[very thick,dotted] table [x=beta_0, y=rccRej, col sep=semicolon] {IntReg_RejRate_condRCCproj_S_5000_maxStart_5_dx_1_n_1000.txt};
\end{axis}
\node[align=center, below] at (3,-.6) {(b) $n=1000$};
\end{tikzpicture}
\end{center}
\label{Fig:srccproj}
\end{figure}

Computing the projection-based tests is tricky even though we already use the computationally simple RCC test for $\textup{cv}$. 
This is because the infimum over $\delta$ is taken over a non-convex function, and finding global minimum over a non-convex function is challenging. 
In section \ref{sub:IDset} in the appendix, we detail the numerical algorithm used to calculate the global minimum. 
The algorithm works well, but is not guaranteed to find the global minimum. 
Not finding the global minimum biases the power upward, and thus what we find is an {\em upper} bound for the power of the projection-based tests.

The results are plotted in Figure \ref{Fig:srccproj} for $d_c=2$ and sample sizes $n=500$ and $1000$. 
As we can see, the Proj-C test and the Proj-U test perform similarly and both are uniformly dominated by the sRCC test. 
Computationally, the projection tests (for one $\theta$ value and 5000 Monte Carlo simulations) each took more than 1 hour, or more than 4 times the 13 minutes needed for the sRCC test. 
Therefore, on the basis of both power and computational cost, we recommend the sRCC test over the projection-based tests.

\section{Conclusion}

This paper proposes the refined conditional chi-squared (RCC) test for moment inequality models. 
This test compares a quasi-likelihood ratio statistic to a chi-squared critical value, where the number of degrees of freedom is the rank of the active inequalities. 
This test has many desirable properties, including being simple, adaptive, and tuning parameter and simulation free. 
We show that, with an easy refinement, it has exact size in normal models and has uniformly asymptotically exact size more generally. 
We also propose a version of the test for subvector inference with conditional moment inequalities and when the nuisance parameters enter linearly that has computational and power advantages. 

\newpage
\appendix
\section{Proof of Theorem \ref{thm:normal-rcc}}\label{sec:proof-normal}

For this proof, we assume $\Sigma_n(\theta)=nI_{d_m}$. 
If this is not the case, then the following proof can be applied after premultiplying $\overline{m}_n(\theta)$ by $n^{1/2}\Sigma_n(\theta)^{-1/2}$ and postmultiplying $A$ by $n^{-1/2}\Sigma_n(\theta)^{1/2}$. 

Fix $\theta$ and let $X=\overline m_n(\theta)\sim N(\mu,I_{d_m})$, where $\mu=\E_F \overline m_n(\theta)$. 
Let $C=\{\mu\in\R^{d_m}| A\mu\le b\}$. 
The fact that $\theta\in\Theta_0(F)$ implies $\mu\in C$. 
These simplifications imply that $T_n(\theta)=\|X-\hat\mu\|^2$ and 
\begin{equation}
\hat\tau_j = \begin{cases}
\frac{\|a_1\|(b_j-a'_j \hat\mu)}{\|a_1\|\|a_j\|-a'_1 a_j}&\text{ if }\|a_1\|\|a_j\|\neq a'_1a_j\\
\infty&\text{ otherwise }
\end{cases}.
\end{equation}
The definitions of $\hat\mu$, $\widehat J$, $\hat r$, $\hat\tau$, and $\hat\beta$ are unchanged. 
$\hat\mu$ is the projection of $X$ onto $C$. 
We also denote it by $P_C X$. 
We also denote $\widehat J$ by $J(X)$, $\hat r$ by $r(X)$, $\hat\tau$ by $\tau(X)$, and $\hat\beta$ by $\beta(X)$. 

\subsection{Auxiliary Lemmas}\label{sub:auxlem}
The proof of Theorem \ref{thm:normal-rcc} relies on four lemmas. 

The first lemma partitions $\R^{d_m}$ according to which inequalities are active. 
We define some notation for the partition. 
For any $J\subseteq\{1,...,d_A\}$, let $J^c=\{1,...,d_A\}/ J$, and let $C_J=\{x\in C: \forall j\in J, a'_j x=b_j, \text{ and }\forall j\in J^c, a'_jx<b_j\}$. 
Then $C_J$ forms a partition of $C$. 
Also let $V_J=\{\sum_{j\in J}v_ja_j: v_j\in \R, v_j\ge 0\}$, and let $K_J=C_J+V_J$.\footnote{When $J=\emptyset$, then $V_J=\{\bm{0}_{d_X}\}$.} 
The following lemma shows that $K_J$ forms a partition that characterizes which inequalities are active. 
\begin{lemma}\label{lem:partition}
\begin{enumerate}[label=(\alph*)]
\item[\textup{(a)}] If $X\in K_J$, then $X-P_C X\in V_J$ and $P_C X\in C_J$. 
\item[\textup{(b)}] The set of all $K_J$ for $J\subseteq \{1,...,d_A\}$ is a partition of $\R^{d_m}$. 
\item[\textup{(c)}] For every $J\subseteq \{1,...,d_A\}$, $X\in K_J$ iff $J=J(X)$. 
\end{enumerate}
\end{lemma}

The next lemma considers the event $\hat r=0$, and partitions that event according to which face of $C$ is closest to the realization of $X$. 
Let $J_0=\{j\in\{1,...,d_A\}| a_j=0\}$ and let $J_{00} = \{j\in\{1,...,d_A\}| a_j = 0 \text{ and }b_j = 0\}$. 
Also let 
\begin{align}
\mathcal{J}_1 =& \{J\subseteq \{1,...,d_A\}| \textup{rk}(A_J)=1, J\cap J_0=J_{00}, \\
&\text{ and if }j\in J, \ell\in J^c, \text{ s.t. }\|a_j\|>0, \|a_\ell\|>0, \nonumber\\
&\text{ then }\frac{a_j}{\|a_j\|}\neq \frac{a_\ell}{\|a_\ell\|} \text{ or }\frac{b_j}{\|a_j\|}\neq \frac{b_\ell}{\|a_\ell\|}\}. \nonumber
\end{align}
Further subdivide
\begin{align}
\mathcal{J}_1^{os}=&\{J\in \mathcal{J}_1| \text{ if } j,\ell\in J \text{ s.t. } \|a_j\|>0, \|a_\ell\|>0, \text{ then }\frac{a_j}{\|a_j\|}=\frac{a_\ell}{\|a_\ell\|}\}\\
\mathcal{J}_1^{ts}=&\{J\in\mathcal{J}_1| \exists j,\ell\in J \text{ s.t. } \|a_j\|>0, \|a_\ell\|>0, \frac{a_j}{\|a_j\|}=-\frac{a_\ell}{\|a_\ell\|} \text{ and } \frac{b_j}{\|a_j\|}=-\frac{b_\ell}{\|a_\ell\|}\}.\nonumber
\end{align}

The next lemma provides a partition of $C_0:=\cup_{J\subseteq\{1,...,d_A\}:\textup{rk}(A_J)=0}C_J$. 
(Note that for these sets, $C_J=K_J$.) 
Let $J_{\neq 0}=\{j=1,\dots,d_A: \|a_j\|\neq 0\}$. 
For each $J\in\mathcal{J}_1^{os}$, let 
\begin{equation}
C^\Delta_J=\{x\in C_0 | \text{argmin}_{j\in J_{\neq 0}} \|a_j\|\inv (b_j-a'_j x)=J\cap J_{\neq 0}\}. 
\end{equation}
The set $C^\Delta_J$ is the set of points in $C$ that are closer to $C_J$ than to any other $C_{\tilde J}$ for $\tilde J\in\mathcal{J}_1$. 
It is helpful to picture $C_J$ for $J\in\mathcal{J}_1$ as the faces of a polyhedron, $C$, and $C^\Delta_J$ as a partition of $C$ into triangularly shaped sets. 
Also let 
\begin{equation}
C^| = C_0 / \left(\cup_{J\in \mathcal{J}_1^{os}} C^\Delta_J\right). 
\end{equation}

\begin{lemma}\label{lem:partition2}
\begin{enumerate}[label=(\alph*)]
\item[\textup{(a)}] $C_0=C_{J_{00}}$. 
\item[\textup{(b)}] The sets $C^|$ and $C^\Delta_J$ for $J\in\mathcal{J}_1^{os}$ form a partition of $C_0$. 
\item[\textup{(c)}] If $A\neq 0_{d_A\times d_m}$, then $C^|$ has Lebesgue measure zero. 
\item[\textup{(d)}] $\cup_{J\subseteq\{1,...,d_A\}| \textup{rk}(A_J)=1} K_J = \cup_{J\in \mathcal{J}_1^{os}\cup \mathcal{J}_1^{ts}}K_J$. 
\end{enumerate}
\end{lemma}

The next lemma bounds the probabilities of translations of sets in the multivariate normal distribution. 
Let $V$ denote an arbitrary cone in $\R^{r}$ for a positive integer $r$.\footnote{A cone is a set, $V$, such that for all $v\in V$ and for all $\lambda\ge 0$, $\lambda v\in V$.} 
Let $V^*$ denote the polar cone. 
That is, $V^*=\{\gamma\in\R^r| \langle y,\gamma\rangle\le 0\text{ for all }y\in V\}$. 
For any $\gamma\in V^*$, let $Y\sim N(\gamma,I_r)$. 
The following lemma provides a property of probabilities of cones under a translation. 
\begin{lemma}\label{lem:prob-bnd}
For every $\gamma\in V^*$, $Pr_\gamma(\|Y\|^2>\chi^2_{r, 1-\alpha}|Y\in V)\le\alpha$, with equality if $\gamma=0$. 
\end{lemma}
Lemma \ref{lem:prob-bnd} states that the probability that a random vector, $Y$, belongs to the tail of its distribution, conditional on belonging to the cone, $V$, is less than or equal to $\alpha$, where the tail is any point outside a sphere of radius $\sqrt{\chi^2_{r,1-\alpha}}$. 
The key assumption is that the mean of $Y$ must belong to the polar cone, $V^*$, which translates the distribution away from the cone, $V$. 
When $\gamma=0$, this lemma holds with equality because unconditionally $\|Y\|^2\sim \chi^2_r$, the tail of which has mass exactly $\alpha$, and because $\|Y\|^2$ has exactly the same distribution whether or not we condition on $Y\in V$. 
Lemma \ref{lem:prob-bnd} follows from Lemma 1 in \cite{MohamadGoemanZwet2020}, and thus the proof is omitted. 

The following lemma is the key to validity of the refinement to the CC test. 
It is a bound on translations of sets in the univariate normal model. 
\begin{lemma}\label{lem:prob-bnd2}
For every $\mu\le 0$, for every $\tau\ge 0$, and for every $\alpha\in[0,1/2]$, 
\[
Pr_\mu \left(Z> z_{1-\beta/2}|Z> -\tau\right)\le\alpha,
\]
where $Z\sim N(\mu,1)$ and $\beta=2\alpha \Phi(\tau)$, with equality if $\mu=0$. 
\end{lemma}

\subsection{Proof of Theorem \ref{thm:normal-rcc}}

\noindent\underline{First, we show part (a)}. Notice that 
\begin{align}
&\Pr(||X-P_C X||^2>\chi^2_{\textup{rk}(A_{J(X)}), 1-\beta(X)})\nonumber\\
=&\sum_{J\subseteq \{1,...,d_A\}} \Pr(X\in K_J \text{ and } ||X-P_C X||^2>\chi^2_{\textup{rk}(A_J), 1-\beta(X)})\nonumber\\
=&\sum_{J\subseteq \{1,...,d_A\}| \textup{rk}(A_J)\ge 2} \Pr(X\in K_J \text{ and } ||X-P_C X||^2>\chi^2_{\textup{rk}(A_J), 1-\alpha})\label{PrJ-a2}\\
&+\sum_{J\in\mathcal{J}_1^{ts}} \Pr(X\in K_J \text{ and } ||X-P_C X||^2>\chi^2_{1, 1-\alpha})\label{PrJ-a1ts}\\
&+\sum_{J\in\mathcal{J}_1^{os}} \Pr(X\in K_J \text{ and } ||X-P_C X||^2>\chi^2_{1, 1-\beta(X)})\label{PrJ-a1os}\\
&+\sum_{J\subseteq \{1,...,d_A\}| \textup{rk}(A_J)=0} \Pr(X\in K_J \text{ and } ||X-P_C X||^2>\chi^2_{0, 1-\alpha}), \label{PrJ-a0}
\end{align}
where the first equality follows from Lemma \ref{lem:partition}(b,c), and the second equality uses Lemma \ref{lem:partition2}(d) and the fact that $\beta(X)=\alpha$ whenever $\textup{rk}(A_{J(X)})\neq 1$ or $J\in\mathcal{J}_1^{ts}$. 
That latter fact follows because for $J\in\mathcal{J}_1^{ts}$ with $X\in K_J$, there exists $j,\ell\in J$ such that $\|a_\ell\|\inv a_\ell = - \|a_j\|\inv a_j$ and $\|a_\ell\|\inv b_\ell = - \|a_j\|\inv b_j$, which implies that $b_\ell-a'_\ell P_C X = b_j-a'_j P_C X = 0$ (and therefore $\tau(X)=0$). 

For each $J$, we consider the span of $V_J$ as a subspace of $\R^{d_m}$. 
Let $P_J$ denote the projection onto $\text{span}(V_J)$, and $M_J$ denote the projection onto its orthogonal complement. 
We note that, given $J$, there exists a $\kappa_J\in \text{span}(V_J)$ such that for every $z\in C_J$, $P_J z=\kappa_J$. 
This follows because for two $z_1, z_2\in C_J$, and for any $v\in \text{span}(V_J)$, $\langle z_1-z_2,v \rangle=0$, which implies $z_1-z_2\perp \text{span}(V_J)$, so that $P_J(z_1-z_2)=\bm{0}_{d_m}$. 
Thus, for any $X\in K_J$, we can write $P_J X=P_J(X-P_C X)+P_J P_C X=X-P_CX+\kappa_J$, where the second equality follows by lemma \ref{lem:partition}(a) and the above discussion. 
We also write $M_J X=X-P_J X=P_C X-\kappa_J$. 

First, let's consider the terms in (\ref{PrJ-a0}). 
For $J$ such that $\textup{rk}(A_J)=0$, we have $\text{span}(V_J) = \{\mathbf{0}_{d_m}\}$. 
Thus, $P_JX = \kappa_J = \mathbf{0}_{d_m}$. 
This implies that $\|X-P_CX\|=0$. 
Therefore, 
\begin{align}
\Pr(X\in K_J \text{ and } ||X-P_C X||^2>\chi^2_{0, 1-\alpha})=0.\label{r(J)=0}
\end{align}

For $J$ such that $\textup{rk}(A_J)>0$, we define a linear isometry from $\text{span}(V_J)$ to $\R^{\textup{rk}(A_J)}$. 
Let $B_J$ be a $d_m\times \textup{rk}(A_J)$ matrix whose columns form a basis for $\text{span}(V_J)$. 
Then $P_JX = B_J(B_J'B_J)^{-1}B_J' X$. 
The projection matrix $B_J(B_J'B_J)^{-1}B_J'$ is idempotent with rank $\textup{rk}(A_J)$, and thus there exists a $d_m\times \textup{rk}(A_J)$ matrix with orthonormal columns, $Q_J$, such that $Q_JQ_J' = B_J(B_J'B_J)^{-1}B_J'$. 
The linear isometry from $\text{span}(V_J)$ to $\R^{\textup{rk}(A_J)}$ is $Q_J(X) = Q_J'X$. 
This is an isometry because for any  $v_1, v_2\in \text{span}(V_J)$, 
\begin{align}
\|v_1 - v_2\|^2 &= (v_1-v_2)'(v_1-v_2) \nonumber\\
&=(v_1-v_2)'(P_J(v_1-v_2))\nonumber\\
&= (v_1-v_2)'Q_JQ_J'(v_1-v_2) \nonumber\\
&=\|Q_J(v_1)-Q_J(v_2)\|^2,
\end{align} 
where the second equality holds because $v_1, v_2\in \text{span}(V_J)$. 
Now let $Q_J' V_J = \{Q_J'v: v\in V_J\}$. 
Then $P_JX-\kappa_J\in V_J$ if and only if $Q_J'(P_JX-\kappa_J)\in Q_J' V_J $ because an isometry is bijective. 

Next, we consider the terms in (\ref{PrJ-a2}) and (\ref{PrJ-a1ts}). 
Notice that 
\begin{align}
&\Pr(X\in K_J \text{ and } ||X-P_C X||^2>\chi^2_{\textup{rk}(A_J), 1-\alpha})\notag\\
=&\Pr(M_J X+\kappa_J\in C_J, P_J X-\kappa_J\in V_J, \text{ and } ||P_J X-\kappa_J||^2>\chi^2_{\textup{rk}(A_J), 1-\alpha})\notag\\
=&\Pr(M_J X+\kappa_J\in C_J)\times \Pr(P_J X-\kappa_J\in V_J \text{ and } ||P_J X-\kappa_J||^2>\chi^2_{\textup{rk}(A_J), 1-\alpha}), \label{Prob-Dcmpsd-a}
\end{align}
where the first equality uses Lemma \ref{lem:partition}(a) and the facts that $M_JX+\kappa_J = P_CX$ and $X  = P_JX+M_JX$, and the second equality follows from the fact that $P_J X$ is independent of $M_J X$. 
Applying the isometry, we have
\begin{align}
&\Pr(P_J X-\kappa_J\in V_J \text{ and } ||P_J X-\kappa_J||^2>\chi^2_{\textup{rk}(A_J), 1-\alpha})\nonumber\\
&=\Pr(Q_J'(P_J X-\kappa_J)\in Q_J'V_J \text{ and } ||Q_J'(P_J X-\kappa_J)||^2>\chi^2_{\textup{rk}(A_J), 1-\alpha}). \label{Prob-decompose-b}
\end{align}

We would like to apply Lemma \ref{lem:prob-bnd} to this probability. 
Since $X\sim N\left(\mu,I\right)$, we have
\begin{align}
Q_J'(P_JX-\kappa_J)\sim N(Q_J'(P_J\mu-\kappa_J), Q_J'IQ_J) = N(Q_J'(P_J\mu-\kappa_J),I).
\end{align}
Also note that $Q_J'V_J $ is a cone in $\R^{\textup{rk}(A_J)}$. 
The random vector $Q_J'(P_J X-\kappa_J)\sim N(\gamma,I)$ where $\gamma=Q_J'(P_J \mu-\kappa_J)$. 
The vector $\gamma$ is in the polar cone because, for all $\tilde{y}\in Q_J'V_J$, there exists a  $y=\sum_{j\in J}v_j a_j \in V_J$ such that $\tilde{y} =Q_J'y$, and thus
\begin{align}
\langle \gamma,\tilde{y}\rangle 
&=\langle Q_J'(P_J\mu-\kappa_J), Q_J'y\rangle\nonumber\\
&=\langle P_J\mu-\kappa_J,y\rangle\nonumber\\
&=\langle( \mu-M_J \mu-P_J z),y\rangle\nonumber\\
&=\langle( \mu-M_J \mu-z+M_J z),y\rangle\nonumber\\
&=\langle(\mu-z),y \rangle\nonumber\\
&=\sum_{j\in J}v_j (\langle \mu, a_j\rangle -\langle z, a_j\rangle)\le 0,
\end{align}
where $z$ is any element\footnote{If $C_J$ is empty, so that no such $z$ exists, then (\ref{Prob-Dcmpsd-a}) is zero, and so (\ref{r(J)geq1}) below is not needed.} of $C_J$ so that $\kappa_J=P_J z$, the second equality holds because $\langle Q_J'(P_J\mu-\kappa_J), Q_J'y\rangle= y'Q_JQ_J'(P_J\mu-\kappa_J) = y'P_J(P_J\mu-\kappa_J) = y'(P_J\mu-\kappa_J)$, and the inequality follows because $\langle z, a_j\rangle=b_j\ge\langle \mu, a_j\rangle$, using the facts that $z\in C_J$ and $\mu\in C$. 

Therefore, we can apply Lemma \ref{lem:prob-bnd} to get that, for every $J\subseteq\{1,...,d_A\}$ such that $\textup{rk}(A_J)\ge 1$, we have
\begin{align}
&\Pr(P_J X-\kappa_J\in V_J \text{ and }||P_JX-\kappa_J||^2>\chi^2_{\textup{rk}(A_J),1-\alpha})\nonumber\\
&= Pr\left(Q_J'(P_J X-\kappa_J)\in Q_J'V_J \text{ and } ||Q_J'(P_J X-\kappa_J)||^2>\chi^2_{\textup{rk}(A_J), 1-\alpha}\right)\nonumber\\
&\le \alpha \Pr(Q_J'(P_J X-\kappa_J)\in Q_J'V_J)\nonumber\\
&=\alpha \Pr(P_J X-\kappa_J\in V_J), \label{r(J)geq1}
\end{align}
where the inequality holds as equality if $\gamma=Q_J'(P_J\mu-\kappa_J)=0$. 

Next, consider the terms in (\ref{PrJ-a1os}). 
For each $J\in\mathcal{J}_1^{os}$, let $\bar j\in J$ such that $\|a_{\bar j}\|\neq 0$. 
Notice that we can take $B_J = a_{\bar j}$, so that $P_J = \|a_{\bar j}\|^{-2}a_{\bar j}a'_{\bar j}$, $Q_J = \|a_{\bar j}\|^{-1}a_{\bar j}$, and $\kappa_J = \|a_{\bar j}\|^{-2}a_{\bar j}b_{\bar j}$. 
Notice that 
\begin{align}
&\Pr(X\in K_J \text{ and } ||X-P_C X||^2>\chi^2_{1, 1-\beta(X)})\notag\\
=&\Pr(M_J X+\kappa_J\in C_J, P_J X-\kappa_J\in V_J, \text{ and } ||P_J X-\kappa_J||^2>\chi^2_{1, 1-\beta(X)})\label{Prob-Dcmpsd-a1os}\\
=&\Pr(M_J X+\kappa_J\in C_J, Q'_J(P_J X-\kappa_J)\in Q'_JV_J, \text{ and } ||Q'_J(P_J X-\kappa_J)||^2>\chi^2_{1, 1-\beta(M_JX+\kappa_J)}), \notag
\end{align}
where the first equality uses the definitions of $M_J$, $P_J$, and $\kappa_J$, and the second equality uses the isometry $Q_J$, together with the fact that $\beta(X)$ depends on $X$ only through $M_J X+\kappa_J$ (because the formula for $\tau(X)$ only depends on $P_CX=M_J X+\kappa_J$). 

We next note that $Q'_J V_J=[0,\infty)$. 
This follows because for any $c\ge 0$, $c=Q'_J c a_{\bar j}\|a_{\bar j}\|$, where $c a_{\bar j}\|a_{\bar j}\|\in V_J$. 
Conversely, for any $v=\sum_{\ell\in J}c_\ell a_\ell\in V_J$ for some constants $c_\ell\ge 0$, we have $Q_J'v=\sum_{\ell\in J}c_\ell \|a_{\bar j}\|^{-1} a'_{\bar j}a_\ell $, where $a'_{\bar j}a_\ell\ge 0$ because $a_\ell$ is either zero or a positive scalar multiple of $a_{\bar j}$ by the definition of $\mathcal{J}_1^{os}$. 

Also, the fact that $X\sim N(\mu,I)$ implies that $Z:=Q'_J(P_JX-\kappa_J)\sim N(\gamma,1)$, where $\gamma = Q'_J(P_J\mu-\kappa_J)=\|a_{\bar j}\|\inv (a'_{\bar j}\mu - b_{\bar j})\le 0$. 
Note that $Z$ is independent of $M_JX$. 

Let $z_{1-\alpha}$ denote the $1-\alpha$ quantile of the standard normal distribution. 
We have 
\begin{align}
&\Pr(M_J X+\kappa_J\in C_J, Q'_J(P_J X-\kappa_J)\in Q'_JV_J, \text{ and } ||Q'_J(P_J X-\kappa_J)||^2>\chi^2_{1, 1-\beta(M_JX+\kappa_J)})\nonumber\\
=&\Pr(M_J X+\kappa_J\in C_J, Z>z_{1-\beta(M_JX+\kappa_J)/2})\nonumber\\
=&\E\mathds{1}(M_J X+\kappa_J\in C_J)\Pr(Z>z_{1-\beta(M_JX+\kappa_J)/2}|M_JX+\kappa_J)\nonumber\\
\le&\alpha \E\mathds{1}(M_J X+\kappa_J\in C_J)\Pr(Z>-\tau(M_JX+\kappa_J)|M_JX+\kappa_J)\label{rJ=1os_inequality}\\
=&\alpha\Pr(M_J X+\kappa_J\in C_J, Z>-\tau(M_JX+\kappa_J))\nonumber\\
=&\alpha\Pr(M_J X+\kappa_J\in C_J, Q'_J(P_J X-\kappa_J)\in Q'_JV_J) \nonumber\\
&+ \alpha \Pr(M_J X + \kappa_J\in C_J, Q'_J(P_J X-\kappa_J)\in(-\tau(M_JX+\kappa_J),0))\nonumber\\
=&\alpha(\Pr(X\in K_J)+\Pr(X\in C^\Delta_J)), \label{J1os_result}
\end{align}
where the first equality follows from the events $Z\ge 0$ and $Z^2>\chi^2_{1,1-\beta(M_J X+\kappa_J)}$ being equivalent to the event  $Z>z_{1-\beta(M_JX+\kappa_J)/2}$, the second and third equalities uses the conditional distribution of $Z$ given $M_JX+\kappa_J$, the inequality follows by Lemma \ref{lem:prob-bnd2}, the fourth equality follows from splitting the event $Z>-\tau$ into $Z\ge 0$ (equivalent to $Z\in Q'_JV_J$) and $Z\in (-\tau,0)$, and the final equality follows from the fact that $Q'_J(P_J X-\kappa_J)\in Q'_JV_J$ if and only if $P_J X-\kappa_J\in V_J$, the characterization of $K_J$ using Lemma \ref{lem:partition}(a), together with the argument that follows. 

To show (\ref{J1os_result}), we show that for all $J\in\mathcal{J}_1^{os}$, 
\begin{equation}
C^\Delta_J = \{x\in\R^{d_X}| M_Jx+\kappa_J\in C_J \text{ and }Q'_J(P_J x-\kappa_J)\in(-\tau(M_Jx+\kappa_J),0)\}. \label{J1os1}
\end{equation}
Denote the set on the right hand side of (\ref{J1os1}) by $\Upsilon$. 
We show (1) $x\in C^\Delta_J$ implies $x\in\Upsilon$ and (2) $x\in\Upsilon$ implies $x\in C^\Delta_J$. 
It is useful to point out that for any $x$, we can write $Q'_J(P_J x - \kappa_J)=\|a_{\bar j}\|\inv (a'_{\bar j} x - b_{\bar j})$ and $M_J x +\kappa_J =x-\|a_{\bar j}\|^{-2}a_{\bar j}(a'_{\bar j} x - b_{\bar j})$ using the formulas for $P_J$, $Q_J$, and $\kappa_J$. 

(1) Let $x\in C^\Delta_J$. 
We calculate that $M_Jx+\kappa_J\in C_J$ by showing that equality holds for every $\ell\in J$ and strict inequality holds for every $\ell\notin J$. 
For any $\ell\in J$ either $\ell\in J_{00}$ or $\ell\in J\cap J_{\neq 0}$. If $\ell\in J_{00}$, $a'_\ell (M_Jx+\kappa_J)=0=b_\ell$, so equality holds. 
If $\ell\in J\cap J_{\neq 0}$, 
\begin{equation}
a'_\ell(x-\|a_{\bar j}\|^{-2}a_{\bar j}(a'_{\bar j} x - b_{\bar j}))=a'_\ell(x-\|a_{\bar j}\|^{-1}\|a_\ell\|\inv a_{\bar j}(a'_\ell x - b_\ell))=b_\ell, \label{J1os2}  
\end{equation}
where the first equality uses the fact that $\|a_\ell\|\inv (b_\ell-a'_\ell x)=\|a_{\bar j}\|\inv (b_{\bar j}-a'_{\bar j} x)$ by the definition of $C^\Delta_J$, and the second equality uses the fact that $a'_\ell a_{\bar j} = \|a_{\bar j}\| \|a_\ell\|$ by the definition of $\mathcal{J}_1^{os}$. 
Therefore, equality holds for every $\ell\in J$. 
For any $\ell\in J^c$, we show that strict inequality holds. 
Either $\ell\in J_0/J_{00}$ or $\ell\in J_{\neq 0}/J$. 
If $\ell\in J_0/J_{00}$, $a'_\ell (M_Jx+\kappa_J)=0<b_\ell$.\footnote{$b_\ell$ cannot be negative because, by assumption, $\theta\in\Theta_0(F)$, so $\mu\in C$, and therefore $C$ is non-empty.} 
If $\ell\in J_{\neq 0}$, then 
\begin{align}
a'_\ell(x-\|a_{\bar j}\|^{-2}a_{\bar j}(a'_{\bar j} x - b_{\bar j}))=&a'_\ell x - \|a_{\bar j}\|^{-2}a'_\ell a_{\bar j}(a'_{\bar j}x-b_{\bar j})\nonumber\\
=&a'_\ell x-b_\ell - \|a_{\bar j}\|^{-2}a'_\ell a_{\bar j}(a'_{\bar j}x-b_{\bar j})+b_\ell\nonumber\\
<&\frac{\|a_\ell\|}{\|a_{\bar j}\|}(a'_{\bar j} x-b_{\bar j}) - \|a_{\bar j}\|^{-2}a'_\ell a_{\bar j}(a'_{\bar j}x-b_{\bar j})+b_\ell\nonumber\\
=&\frac{(\|a_\ell\|\|a_{\bar j}\|-a'_\ell a_{\bar j})(a'_{\bar j} x-b_{\bar j})}{\|a_{\bar j}\|^2}+b_\ell
\le b_\ell, \label{J1os3}
\end{align}
where the first inequality uses the fact that $\|a_{\bar j}\|\inv (b_{\bar j}-a'_{\bar j} x)<\|a_\ell\|\inv (b_\ell - a'_\ell x)$ by the definition of $x\in C^\Delta_J$ (because $\ell$ is not in the argmin), and the second inequality uses the fact that $a'_{\bar j} x<b_{\bar j}$ and $\|a_{\bar j}\|\|a_\ell\|\ge a'_\ell a'_{\bar j}$. 
This shows that for every $\ell\in J^c$ the inequality is strict. 
Therefore, $M_J x+\kappa_J\in C_J$. 

We also calculate that $Q'_J(P_Jx-\kappa_J)\in(-\tau(M_J x+\kappa_J),0)$. 
The fact that $x\in C_0$ implies that $a'_{\bar j} x - b_{\bar j}<0$, and so $Q'_J(P_Jx-\kappa_J)=\|a_{\bar j}\|\inv (a'_{\bar j} x - b_{\bar j})<0$. 
Let $\ell\in \{1,...,d_A\}/\{\bar j\}$.\footnote{We note here that $\tau(x)$ is defined for an arbitrary active inequality $\bar j\in J\cap J_{\neq 0}$. One can verify that the definition of $\tau(x)$ does not depend on which $\bar j\in J\cap J_{\neq 0}$ is selected.} 
We show that 
\begin{equation}
\|a_{\bar j}\|\inv (a'_{\bar j} x - b_{\bar j})>-\tau_j(M_Jx+\kappa_J).\label{J1os4}
\end{equation}
If the $\|a_\ell\| \|a_{\bar j}\|-a'_\ell a_{\bar j}=0$, then by definition the right hand side of (\ref{J1os4}) is $-\infty$. 
Otherwise, we can plug in $M_J x+\kappa_J = x-\|a_{\bar j}\|^{-2}a_{\bar j}a'_{\bar j} x + \|a_{\bar j}\|^{-2}a_{\bar j} b_{\bar j}$ and rewrite (\ref{J1os4}) as 
\begin{equation}
(a'_{\bar j} x - b_{\bar j})(\|a_\ell\| \|a_{\bar j}\|-a'_\ell a_{\bar j})>-\|a_{\bar j}\|^2(b_\ell - a'_\ell (x-\|a_{\bar j}\|^{-2}a_{\bar j}a'_{\bar j} x + \|a_{\bar j}\|^{-2}a_{\bar j} b_{\bar j})). \label{J1os5}
\end{equation}
We can simplify this to show that it holds if and only if 
\begin{equation}
\|a_{\bar j}\|\inv (b_{\bar j}-a'_{\bar j} x)<\|a_\ell\|\inv (b_\ell-a'_\ell x). \label{J1os6}
\end{equation}
The fact that $\|a_\ell\|\|a_{\bar j}\|\neq a'_\ell a_{\bar j}$ implies that $\ell\notin J$ (by the definition of $\mathcal{J}_1^{os}$) and therefore, by the definition of $C^\Delta_J$, (\ref{J1os6}) holds (because $\ell$ is not in the argmin). 
Therefore, (\ref{J1os4}) holds for every $\ell\in\{1,...,d_A\}/\{\bar j\}$, which implies that 
\begin{equation}
Q'_J(P_J x-\kappa_J)=\|a_{\bar j}\|\inv (a'_{\bar j} x - b_{\bar j})>-\tau(M_Jx+\kappa_J). \label{J1os7}
\end{equation}
This shows that $x\in\Upsilon$. 

(2) Let $x\in\Upsilon$. 
Consider the set $\argmin_{j\in J_{\neq 0}}\|a_j\|\inv (b_j-a'_j x)$. 
We first show that the argmin is equal to $J\cap J_{\neq 0}$. 
If $\ell\in J_{\neq 0}/J$, an algebraic manipulation similar to above shows that 
\begin{align}
Q'_J(P_J x-\kappa_J)&>-\tau(M_Jx+\kappa_J) \nonumber\\
\Rightarrow \|a_{\bar j}\|\inv (a'_{\bar j} x-b_{\bar j})&>-\frac{\|a_{\bar j}\|(b_\ell - a'_\ell (x-\|a_{\bar j}\|^{-2}a_{\bar j}(a'_{\bar j} x - b_{\bar j})))}{\|a_{\bar j}\|\|a_\ell\|-a'_\ell a_{\bar j}} \nonumber\\
\iff \|a_\ell\|\inv (b_\ell - a'_\ell x)&>\|a_{\bar j}\|\inv (b_{\bar j}-a'_{\bar j}x),  \label{J1os8}
\end{align}
where the first implication uses the definition of $\tau(x)$ and the ``iff'' follows from multiplying by $\|a_{\bar j}\| \|a_\ell\|-a'_\ell a_{\bar j}$ and cancelling the $a'_\ell a_{\bar j}$ term. 
This shows that $\ell\in J_{\neq 0}/J$ cannot be in the argmin. 
Also consider $\ell\in J\cap J_{\neq 0}$. 
The definition of $\mathcal{J}_1^{os}$ implies that $\|a_\ell\|\inv a_\ell = \|a_{\bar j}\|\inv a_{\bar j}$. 
Notice that 
\begin{align}
&0=b_{\bar j}-a'_{\bar j}(M_J x + \kappa_J)=b_\ell-a'_\ell(M_J x + \kappa_J)\nonumber\\
\iff &0=b_{\bar j}-a'_{\bar j}(x-\|a_{\bar j}\|^{-2}a_{\bar j}(a'_{\bar j} x - b_{\bar j}))=b_\ell-a'_\ell(x-\|a_{\bar j}\|^{-2}a_{\bar j}(a'_{\bar j} x - b_{\bar j}))\nonumber\\
\Rightarrow\hspace{2mm} &b_\ell=a'_\ell(x-\|a_{\bar j}\|^{-2}a_{\bar j}(a'_{\bar j} x - b_{\bar j})) \nonumber\\
\iff & \|a_\ell\|\inv (b_\ell - a'_\ell x) = \|a_\ell\|\inv a'_\ell \|a_{\bar j}\|^{-2}a_{\bar j}(b_{\bar j}-a'_{\bar j} x)=\|a_{\bar j}\|^{-1}(b_{\bar j}-a'_{\bar j} x), \label{J1os9}
\end{align}
where the first line holds because $M_J x + \kappa_J\in C_J$, the first ``iff'' holds by plugging in the formula for $M_Jx+\kappa_J$, the implication holds by solving for $b_\ell$ and cancelling $b_{\bar j}$, the second ``iff'' holds by rearranging and the fact that $a'_\ell a_{\bar j} = \|a_\ell\| \|a_{\bar j}\|$. 
We have shown that $\ell$ should be in the argmin. 
Therefore, the argmin is equal to $J\cap J_{\neq 0}$. 

We also show that $x\in C_0$. 
Note that $a'_{\bar j}x<b_{\bar j}$ because $Q'_J(P_J x-\kappa_J)<0$ and plugging in the formulas for $Q_J$, $P_J$, and $\kappa_J$. 
For any other $\ell\in J_{\neq 0}$, we have 
\begin{equation}
\|a_\ell\|\inv (b_\ell-a'_\ell x)\ge \|a_{\bar j}\|\inv (b_{\bar j}-a'_{\bar j}x)>0, \label{J1os10}
\end{equation}
because $\bar j$ belongs to the argmin. 
Thus, $x\in C_0$ because all the inequalities for $\ell\in J_{\neq 0}$ are inactive. 
Therefore, $x\in C^\Delta_J$. 

Therefore, we have shown (\ref{J1os1}), which implies (\ref{J1os_result}). 

To finish the proof of part (b), we plug in (\ref{r(J)=0}), (\ref{Prob-Dcmpsd-a}), (\ref{Prob-decompose-b}), (\ref{r(J)geq1}), (\ref{Prob-Dcmpsd-a1os}), and (\ref{J1os_result}) into (\ref{PrJ-a2}), (\ref{PrJ-a1ts}), (\ref{PrJ-a1os}), and (\ref{PrJ-a0}) to get that 
\begin{align}
&\sum_{J\subseteq \{1,...,d_A\}} \Pr(X\in K_J \text{ and } ||X-P_C X||^2>\chi^2_{\textup{rk}(A_J), 1-\alpha})\nonumber\\
\le& \sum_{J\subseteq \{1,...,d_A\}:\textup{rk}(A_J)\ge 2} \alpha \Pr(M_J X+\kappa_J\in C_J)\times \Pr(P_J X-\kappa_J\in V_J)\nonumber\\
&+\sum_{J\in\mathcal{J}_1^{ts}} \alpha \Pr(M_J X+\kappa_J\in C_J)\times \Pr(P_J X-\kappa_J\in V_J)\nonumber\\
&+\sum_{J\in\mathcal{J}_1^{os}} \alpha (\Pr(X\in K_J)+\Pr(X\in C^\Delta_J))\nonumber\\
=&\alpha\times \left(\sum_{J\subseteq \{1,...,d_A\}:\textup{rk}(A_{J(X)})>0} \Pr(X\in K_J)+\sum_{J\in\mathcal{J}_1^{os}} \Pr(X\in C^\Delta_J)\right)\nonumber\\
=&\alpha(1-\Pr(X\in C^|))\leq \alpha, \label{thm1partafinal}
\end{align}
where the first equality uses Lemma \ref{lem:partition}(a) and the fact that $P_JX$ is independent of $M_J X$, together with Lemma \ref{lem:partition2}(d), and the second equality uses Lemma \ref{lem:partition}(b), together with Lemma \ref{lem:partition2}(b). 

\underline{We next prove part (b)}. 
Fix a $J\subseteq\{1,\dots,d_A\}$. 
We first note that when $C_J$ is empty, the inequality in (\ref{r(J)geq1}) holds with equality because both sides are zero. 
When $C_J\neq \emptyset$, we show that $\kappa_J=P_J\mu$. 
Let $z\in C_J$ and for every $\lambda\in [0,1]$ let $\mu_\lambda=\lambda z+(1-\lambda)\mu$. 
Recall that $A\mu=b$. 
Thus, for each $\lambda\in (0,1]$
\begin{align}
a'_j \mu_\lambda&=\lambda a'_j z+(1-\lambda)a'_j\mu=\lambda b_j+(1-\lambda)b_j=b_j \text{ for }j\in J, \text{ and }\\
a'_j \mu_\lambda&=\lambda a'_j z+(1-\lambda)a'_j\mu<\lambda b_j+(1-\lambda)b_j=b_j \text{ for }j\in J^c.\nonumber
\end{align} 
This implies that  $\mu_\lambda\in C_J$, and hence, for every $\lambda\in (0,1]$, $P_J \mu_\lambda=\kappa_J$. 
Take $\lambda\rightarrow 0$ and by the continuity of the projection, $\kappa_J=P_J \mu$. 
Thus, $\gamma = Q_J'(P_J\mu-\kappa_J)=0$, implying that the inequality in (\ref{r(J)geq1}) holds with equality. 
The fact that $\gamma=0$ also implies that the inequality in (\ref{rJ=1os_inequality}) holds with equality by Lemma 6.  
The inequality in the last line of (\ref{thm1partafinal}) holds with equality by Lemma \ref{lem:partition2}(c). 
Thus, part (b) has been proved.

\underline{Part (c)}. 
Let $\tilde C=\{\mu\in\R^{d_m}| A_J\mu\le b_J\}$. 
Let $\tilde T_n(\theta)$, $P_{\tilde C} X$, $\tilde J(X)$, $\tilde r(X)$, $\tilde\tau(X)$, and $\tilde\beta(X)$ be defined with $A_J$ and $b_J$ in place of $A$ and $b$. 
For each $L\subseteq J$, also define $\tilde C_L$ and $\tilde K_L$ similarly. 
Note that all these objects also depend on $s$ because $A$ and $b$ may depend on $\theta$. 

Notice that 
\begin{align}
&\Pr\left(\phi^{\text{RCC}}_n(\theta_s,\alpha)\neq\phi^{\text{RCC}}_{n,J}(\theta_s,\alpha)\right)\nonumber\\
\le&\sum_{L\subseteq J}\Pr(X\in K_L \text{ and }\phi^{\text{RCC}}_n(\theta_s,\alpha)\neq\phi^{\text{RCC}}_{n,J}(\theta_s,\alpha))+\sum_{L\subseteq \{1,...,d_A\}: L\not\subseteq J} \Pr(X\in K_L)\nonumber\\
=&\sum_{L\subseteq J}\Pr(X\in K_L \text{ and }\phi^{\text{RCC}}_n(\theta_s,\alpha)\neq\phi^{\text{RCC}}_{n,J}(\theta_s,\alpha))+o(1)\label{idi1}\\
=&\sum_{L\subseteq J} \Pr\left(X\in K_L \text{ and }\chi^2_{\text{rk}(A_L),1-\beta(X)}\ge\|X-P_C X\|^2>\chi^2_{\text{rk}(A_L),1-\tilde\beta(X)}\right)+o(1)\label{idi2}\\
=&\sum_{L\subseteq J: \text{rk}(A_L)=1} \Pr\left(X\in K_L \text{ and }\chi^2_{1,1-\beta(X)}\ge\|X-P_C X\|^2>\chi^2_{1,1-\tilde\beta(X)}\right)+o(1)\label{idi3}\\
&\rightarrow0,\label{idi4}
\end{align}
where the first inequality follows from Lemma \ref{lem:partition}(b) and the subsequent equalities and convergence are justified below. 

For (\ref{idi1}), let $\tilde X = X-\mu\sim N(0,I)$. 
Fix the value of $\tilde X$. 
We show that for any $L\not\subseteq J$, $\tilde X+\mu\notin K_L$ eventually as $s\rightarrow\infty$. 
For $\ell\in L$ but $\ell\notin J$, we have 
\begin{equation}
a'_\ell P_{C}(\tilde X+\mu)=a'_\ell (P_{C-\mu}\tilde X+\mu).
\end{equation}
This expression is less than $b_\ell$ eventually because 
\begin{equation}
\frac{1}{\|a_\ell\|}\left(a'_\ell P_{C-\mu}\tilde X+a'_\ell \mu-b_\ell\right)=\frac{a'_\ell P_{C-\mu}\tilde X}{\|a_\ell\|}+\frac{a'_\ell \mu-b_\ell}{\|a_\ell\|}\rightarrow -\infty \text{ as }s\rightarrow\infty, \label{idi5}
\end{equation}
where the convergence follows from the boundedness of the first term (for fixed $\tilde X$ using the fact that $0\in C-\mu$) and by assumption on the second term. 
Therefore, $P_C(\tilde X+\mu)\notin \tilde C_L$ eventually. 
By Lemma \ref{lem:partition}(a), $\tilde X+\mu\notin K_L$ eventually as $s\rightarrow\infty$. 
Therefore, for every $L\not\subseteq J$, 
\[
\Pr(X\in K_L)=\Pr(\tilde X+\mu\in K_L)\rightarrow 0, 
\]
where the equality follows from the fact that $X$ has the same distribution as $\tilde X+\mu$, and the convergence follows from the bounded convergence theorem. 

For (\ref{idi2}), note that for any $L\subseteq J$, if $X\in K_L$, then $P_CX\in C_L$ by Lemma \ref{lem:partition}(a). 
Also, $C_L\subseteq C\subseteq \tilde C$ because $\tilde C$ is defined by removing some inequalities. 
By the uniqueness of projection onto a convex set, it follows that $P_CX=P_{\tilde C}X$. 
The fact that $P_{\tilde C}X\in C_L$ implies that $P_{\tilde C}X\in \tilde C_L$ because $C_L\subseteq \tilde C_L$. 
By Lemma \ref{lem:partition}(b), this implies that $X\in\tilde K_L$. 
Thus, $r(X)=\tilde r(X) = \text{rk}(A_L)$, and $T_n(\theta)=X-P_C X=\tilde T_n(\theta)$. 
Also, the fact that $\tilde\beta(X)\ge \beta(X)$ implies that the only way $\phi^{\text{RCC}}_n(\theta_s,\alpha)\neq\phi^{\text{RCC}}_{n,J}(\theta_s,\alpha)$ is if $\chi^2_{\text{rk}(A_L),1-\beta(X)}\ge\|X-P_C X\|^2>\chi^2_{\text{rk}(A_L),1-\tilde\beta(X)}$.

For (\ref{idi3}), we use the fact that $\beta(X)=\tilde\beta(X)=\alpha$ whenever $\text{rk}(A_L)\neq 1$. 

Finally we show (\ref{idi4}). 
For each $L\subseteq J$, note that $\text{rk}(A_L)$ may depend on $s$. 
Fix any subsequence in $s$ such that $\text{rk}(A_L)=1$ along the subsequence. 
Write $P_L=Q^P_L{Q^P_L}'$ and $M_L=Q^M_L {Q^M_L}'$, where $Q^P_L$ is a $d_m\times 1$ vector with unit length, $Q^M_L$ is a $d_m\times (d_m-1)$ matrix with orthonormal columns, and ${Q^P_L}'Q^M_L=0$. 
Let $\tilde X=(\tilde x_1,\tilde x_2)\sim N(0,I)$, where $\tilde x_1\in\R^{d_m-1}$ and $\tilde x_2\in\R$. 
We can then write 
\begin{equation}
X=Q^M_L\tilde x_1+Q^P_L \tilde x_2+\mu.\label{idi6}
\end{equation} 
Note that $\tilde X$ does not depend on $s$, while $\mu$, $Q^M_L$, and $Q^P_L$ may. 

For each $L$, we can rewrite the term in (\ref{idi3}) as 
\begin{align}
&\int_{\tilde x_1}\int_{\tilde x_2}1\{X\in K_L\}1\{\chi^2_{1,1-\beta(X)}\ge\|X-P_C X\|^2>\chi^2_{1,1-\tilde\beta(X)}\}\phi(\tilde{x}_1)\phi(\tilde{x}_2)d\tilde x_2d\tilde x_1\nonumber\\
=&\int_{\tilde x_1}1\{M_L X+\kappa_L\in C_L\} \int_{\tilde x_2}g_s(\tilde x_1,\tilde x_2)\phi(\tilde{x}_2)d\tilde x_2 \phi(\tilde{x}_1)d\tilde x_1, \label{idi7}
\end{align}
where $X$ is viewed as a function of $(\tilde x_1,\tilde x_2)$ using (\ref{idi6}), $\phi(\cdot)$ is the probability density function of the standard normal distribution of the dimension determined by the dimension of its argument, and 
\begin{equation}
g_s(\tilde x_1,\tilde x_2)=1\{P_L X-\kappa_L\in V_L\}1\{\chi^2_{1,1-\beta(M_LX+\kappa_L)}\ge\|P_L X-\kappa_L\|^2>\chi^2_{1,1-\tilde\beta(M_LX+\kappa_L)}\}, \label{idi8}
\end{equation}
which uses the same decomposition of $X\in K_L$ as in (\ref{Prob-Dcmpsd-a}), and the fact that $M_LX$ only depends on $\tilde x_1$ (and that $\beta(X)$ and $\tilde\beta(X)$ only depend on $X$ through $P_CX=M_LX+\kappa_L$). 
Fix $\tilde x_1$. 
We show that the inner integral goes to zero as $s\rightarrow\infty$. 

Fix an arbitrary subsequence in $s$. 
We show that there exists a further subsequence such that the inner integral goes to zero. 
Since $\beta(M_L X+\kappa_L)$ and $\tilde\beta(M_L X+\kappa_L)$ do not depend on $\tilde x_2$ and both lie in $[\alpha,2\alpha]$ for all $s$, there exists a further subsequence along which both converge. 
Denote the limits by $\beta_\infty$ and $\tilde\beta_\infty$. 
We show that the limits must be the same: $\beta_\infty=\tilde\beta_\infty$. 
Let $\bar j\in L$ such that $a_{\bar j}\neq 0$. 
Then note that for each $\ell\notin J$, 
\begin{equation}
\tau_\ell(X)=\frac{\|a_{\bar j}\|(b_\ell-a'_\ell P_C X)}{\|a_{\bar j}\|\|a_\ell\|-a'_{\bar j}a_\ell}=\frac{(b_\ell-a'_\ell P_C X)/\|a_\ell\|}{1-a'_{\bar j}a_\ell/(\|a_{\bar j}\|\|a_\ell\|)}\ge \frac{1}{2}\frac{b_\ell-a'_\ell P_CX}{\|a_\ell\|}\rightarrow \infty,
\end{equation}
where the convergence follows from (\ref{idi5}). 
Therefore, 
\begin{equation}
\tau(X) = \inf_{\ell\neq \bar j}\tau_\ell(X)=\min\left(\inf_{\ell\in J; \ell\neq \bar j}\hat\tau_\ell(X),\inf_{\ell\notin J}\hat\tau_\ell(X)\right)=\min\left(\tilde\tau(X),\inf_{\ell\notin J}\hat\tau_\ell(X)\right). 
\end{equation}
If $\tilde\tau(X)\rightarrow\infty$, then $\tau(X)\rightarrow\infty$ too, and if $\tilde\tau(X)$ converges to a finite value, $\tau(X)$ converges to the same value. 
This shows that $\beta_\infty=\tilde\beta_\infty$. 

Returning to the inner integral, note that $P_LX-\kappa_L=Q^P_L\tilde x_2+P_L\mu-\kappa_L$. 
Take a further subsequence such that $P_L\mu-\kappa_L$ diverges or converges and such that $Q^P_L$ converges to $Q^P_{L,\infty}$ (since $Q^P_L$ has unit length, it must converge along a subsequence). 
If $P_L\mu-\kappa_L$ diverges, then for every $\tilde x_2$, $\|Q^P_L\tilde x_2+P_L\mu-\kappa_L\|^2\ge (\|P_L\mu-\kappa_L\|-\|Q^P_L\tilde x_2\|)^2\rightarrow\infty$, so $g_s(\tilde x_1,\tilde x_2)=0$ eventually as $s\rightarrow\infty$ along this subsequence. 
Therefore by the bounded convergence theorem, the inner integral in (\ref{idi7}) goes to zero. 
If $P_L\mu-\kappa_L$ converges to some $\kappa_\infty$, then for every $\tilde x_2$ such that $\|Q^P_{L,\infty}\tilde x_2+\kappa_\infty\|^2\neq\chi^2_{1,1-\beta_\infty}$, $g_s(\tilde x_1,\tilde x_2)=0$ eventually (along this subsequence). 
Note that the set of such $\tilde x_2$ is a set of probability one with respect to $\tilde x_2\sim N(0,1)$. 
Therefore by the bounded convergence theorem, the inner integral in (\ref{idi7}) goes to zero. 

Since the inner integral in (\ref{idi7}) converges to zero for every fixed $\tilde x_1$, by the bounded convergence theorem, the outer integral converges to zero too. 
This shows (\ref{idi4}). \qed
 
\subsection{Proof of the Auxiliary Lemmas}\label{sub:proof-lem}

\subsection*{Proof of Lemma \ref{lem:partition}}
\begin{enumerate}[label=(\alph*)]
\item By assumption, $X\in K_J=C_J+V_J$. 
So, we write $X=X_1+X_2$, where $X_1\in C_J$ and $X_2\in V_J$. 
Then, $P_C X_1=X_1$ because $X_1\in C$ already. 
We show that $P_C X=X_1$. 
By a property of projection onto convex sets, it is necessary and sufficient that for all $y\in C$, we have $\langle X-X_1,y-X_1\rangle\le 0$.\footnote{See Section 3.12 in \cite{Luenberger1969}. Hereafter, call this property of projection onto a convex set the ``inner-product property.''} 
This follows because $X_2=\sum_{j\in J} v_j a_j$ with $v_j\ge 0$, so 
\begin{equation}
\langle X_2,y-X_1 \rangle=\sum_{j\in J}v_j (\langle a_j,y \rangle-\langle a_j,X_1\rangle) \le 0, 
\end{equation}
where the inequality uses the fact that $y\in C$, so $a'_jy\le b_j$ and $X_1\in C_J$, so $a'_j X_1=b_j$. 
Combining these, we get that $P_C X=X_1\in C_J$ and $X-P_CX=X-X_1=X_2\in V_J$. 

\item We first show that every $X$ belongs to some $K_J$. 
For every $X$, $P_C X\in C$, so there exists a $J$ such that $P_C X\in C_J$. 

By the inner-product property of projection, we know that for all $y\in C$, $\langle y-P_C X,X-P_C X\rangle\le 0$. 
Using this fact, let $z\perp \text{span}(V_J)$. 
Then, there exists a $\epsilon>0$ such that $P_C X+\epsilon z$ and $P_C X-\epsilon z$ both belong to $C$.\footnote{This uses the slackness of the inequalities in the definition of $C_J$.} 
Then, $\langle \epsilon z,X-P_C X\rangle\le 0$ and $\langle -\epsilon z,X-P_C X\rangle\le 0$. 
These two inequalities imply that $\langle z,X-P_C X\rangle =0$. 
Thus, $X-P_C X$ is orthogonal to all vectors, $z$, which are orthogonal to $\text{span}(V_J)$. 
This implies that $X-P_C X\in \text{span}(V_J)$. 

If $X-P_C X\notin V_J$, then by the separating hyperplane theorem,\footnote{See Section 11 of \cite{Rockafellar1970} or Section 5.12 in \cite{Luenberger1969}.} there exists a direction, $c\in \text{span}(V_J)$ such that $\langle c, X-P_C X\rangle>0$ and $\langle c, a_j\rangle<0$ for all $j\in J$. 
We consider $P_C X+\epsilon c$. 
We show that for $\epsilon$ sufficiently small, (1) $P_C X+\epsilon c\in C$, and (2) $\langle X-P_C X,\epsilon c \rangle >0$. 

(1) For $j\in J$, $\langle P_CX+\epsilon c,a_j \rangle=b_j+\epsilon\langle c,a_j\rangle<b_j$, where the equality follows because $P_C X\in C_J$ and the inequality follows from the definition of $c$. 
For $j\in J^c$, $\langle P_C X +\epsilon c, a_j\rangle=\langle P_C X, a_j\rangle +\epsilon\langle c, a_j\rangle$, which is less than $b_j$ for $\epsilon$ sufficiently small because $\langle P_C X, a_j\rangle<b_j$. 

(2) $\langle X-P_C X, \epsilon c\rangle = \epsilon\langle X-P_C X,  c\rangle >0$ by the definition of $c$. 

This contradicts the inner-product property of projection onto a convex set, and therefore $X-P_C X\in V_J$, and $X\in K_J$. 

We next show that no $X$ belongs to two distinct $K_J$. 
If $X\in K_J$ and $K_{J'}$, then, by part (a), $P_C X\in C_J$ and $P_C X\in C_{J'}$. 
But this is a contradiction because the projection onto a convex set is unique, and the $C_J$ form a partition of $C$. 

\item If $X\in K_J$, then $P_C X\in C_J$, so all the inequalities in $J$ are active. 
If $X\notin K_J$, then $X$ is in a different $K_{J'}$, for some $J'\neq J$, by part (b). 
Thus, $J\neq J(X)=J'$. \qed
\end{enumerate}

\subsection*{Proof of Lemma \ref{lem:partition2}}
\begin{enumerate}[label=(\alph*)]
\item Note that $J_{00}$ satisfies $\textup{rk}(A_{J_{00}})=0$. 
Thus, it is sufficient to show that $C_J=\emptyset$ for all $J\subseteq J_0$ that are not $J_{00}$. 
If $J\neq J_{00}$, then either (i) there exists $j\in J_{00}/J$ or (ii) there exists $j\in J/J_{00}$. 
In the first case, any $x\in C_J$ would have to satisfy $0'x<0$, a contradiction. 
In the second case, any $x\in C_J$ would have to satisfy $0'x=b_j$, where $b_j\neq 0$, another contradiction. 

\item We first show that the $C^\Delta_J$ are disjoint for different $J\in\mathcal{J}_1^{os}$. 
If $x\in C^\Delta_{J_1}\cap C^\Delta_{J_2}$ for $J_1, J_2\in\mathcal{J}_1^{os}$, then both 
\begin{equation}
J_{\neq 0}\cap J_1 = J_{\neq 0}\cap J_2 = \text{argmin}_{j\in J_{\neq 0}} \|a_j\|\inv (b_j-a'_j x)
\end{equation}
and 
\begin{equation}
J_0\cap J_1 = J_0\cap J_2 = J_{00}. 
\end{equation}
This implies that $J_1 = J_2$. 
Part (b) then follows from the definition of $C^|$. 

\item For any $x\in C^|$, let $\tilde J_{\neq 0}(x) = \text{argmin}_{j\in J_{\neq 0}} \|a_j\|\inv (b_j-a'_j x)$. 
We show below that
\begin{align}
\exists j,\ell\in\tilde{J}_{\neq 0}(x)\text{ s.t. }\|a_j\|^{-1}a_j\neq \|a_\ell\|^{-1}a_\ell.\label{ajl}
\end{align}
That implies that for any $x\in C^|$, there exists $j,\ell\in J_{\neq 0}$ such that $\|a_j\|^{-1}a_j\neq \|a_\ell\|^{-1}a_\ell$ and $\|a_j\|\inv (b_j-a'_j x)=\|a_\ell\|\inv (b_\ell-a'_\ell x)$. 
Or equivalently,
\begin{align}
C^|\subseteq \cup_{\underset{\|a_{j}\|a_{\ell}\neq\|a_{\ell}\|a_{j}}{j,\ell\in J_{\neq 0}}}\{x\in\R^{d_m}| \|a_{j}\|b_{\ell}-\|a_{\ell}\|b_{j}=(\|a_{j}\|a_{\ell}-\|a_{\ell}\|a_{j})'x\}.
\end{align}
Since the right hand-side is a finite union of measure-zero subspaces of $\mathbb{R}^{d_m}$, it must be that  $C^|$ has Lebesgue measure zero, establishing part (c). 

Now we show (\ref{ajl}). 
Let $\tilde J(x) = J_{00}\cup \tilde J_{\neq 0}(x)$. 
We note that $\tilde J_{\neq 0}(x)$ is not empty because $A\neq 0_{d_A\times d_X}$. 
This implies that $\textup{rk}(A_{\tilde J(x)})\ge 1$. 
Then there are two possibilities: $\textup{rk}(A_{\tilde J(x)})\ge 2$ and $\textup{rk}(A_{\tilde J(x)})= 1$. 
In the first case, (\ref{ajl}) holds trivially. 

In the latter case, we first show that $\tilde{J}(x)\in {\cal J}_1$. 
Suppose there exists $j\in\tilde J(x)$ and $\ell\in \{1,...,d_A\}/\tilde J(x)$ such that $\|a_j\|>0, \|a_\ell\|>0$, $\frac{a_j}{\|a_j\|}= \frac{a_\ell}{\|a_\ell\|}$, and $\frac{b_j}{\|a_j\|}= \frac{b_\ell}{\|a_\ell\|}$. 
This implies that 
\begin{equation}
\|a_\ell\|\inv (b_\ell-a'_\ell x)=\|a_j\|\inv (b_j-a'_j x), 
\end{equation}
so $\ell$ should also belong to $\tilde J(x)$. 
Since such a $j$ and $\ell$ cannot exist, it must be the case that $\tilde J(x)\in \mathcal{J}_1$. 
The fact that $x\in C^|$ means that $\tilde J(x)\notin \mathcal{J}_1^{os}$. 
Thus, it must be that $\tilde J(x)\in {\cal J}_1^{ts}$, which also implies (\ref{ajl}). 
Therefore (\ref{ajl}) holds in all cases. 
This concludes the proof of part (c).

\item First note that for every $J\in\mathcal{J}_1^{os}\cup \mathcal{J}_1^{ts}$ we have $\textup{rk}(A_J)=1$. 
Thus, it is sufficient to show that for every $J\subseteq \{1,...,d_A\}$ with $\textup{rk}(A_J)=1$ and $J\notin \mathcal{J}_1^{os}\cup \mathcal{J}_1^{ts}$, we have $K_J=\emptyset$. 

Note that if $J\cap J_0\neq J_{00}$, then either (i) there exists $j\in J_{00}/(J_0\cap J)$ or (ii) there exists $j\in (J\cap J_0)/J_{00}$. 
In the first case, any $x\in C_J$ would have to satisfy $0'x<0$, a contradiction. 
In the second case, any $x\in C_J$ would have to satisfy $0'x=b_j$, where $b_j\neq 0$, another contradiction. 
This implies that $C_J$, and therefore $K_J$, is empty. 

We next note that if $j\in J$ while $\ell\in\{1,...,d_A\}/J$ with $\|a_j\|>0$, $\|a_\ell\|>0$, $\frac{a_j}{\|a_j\|}= \frac{a_\ell}{\|a_\ell\|}$, and $\frac{b_j}{\|a_j\|}= \frac{b_\ell}{\|a_\ell\|}$, then any $x\in C_J$ should satisfy 
\begin{equation}
\|a_\ell\|\inv (b_\ell-a'_\ell x)=\|a_j\|\inv (b_j-a'_j x)=0, 
\end{equation}
so $\ell$ should also belong to $J$. 
This contradiction implies that $C_J$, and therefore $K_J$, must be empty. 

This implies that the only nonempty $K_J$ with $\textup{rk}(A_J)=1$ must belong to $\mathcal{J}_1$. 
If we suppose that $J\notin \mathcal{J}_1^{os}$, then there must exist $j,\ell\in J$ s.t. $\|a_j\|>0, \|a_\ell\|>0$, and $\frac{a_j}{\|a_j\|}\neq\frac{a_\ell}{\|a_\ell\|}$. 
However, since $\textup{rk}(A_J)=1$, $a_\ell$ and $a_j$ must be collinear. 
This implies that $\frac{a_j}{\|a_j\|}=-\frac{a_\ell}{\|a_\ell\|}$. 
Then, any $x\in C_J$ must satisfy 
\begin{equation}
0=\|a_\ell\|\inv (b_\ell-a'_\ell x)=\|a_j\|\inv (b_j-a'_j x). 
\end{equation}
This implies $\|a_\ell\|\inv b_\ell=-\|a_j\|\inv b_j,$ which implies that $J\in \mathcal{J}^{ts}$. 

Therefore, the only $J\subseteq\{1,...,d_A\}$ with $\textup{rk}(A_J)=1$ and $K_J\neq \emptyset$ belong to $\mathcal{J}_1^{os}\cup\mathcal{J}_1^{ts}$. \qed
\end{enumerate}

\subsection*{Proof of Lemma \ref{lem:prob-bnd2}}
For every $\lambda\ge 0$, let 
\begin{equation}
f(\lambda)=\int_{-\tau}^\infty (\alpha -\mathds{1}\{Z>z_{1-\beta/2}\})e^{-\frac{1}{2}(Z+\lambda)^2}dZ. 
\end{equation}
We show that $f(\lambda)\ge 0$ for all $\lambda\ge 0$. 
This is sufficient because 
\begin{align}
&\alpha\Pr{}_\mu(Z\ge -\tau)-Pr_\mu(\{Z\ge z_{1-\beta/2}\})\nonumber\\
&=\int_{-\tau}^\infty (\alpha -\mathds{1}\{Z>z_{1-\beta/2}\})\frac{1}{\sqrt{2\pi}}e^{-\frac{1}{2}(Z-\mu)^2}dZ\nonumber\\
&=\frac{f(-\mu)}{\sqrt{2\pi}}\ge 0
\end{align}
for all $\mu\le 0$. 

Let $f'(\lambda)$ denote the derivative of $f$. 
We show that (1) $f(0)\ge 0$ and (2) for all $\lambda\ge 0$, $f'(\lambda)\ge -\left(z_{1-\beta/2}+\lambda\right)f(\lambda)$. 
Together, these two properties imply that $f(\lambda)\ge 0$ because, if not, then there exists a $\lambda>0$ such that $f(\lambda)<0$. 
Then, by the mean value theorem, there exists a $\tilde\lambda\in(0,\lambda)$ such that $f(\tilde\lambda)<0$ and $f'(\tilde\lambda)<0$, which contradicts property (2). 

Property (1) holds because 
\begin{align}
\frac{f(0)}{\sqrt{2\pi}}&=\int_{-\tau}^\infty (\alpha -\mathds{1}\{Z>z_{1-\beta/2}\})\frac{1}{\sqrt{2\pi}}e^{-\frac{1}{2}Z^2}dZ\nonumber\\
&=\alpha\Phi(\tau)-(1-\Phi(z_{1-\beta/2}))=\alpha \Phi(\tau)-\beta/2=0. 
\end{align}
This also shows that equality holds when $\mu=0$. 

To show that property (2) holds, we evaluate 
\begin{align}
f'(\lambda)&=\frac{d}{d\lambda}\int_{-\tau}^\infty (\alpha -\mathds{1}\{Z>z_{1-\beta/2}\})e^{-\frac{1}{2}(Z+\lambda)^2}dZ\nonumber\\
&=-\int_{-\tau}^\infty (Z+\lambda)(\alpha -\mathds{1}\{Z>z_{1-\beta/2}\})e^{-\frac{1}{2}(Z+\lambda)^2}dZ\nonumber\\
&=-\int_{-\tau}^{z_{1-\beta/2}} \alpha(Z+\lambda)e^{-\frac{1}{2}(Z+\lambda)^2}dZ\nonumber\\
&\hphantom{=}+\int_{z_{1-\beta/2}}^\infty (1-\alpha)(Z+\lambda)e^{-\frac{1}{2}(Z+\lambda)^2}dZ\nonumber\\
&\ge-\int_{-\tau}^{z_{1-\beta/2}} \alpha(z_{1-\beta/2}+\lambda)e^{-\frac{1}{2}(Z+\lambda)^2}dZ\nonumber\\
&\hphantom{=}+\int_{z_{1-\beta/2}}^\infty (1-\alpha)(z_{1-\beta/2}+\lambda)e^{-\frac{1}{2}(Z+\lambda)^2}dZ\nonumber\\
&=-(z_{1-\beta/2}+\lambda)\int_{-\tau}^\infty (\alpha -\mathds{1}\{Z>z_{1-\beta/2}\}) e^{-\frac{1}{2}(Z+\lambda)^2}dZ\nonumber\\
&=-(z_{1-\beta/2}+\lambda) f(\lambda),
\end{align}
where the second equality follows by dominated convergence and the inequality follows from the events $\{Z>z_{1-\beta/2}\}$ and $\{Z\le z_{1-\beta/2}\}$. \qed

\section{Proof of Theorem \ref{thm:asysize-iid} and General Asymptotic Results}\label{sec:asyRCC}
In this section we prove a general theorem for uniform asymptotic properties of the CC and RCC tests. 
The general theorem is used to prove Theorem \ref{thm:asysize-iid}. 

\subsection{A General Asymptotic Theorem}
The following general asymptotic theorem uses the following condition. 
In this section, we sometimes make explicit the dependence of $A$ and $b$ on $\theta$, denoting them by $A(\theta)$ and $b(\theta)$. 
The rows of $A(\theta)$ are denoted by $a_j(\theta)$, and submatrices composed of the rows of $A(\theta)$ are denoted by $A_J(\theta)$. 

\begin{assumption}\label{assu:para_space} 
The given sequence $\{(F_n,\theta_n):F_n\in\mathcal{F}, \theta_n\in\Theta_0(F_n)\}_{n=1}^{\infty}$ satisfies, for every subsequence, $n_m$, there exists a further subsequence, $n_q$, and there exists a sequence of positive definite $d_m\times d_m$ matrices, $\{D_q\}$ such that:

\textup{(a)} Under the sequence $\{F_{n_q}\}_{q=1}^\infty$,
\begin{equation}
\sqrt{n_q}D_{q}^{-1/2}(\overline{m}_{n_q}(\theta_{n_q})-\E_{F_{n_q}}\overline{m}_{n_q}(\theta_{n_q}))\to_d N(\mathbf{0},\Omega), \label{weakconvergence}
\end{equation}
for some positive definite correlation matrix, $\Omega$, and
\begin{equation}
\|D_{q}^{-1/2}\widehat\Sigma_{n_q}(\theta_{n_q})D_{q}^{-1/2}-\Omega\|\rightarrow_p 0. \label{varianceconsistency}
\end{equation}

\textup{(b)}  $\Lambda_q A(\theta_{n_q}) D_q\rightarrow \bar A_0$ for some $d_A\times d_m$ matrix $\bar{A}_0$, and for every $J\subseteq\{1,...,d_A\}$, $\textup{rk}(I_JA(\theta_{n_q})D_q)=\textup{rk}(I_J \bar A_0)$, where $\Lambda_q$ is the diagonal $d_A\times d_A$ matrix whose $j$th diagonal entry is one if $e'_j A(\theta_{n_q})=\mathbf{0}$ and $\|e'_j A(\theta_{n_q})D_q\|\inv$ otherwise.
\end{assumption}

\noindent\textbf{Remark.} 
The matrix $D_q$ typically is the diagonal matrix of variances of the elements of $\sqrt{n_q}\overline{m}_{n_q}(\theta_{n_q})$. 
In part (a), we allow each diagonal element to go to zero (or infinity) at different rates, to incorporate the cases where different moments are on different scales or where different moments involve time series processes that are integrated at different orders. 
\cite{AndrewsGuggenberger2009}, AS, and \cite{AndrewsChengGuggenberger2019} also use a diagonal normalizing matrix for this purpose. 

Moreover, the matrix $D_q$ can be non-diagonal, which is useful when the asymptotic variance matrix of $\sqrt{n_q}(\overline{m}_{n_q}(\theta _{n_q})- \E_{F_n}\overline{m}_{n_q}(\theta_{n_q}))$ is singular but a certain rotation of the vector with proper scaling has a non-singular asymptotic variance matrix. 

Part (b) is not required to show the uniform asymptotic validity of the RCC test. 
It is only used to show asymptotic size-exact and the asymptotic IDI property. 
The existence of $\bar A_0$ follows by the choice of the subsequence, while the rank condition is used to verify Lemma \ref{lem:set_consistency}, below. \qed

The following theorem is a general asymptotic theorem used to show the uniform asymptotic properties of the RCC test. 
\begin{theorem}\label{thm:momineq-sz} 
\textup{(a)} Suppose Assumption \textup{\ref{assu:para_space}(a)} holds for all sequences $\{(F_n,\theta_n): F_n\in {\cal F}, \theta_n\in\Theta_0(F_n)\}_{i=1}^n$. 
Then,
\begin{align*}
\underset{n\to\infty}{\textup{limsup}}\sup_{F\in {\cal F}}\sup_{\theta\in\Theta_0(F)}\E_{F}(\phi^{\textup{RCC}}_n(\theta,\alpha))\leq \alpha.
\end{align*}
Next consider a sequence $\{(F_n,\theta_n):F_n\in{\cal F},\theta_n\in\Theta_0(F_n)\}_{n=1}^\infty$ satisfying Assumption \textup{\ref{assu:para_space}(a,b)}.

\textup{(b)} If, along any further subsequence, for all $j=1,...,d_A$, $\sqrt{n_q}e'_j\Lambda_q(A(\theta_{n_q})\E_{F_{n_q}}\overline{m}_{n_q}(\theta_{n_q})-b(\theta_{n_q}))\rightarrow 0$, and if $\bar A_0\neq 0_{d_A\times d_m}$, then, 
\[
\lim_{n\rightarrow\infty}\E_{F_n}\phi^{\textup{RCC}}_n(\theta_n,\alpha)=\alpha. 
\]

\textup{(c)} If, for $J\subseteq\{1,...,d_A\}$, along any further subsequence, $\sqrt{n_q}e'_j\Lambda_q(A(\theta_{n_q})\E_{F_{n_q}}\overline{m}_{n_q}(\theta_{n_q})-b(\theta_{n_q}))\rightarrow-\infty$ as $q\rightarrow\infty$, for all $j\notin J$, then 
\[
\lim_{n\rightarrow\infty}\textup{Pr}_{F_n}\left(\phi^{\textup{RCC}}_n(\theta_n,\alpha)\neq \phi^{\textup{RCC}}_{n,J}(\theta_n,\alpha)\right)=0. 
\]
\end{theorem}
\noindent \textbf{Remarks.}  
(1) Notice that no assumptions are placed on $A(\theta)$ for Theorem \ref{thm:momineq-sz}(a). It can be low-rank or any submatrix of $A(\theta)$ can be local to singular as $\theta$ varies. This is achieved by an extra step in the proof that adds inequalities that are redundant in the finite sample but are relevant in the limit (see Lemma \ref{BExistence} below). 

(2) Part (c) states the asymptotic IDI property of the RCC test under Assumption \ref{assu:para_space}. Part (c) can be combined with part (b) to show that the RCC test has exact asymptotic size (and thus not conservative) if there exists a sequence $(F_n,\theta_n)$ along which a subset of inequalities are binding and the rest asymptote to infinitely slack, a condition that is automatically satisfied if, for example, the data are stationary, the moment inequalities are continuous in $\theta$, and $\Theta\neq \Theta_0(F_0)$ where such a sequence is $(F_n,\theta_n) = (F,\theta)$ for any fixed $(F,\theta)$ such that $\theta$ is on the boundary of $\Theta_0(F_0)$ (but not on the boundary of $\Theta$). \qed

\subsection{Auxiliary Lemmas for Theorem \ref{thm:momineq-sz}}
The proof of Theorem \ref{thm:momineq-sz} uses four important lemmas. 
Lemma \ref{lem:argmin_converge} establishes a condition under which the projection onto a sequence of polyhedra converges when the coefficient matrix defining the polyhedra converges. 
The condition is verified in a special context in Lemma \ref{lem:set_consistency}, which is used to prove part (b) of Theorem \ref{thm:momineq-sz}. 
The conditions for part (a) are not strong enough for us to apply Lemma \ref{lem:argmin_converge} because we do not restrict the rank of $A(\theta)$. 
Nonetheless, Lemma \ref{BExistence} shows that inequalities redundant in finite sample but relevant in the limit can be added to guarantee the condition of Lemma \ref{lem:argmin_converge}, and help us to prove part (a) of Theorem \ref{thm:momineq-sz}. 
Lemma \ref{InactivityIsUnchanged} shows that the additional inequalities from Lemma \ref{BExistence} do not change the definition of $\hat\beta$. 

First we define some notation. 
For any $d_A\times d_m$ real-valued matrix $A$ and vector $h\in \R_{+,\infty}^{d_A}:=[0,\infty]^{d_A}$, let $\textup{poly}(A,h)=\{\mu\in\R^{d_m}: A\mu\le h\}$ denote the polyhedron defined by inequalities with coefficients given by $A$ and constants given by $h$. 
Also define 
\begin{equation}
\mu^*(x;A,h) = \underset{\mu{\scriptstyle\in\textup{poly}}(A,h)}{\text{argmin}}\|x-\mu\|^2. 
\end{equation}
The lemma considers a sequence of $d_A\times d_m$ real-valued matrices $\{A_n\}_{n=1}^\infty$ and a sequence of $d_A\times 1$ vectors $h_n\in \R_{+}^{d_A}:=[0,\infty)^{d_A}$ such that, as $n\to\infty$, $A_n\to A_0$ and $h_n\to h_0$ for a $d_A\times d_m$ real-valued matrix $A_0$ and a vector $h_0\in \R_{+,\infty}^{d_A}$. 
Also, let $x_n\in \R^{d_m}$ be a sequence of vectors such that $x_n\to x_0\in \R^{d_m}$ as $n\to\infty$. 
We say that a sequence of sets, $\textup{poly}(A_n,h_n)$, Kuratowski converges to a limit set, $\textup{poly}(A_0,h_0)$, denoted by 
\begin{equation}
\textup{poly}(A_n,h_n)\overset{K}{\rightarrow}\textup{poly}(A_0,h_0), 
\end{equation}
if (i) for every $x_0\in\textup{poly}(A_0,h_0)$ there exists a sequence $x_n\in\textup{poly}(A_n,h_n)$ such that $x_n\rightarrow x_0$, and (ii) for every subsequence $n_q$ and for every converging sequence $x_{n_q}\in\textup{poly}(A_{n_q},h_{n_q})$ that converges to a point $x_0$, we have $x_0\in \textup{poly}(A_0,h_0)$.\footnote{One can check that this definition of Kuratowski convergence is equivalent to other definitions given in, for example \cite{AubinFrankowska1990}.} 

\begin{lemma}\label{lem:argmin_converge} 
If $\textup{poly}(A_n,h_n)\overset{K}{\rightarrow}\textup{poly}(A_0,h_0)$, then $\mu^*(x_n;A_n,h_n)\to \mu^*(x_0; A_0, h_0)$. 
\end{lemma}

We denote submatrices of $A_n$ and $A_0$ formed by the rows with indices in $J\subseteq\{1,...,d_A\}$ by $A_{J,n}$ and $A_{J,0}$. 
Important for the following lemma is the fact that every element of $h_n$ is nonnegative for all $n$. 
\begin{lemma}\label{lem:set_consistency} 
If for all $J\subseteq \{1,\dots,d_A\}$, $\textup{rk}(A_{J,n}) = \textup{rk}(A_{J,0})$ for all $n$, then $\textup{poly}(A_n,h_n)\overset{K}{\rightarrow}\textup{poly}(A_0,h_0)$. 
\end{lemma}

For any $d_A\times d_m$ matrix, $A$, and for any vector, $g$, let $J(x; A,g)= \{j\in\{1,\dots,d_A\}:a_j'\mu^*(x;A,g)=g_j\}$. 
This generalizes the previous notation for active inequalities to make explicit the dependence on $A$ and $g$. 
Also let $[A;B]$ denotes the vertical concatenation of two matrices, $A$ and $B$. 

\begin{lemma}\label{BExistence}
Let $A_n$ be a sequence of $d_A\times d_m$ matrices such that each row is either zero or belongs to the unit circle. 
Let $g_n$ be a sequence of nonnegative $d_A$-vectors. 
Then, there exists a subsequence, $n_q$, a sequence of $d_B\times d_m$ matrices, $B_q$, and a sequence of nonnegative $d_B$-vectors $h_q$ such that the following hold. 
\begin{enumerate}[label=(\alph*)]
\item[\textup{(a)}] $A_{n_q}\rightarrow A_0$, $B_q\rightarrow B_0$, $g_{n_q}\rightarrow g_0$, and $h_{q}\rightarrow h_0$ (some of the elements of $g_0$ and $h_0$ may be $+\infty$, in which case the convergence/divergence occurs elementwise). 
\item[\textup{(b)}] $\textup{poly}(A_{n_q},g_{n_q})\subseteq \textup{poly}(B_q,h_{q})$ for all $q$. 
\item[\textup{(c)}] For all $q$ and for all $x\in \textup{poly}(A_{n_q},g_{n_q})$, $$\textup{rk}(I_{J(x;A_{n_q},g_{n_q})}A_{n_q})=\textup{rk}([I_{J(x;A_{n_q},g_{n_q})}A_{n_q};I_{J(x;B_q,h_q)}B_{q}]).$$ 
\item[\textup{(d)}] $\textup{poly}([A_{n_q};B_q],[g_{n_q};h_{q}])\overset{K}{\rightarrow} \textup{poly}([A_0;B_0],[g_0;h_0])$ as $q\rightarrow\infty$. 
\end{enumerate}
\end{lemma}

Suppose that $\bar j\in J(x;A,g)$ and $a_{\bar j}\neq \mathbf{0}$. 
If such a $\bar j$ does not exist, let $\tau_j(x; A,g)=0$ for all $j\in\{1,...,d_A\}$. 
Otherwise, let
\begin{align}
\tau_j(x;A,g) = \left\{\begin{array}{ll}\frac{\|a_{\bar j}\|(g_j-a'_j\mu^*(x;A,g))}{\|a_{\bar j}\|\|a_j\|-a'_{\bar j}a_j}&\text{ if } \|a_{\bar j}\|\|a_j\|\neq a'_{\bar j}a_j\\ \infty&\text{ otherwise.}\end{array}\right.\label{zAgj}
\end{align}
Let $\tau(x;A,g) = \inf_{j\in\{1,\dots,d_A\}} \tau_j(x;A,g)$. 
One can verify that the definition of $\tau(x;A,g)$ does not depend on which $\bar j\in J(x;A,g)$ is used to define it, when more than one is available. 
This definition coincides with the definition of $\hat\tau$ or $\tau(X)$, as used in the proof of Theorem \ref{thm:normal-rcc}, making explicit the dependence on $A$ and $g$. 

\begin{lemma}
\label{InactivityIsUnchanged}
If $\textup{poly}(A,g)\subseteq\textup{poly}(B,h)$, then $\tau(x;A,g)=\tau(x;[A;B],[g;h])$ for all $x\in\R^{d_m}$. 
\end{lemma}

\subsection{Proof of Theorems \textup{\ref{thm:asysize-iid}} and \textup{\ref{thm:momineq-sz}}}

\begin{proof}[Proof of Theorem \textup{\ref{thm:asysize-iid}}] 
We verify the conditions of Theorem \ref{thm:momineq-sz}. 
We first show that Assumption \ref{assu:iidasy} implies that Assumption \ref{assu:para_space}(a) holds for any sequence $\{(\theta_n,F_n): F_n\in {\cal F}, \theta_n\in\Theta_0(F_n)\}$. 
Fix an arbitrary sequence $\{(\theta_n,F_n): F_n\in {\cal F}, \theta_n\in\Theta_0(F_n)\}$. 
Let $\Sigma_n = \textup{Var}_{F_n}(m(W_i,\theta_n))$, which does not depend on $i$ due to Assumption \ref{assu:iidasy}(a). 
Let $D_n$ be the diagonal matrix formed by the diagonal elements of $\Sigma_n$. 
By Assumption \ref{assu:iidasy}(b), $D_n$ is invertible and thus we can define
\begin{align}
\Omega_n = D_n^{-1/2}\Sigma_nD_n^{-1/2}= \textup{Corr}_{F_n}(m(W_i,\theta_n)).
\end{align}
The elements of $\Omega_n\in[-1,1]$ which is a compact set. 
Thus, for any subsequence of $\{n\}$, there is a further subsequence $\{n_q\}$ such that
\begin{align}
\Omega_{n_q}\to \Omega,
\end{align}
for some matrix $\Omega$. 
By Assumption \ref{assu:iidasy}(c), $\Omega$ is positive definite. 

Now consider an arbitrary vector $a\in \R^{d_m}$ such that $a'a=1$, and consider the sequence of random variables:
\begin{align}
&n^{1/2}a'D_q^{-1/2}(\overline{m}_{n_q}(\theta_{n_q}) -\E_{F_{n_q}}\overline{m}_{n_q}(\theta_{n_q})) \nonumber\\
&= n^{-1/2}\sum_{i=1}^n a'D_q^{-1/2}(m(W_i,\theta_{n_q}) - \E_{F_{n_q}}m(W_i,\theta_{n_q}))\nonumber\\
&\to_d N(0,a'\Omega a),\label{convdistr}
\end{align}
by the Lindeberg-Feller central limit theorem where the Lindeberg condition holds because
\begin{align}
\E_{F_{n_q}}|a'D_q^{-1/2}m(W_i,\theta_{n_q})|^{2+\epsilon}\leq \E_{F_{n_q}}[\sum_{j=1}^{d_m}a_j|m_j(W_i,\theta_{n_q})/\sigma_{F_{n_q},j}(\theta_{n_q})|^{2+\epsilon}]\leq M<\infty,
\end{align}
where the first inequality holds by the convexity of $g(x) = |x|^{2+\epsilon}$ and the second and the third inequalities hold by Assumption \ref{assu:iidasy}(d). 

The Cramer-Wold device combined with (\ref{convdistr}) proves (\ref{weakconvergence}) in Assumption \ref{assu:para_space}.

To show (\ref{varianceconsistency}), consider that
\begin{align}
&D_q^{-1/2}\widehat{\Sigma}_{n_q}(\theta_{n_q})D_q^{-1/2}\nonumber\\
& = n^{-1}\sum_{i=1}^nD_q^{-1/2}(m(W_i,\theta_{n_q})-\E_{F_{n_q}}m(W_i,\theta_{n_q}))(m(W_i,\theta_{n_q})-\E_{F_{n_q}}m(W_i,\theta_{n_q}))'D_q^{-1/2}\nonumber\\
& - D_q^{-1/2}(\overline{m}_{n_q}(\theta_{n_q})-\E_{F_{n_q}}m(W_i,\theta_{n_q}))(\overline{m}_{n_q}(\theta_{n_q})-\E_{F_{n_q}}m(W_i,\theta_{n_q}))'D_q^{-1/2}. 
\end{align}
By Assumptions \ref{assu:iidasy}(a) and (d), the law of large numbers for rowwise i.i.d. triangular arrays applies and gives us
\begin{align}
&n^{-1}\sum_{i=1}^nD_q^{-1/2}(m(W_i,\theta_{n_q})-\E_{F_{n_q}}m(W_i,\theta_{n_q}))(m(W_i,\theta_{n_q})-\E_{F_{n_q}}m(W_i,\theta_{n_q}))'D_q^{-1/2} \nonumber\\
&\to_p \Omega,
\end{align}
and similarly
\begin{align}
D_q^{-1/2}(\overline{m}_{n_q}(\theta_{n_q})-\E_{F_{n_q}}m(W_i,\theta_{n_q}))\to_p \mathbf{0}.
\end{align}
Thus, (\ref{varianceconsistency}) is also verified. 

Next, we show that Assumption \ref{assu:iidasy}, combined with the additional assumptions in Theorem \ref{thm:asysize-iid}(b), implies Assumption \ref{assu:para_space}(b). 
First note that each element of $\Lambda_q$ is either one or $\|e'_jA(\theta_{n_q})D_q\|^{-1}$. 
By the common additional condition for Theorem \ref{thm:asysize-iid}(b,c), $\|e'_jA(\theta_{n_q})D_q\|^{-1}\to \|e_j'A_{\infty}\|^{-1}$. 
Note that $e'_jA(\theta_{n_q})D_q$ cannot go to zero if $e_j'A(\theta_{n_q})\neq \mathbf{0}$ because that would violate the common additional condition for Theorem \ref{thm:asysize-iid}(b,c) for $J=\{j\}$. 
Therefore, there exists a further subsequence along which $\Lambda_q\to \Lambda_\infty$ for  a positive definite diagonal matrix $\Lambda_\infty$. 
Therefore, $\Lambda_q A(\theta_{n_q})D_q\to\bar A_0=\Lambda_\infty A_\infty$. 
Also note that for each $J\subseteq\{1,...,d_A\}$, $\text{rk}(A_J(\theta_{n_q})D_q)=\text{rk}(I_J A_\infty)=\text{rk}(I_J \bar A_0)$, where the first equality follows from the common additional condition for Theorem \ref{thm:asysize-iid}(b,c) and the second equality follows because each row of $\bar A_0$ is a positive scalar multiple of the corresponding row of $A_\infty$. 
This verifies Assumption \ref{assu:para_space}(b). 

We also note that along every further subsequence, each diagonal element of $\Lambda_q$ converges to a positive value. 
This implies that, for part (b), we have for every $j=1,...,d_A$, 
\begin{equation}
\sqrt{n_q}e'_j\Lambda_q(A_j(\theta_{n_q})\E_{F_{n_q}}\overline{m}_{n_q}(\theta_{n_q})-b(\theta_{n_q}))\rightarrow 0. 
\end{equation}
Also, for part (c), we have for every $j\notin J$, 
\begin{equation}
\sqrt{n_q}e'_j\Lambda_q(A_j(\theta_{n_q})\E_{F_{n_q}}\overline{m}_{n_q}(\theta_{n_q})-b(\theta_{n_q}))\rightarrow -\infty. 
\end{equation}
Also, for part (b), $\bar{A}_0\neq \mathbf{0}$ is implied by $A_\infty\neq \mathbf{0}$ because $\Lambda_\infty$ is positive definite. 

Therefore, Theorem \ref{thm:asysize-iid} follows from Theorem \ref{thm:momineq-sz}. 
\end{proof}

\begin{proof}[Proof of Theorem \textup{\ref{thm:momineq-sz}}]
\noindent\underline{We first prove part (a)}. 
Let $\{\theta_n,F_n\}_{n=1}^{\infty}$ be an arbitrary sequence satisfying $F_n\in{\cal F}$ and $\theta_n\in\Theta_0(F_n)$ for all $n$. 
Let $\{n_m\}$ be an arbitrary subsequence of $\{n\}$. 
It is sufficient to show that there exists a further subsequence, $\{n_q\}$, such that as $q\to\infty$,
\begin{align}
\underset{q\to\infty}{\textup{liminf}}\Pr{}_{F_{n_q}}\left(T_{n_q}(\theta_{n_q})\leq
\chi^2_{\hat r, 1-\hat\beta}\right)\geq 1-\alpha. \label{subsequencing}
\end{align}

Fix an arbitrary subsequence, $\{n_m\}$. 
By Assumption \ref{assu:para_space}(a), there exists a further subsequence, $\{n_q\}$, a sequence of positive definite matrices, $D_q$, and a positive definite correlation matrix, $\Omega_0$, such that\footnote{For notational simplicity, we denote all further subsequences by $\{n_q\}$}
\begin{align}
\sqrt{n_q}D_q^{-1/2}(\overline{m}_{n_q}(\theta_{n_q})-\E_{F_{n_q}}\overline{m}_{n_q}(\theta_{n_q}))&\rightarrow_d Y\sim N(0, \Omega_0),\text{ and } \label{convergence}\\
D_q^{-1/2}\widehat\Sigma_{n_q}(\theta_{n_q})D_q^{-1/2}&\rightarrow_p \Omega_0.\label{variance}
\end{align} 
We introduce some simplified notation. 
Let $\widehat\Omega_q=D_q^{-1/2}\widehat\Sigma_{n_q}(\theta_{n_q})D_q^{-1/2}$, $X=\Omega^{-1/2}_0 Y\sim N(0,I)$, $Y_q=\sqrt{n_q}D_{q}^{-1/2}(\overline{m}_{n_q}(\theta_{n_q})-\E_{F_{n_q}}\overline{m}_{n_q}(\theta_{n_q}))$, and $X_q=\widehat{\Omega}_q^{-1/2}Y_q$. 
Equations (\ref{convergence}) and (\ref{variance}) imply that 
\begin{align}
X_q&\to_d X\sim N(0,I),\text{ and } \label{convergence1}\\
\widehat{\Omega}_q&\rightarrow_p \Omega_0.\label{variance1}
\end{align} 

The remainder of the proof proceeds in four steps. 
(A) In the first step, the problem defined in (\ref{Tn}) is transformed to include additional inequalities. 
(B) In the second step, notation is defined for partitioning $\R^{d_m}$ according to Lemma \ref{lem:partition}, for both finite $q$ and the limit. 
(C) In the third step, the almost sure representation theorem is invoked on the convergence in (\ref{convergence1}) and (\ref{variance1}). 
(D) In the final step, we show that (almost surely) the event $T_{n_q}(\theta_{n_q})\le  \chi^2_{\hat r, 1-\hat\beta}$ eventually implies a limiting event based on $X$ and $\Omega_0$. 
This limiting event has probability at least $1-\alpha$ from Theorem \ref{thm:normal-rcc}. 

(A) Consider the sequence of matrices $A(\theta_{n_q})D_q^{1/2}$. 
For each $q$, let $\Lambda_q$ denote a $d_A\times d_A$ diagonal matrix with positive entries on the diagonal such that each row of $\Lambda_q A(\theta_{n_q})D_q^{1/2}$ is either zero or belongs to the unit circle. 
Such a $\Lambda_q$ always exists by taking the diagonal element to be the inverse of the magnitude of the corresponding row of $A(\theta_{n_q})D_q$, if it is nonzero, and one otherwise. 
Let $g_q=\sqrt{n_q}\Lambda_q(b(\theta_{n_q})-A(\theta_{n_q})\E_{F_{n_q}}\overline{m}_{n_q}(\theta_{n_q}))$. 
With this notation, we can write 
\begin{align}
T_{n_q}(\theta_{n_q})&=\inf_{y: \Lambda_q A(\theta_{n_q})D_{q}^{1/2}y\le g_q} (Y_q-y)'\widehat\Omega_q^{-1}(Y_q-y), \label{NewTn2}
\end{align}
which adds and subtracts $\E_{F_{n_q}} \overline{m}_{n_q}(\theta_{n_q})$ in the objective and applies the change of variables, $y=\sqrt{n_q}D_q^{-1/2}(\mu-\E_{F_{n_q}} \overline{m}_{n_q}(\theta_{n_q}))$. 

We can apply Lemma \ref{BExistence} to $\Lambda_q A(\theta_{n_q})D_q^{1/2}$ and $g_q$ to get a further subsequence, $n_q$, a sequence of matrices, $B_q$, a sequence of vectors, $h_q$, matrices $A_0$ and $B_0$, and vectors $g_0$ and $h_0$, satisfying Lemma \ref{BExistence}(a-d). 
Let 
\begin{equation}
\bar A_q=\left[\begin{array}{c}\Lambda_q A(\theta_{n_q})D_{q}^{1/2}\\B_q\end{array}\right] \text{ and } \bar h_q = \left[\begin{array}{c}g_q\\h_q\end{array}\right], \label{InequalityAugmentation}
\end{equation}
and similarly for $\bar A_0$ and $\bar h_0$. 
Let $d_{\bar A}=d_A+d_B$. 
We have that 
\begin{align}
T_{n_q}(\theta_{n_q})&=\inf_{y: \bar{A}_qy\le \bar h_q} (Y_q-y)'\widehat\Omega_q^{-1}(Y_q-y) \label{NewTn3}\\
&=\inf_{t: \bar{A}_q\widehat\Omega^{1/2}_q t\le \bar h_q} (X_q-t)'(X_q-t), \label{NewTn}
\end{align}
where the first equation follows from Lemma \ref{BExistence}(b) and the second equation follows from the change of variables $t=\widehat\Omega^{-1/2}_q y$. 

Equation (\ref{NewTn}) has changed the problem by adding additional inequalities. 
We verify that the rank of the active inequalities is unchanged. 
For any positive definite matrix, $\Omega$, let $\bar J_q( x,  \Omega)$ be the set of indices for the active inequalities in the problem: 
\begin{equation}
\inf_{y: \Lambda_q A(\theta_{n_q}) D_q^{1/2} y \le g_q} (x-y)'\Omega^{-1}(x-y). 
\end{equation}
Recall that $\widehat J$ is the set of active inequalities for the problem defined in (\ref{Tn}), which is equal to $\bar J_q(Y_q,\widehat\Omega_q)$ by a change of variables. 
Similarly, let $J_q(x,\Omega)$ be the set of active inequalities in the problem: 
\begin{equation}
\inf_{t: \bar A_q\Omega^{1/2}t \le \bar h_q} (x-t)'(x-t). \label{qMinimizerDefinition}
\end{equation}
Also let $t^*_q(x, \Omega)$ denote the unique minimizer. 
We have that for any $y\in\R^{d_m}$ and for any positive definite $\Omega$, 
\begin{align}
\textup{rk}(A_{\bar J_q(y,\Omega)}(\theta_{n_q}))&=\textup{rk}(I_{\bar J_q(y,\Omega)}\Lambda_q A(\theta_{n_q})D_q^{1/2})\nonumber\\&=\textup{rk}(I_{J_q(\Omega^{-1/2}y,\Omega)}\bar A_{q})=\textup{rk}(I_{J_q(\Omega^{-1/2}y,\Omega)}\bar A_{q}\Omega^{1/2}), \label{J_hat}
\end{align} 
where the first equality follows because $\Lambda_q$ is diagonal with positive entries on the diagonal and $D_q$ is positive definite, the second equality follows by Lemma \ref{BExistence}(c), and the final equality follows from the fact that $\Omega$ is positive definite. 

Before proceeding to the next step, we simplify the rank calculation by taking a further subsequence. 
Notice that for each $J\subseteq \{1,...,d_{\bar A}\}$, $\textup{rk}(I_J\bar A_q)\in\{1,...,d_m\}$. 
We can denote it by $r^q_J$, and then take a subsequence, $n_q$, so that for all $J$, $r^q_J$ does not depend on $q$. 
Similarly, we define $r^\infty_J=\textup{rk}(I_J\bar A_0)$. 
Note that by the convergence of $\bar A_q$ to $\bar A_0$, $r^q_J\ge r^\infty_J$ for all $J$. 

(B) For any positive definite $d_m\times d_m$ matrix, $\Omega$, and for every $J\subseteq \{0,1,...,d_{\bar A}\}$, let 
\begin{align}
A^q(\Omega)&=\bar A_q \Omega^{1/2}\nonumber\\
{{a}^q_{\ell}}'(\Omega)&=\text{$\ell^{\text{th}}$ row of } A^q(\Omega)\nonumber\\
C^q(\Omega)&=\{x\in\R^{d_m}:  {{a_{\ell}^{q}}'}(\Omega) x\le \bar h_{\ell, q} \text{ for all }\ell =1,...,d_{\bar A}\}\nonumber\\
C^q_J(\Omega)&=\{x\in  C^q:  {{a_{\ell}^{q}}'(\Omega)} x= \bar h_{\ell, q} \text{ for all }\ell\in J \text{ and } {{a_{\ell}^{q}}'(\Omega)} x< \bar h_{\ell, q} \text{ for all }\ell\in J^c\}\nonumber\\
V^q_J(\Omega)&=\left\{\sum_{\ell\in J} v_\ell  a^q_{\ell}(\Omega) : v_\ell\in\R, v_\ell\ge 0\right\}\text{, and }\nonumber\\
K^q_J(\Omega)&= C^q_J(\Omega)+ V^q_J(\Omega).\label{finite-objects}
\end{align}
Furthermore, for every $J\subseteq\{1,...,d_{\bar A}\}$, let $ P^q_J(\Omega)$ denote the projection onto $\text{span}( V^q_J(\Omega))$, and let $ M^q_J(\Omega)$ denote its orthogonal projection. 
There exists a $\kappa^q_J(\Omega)\in \text{span}( V^q_J(\Omega))$ such that for every $ x\in  C^q_J(\Omega)$, $ P^q_J(\Omega)  x=\kappa^q_J(\Omega)$. 
This follows because for two $ x_1, x_2\in C^q_J(\Omega)$, and for any $v\in\text{span}( V^q_J(\Omega))$, $v'( x_1- x_2)=0$, which implies that $ P^q_J(\Omega)( x_1- x_2)=0$. 

For every given $\Omega$, we can apply Lemma \ref{lem:partition} to the objects defined in (\ref{finite-objects}). 
This implies that
\begin{itemize}
\item[(a)] if $x\in  K^q_J(\Omega)$ then $x-t^*_q( x, \Omega)\in  V^q_J(\Omega)$ and $t^*_q( x, \Omega)\in  C^q_J(\Omega)$, 
\item[(b)] the sets $ K^q_J(\Omega)$ for all $J\subseteq \{1,\dots,d_{\bar A}\}$ form a partition of $\R^{d_m}$, and 
\item[(c)] for each $J\subseteq \{1,\dots,d_{\bar A}\}$, we have $x\in  K^q_J(\Omega)\text{ iff }J=J_q(x,\Omega).$
\end{itemize}
These properties imply that, for all $ x\in K^q_J(\Omega)$, we can write 
\begin{align} 
P^q_J (\Omega)x= P^q_J(\Omega) (x- t^*_q( x, \Omega))+ P^q_J(\Omega) t^*_q(x, \Omega)= x-t^*_q(x, \Omega)+ \kappa^q_J(\Omega),\label{Proj_J}
\end{align} 
where the second equality follows by (a) and the definition of $ \kappa^q_J(\Omega)$. 
Then, we can also write $ M^q_J(\Omega) x= x- P^q_J(\Omega) x=t^*_q( x, \Omega)-\kappa^q_J(\Omega)$. 

Let $r^q(x,\Omega)=\text{rk}(\bar A_{J_q(x,\Omega),q})$. 
When $r^q(x,\Omega)=1$, we can define 
\begin{equation}
\tau^q_j(x,\Omega)=\begin{cases}
\frac{\|a_{\bar j}^q(\Omega)\|(\bar h_{j, q}-a_j^q(\Omega)'t^*_q(x,\Omega))}{\|a_j^q(\Omega)\|\|a_{\bar j}^q(\Omega)\|-a_j^q(\Omega)'a_{\bar j}^q(\Omega)}&\text{ if }\|a_j^q(\Omega)\|\|a_{\bar j}^q(\Omega)\|\neq a_j^q(\Omega)'a_{\bar j}^q(\Omega)\\
\infty&\text{ else }
\end{cases}, 
\end{equation}
where $\bar j\in J_q(x,\Omega)$ such that $a_{\bar j}^q(\Omega)\neq 0$. 
We also let $\tau^q(x,\Omega)=\inf_{j=1,...,d_{\bar A}} \tau^q_j(x,\Omega)$, and $\beta^q(x,\Omega)=2\alpha\Phi(\tau^q(x,\Omega))$. 
When $r^q(x,\Omega)\neq 1$, let $\tau^q_j(x,\Omega)=0$, so that $\beta^q(x,\Omega)=\alpha$. 
Note that $\hat\beta=\beta^q(X_q,\widehat\Omega_q)$, where the addition of extra inequalities via Lemma \ref{BExistence} has no effect on $\hat\beta$ or $\hat\tau$ because of Lemma \ref{InactivityIsUnchanged}, where the condition is satisfied by Lemma \ref{BExistence}(b). 

We define similar notation for the limiting objects. 
Let $J^\infty=\{\ell\in\{1,...,d_{\bar A}\}: \bar h_{\ell,0} <\infty\}$. 
These are the indices for the inequalities that are ``close-to-binding.'' 
For any positive definite matrix, $\Omega$, let $A^\infty(\Omega)$ denote the matrix formed by the rows of $\bar A_0\Omega^{1/2}$ associated with the indices in $J^\infty$. 
For notational simplicity, we refer to the rows of $A^\infty(\Omega)$ using $\ell\in J^\infty$ even though the matrix $A^\infty(\Omega)$ has been compressed. 

Let 
\begin{align}
{a^\infty_\ell}'(\Omega)&=\text{$\ell^{\text{th}}$ row of }A^\infty(\Omega)\text{ for }\ell\in J^\infty\nonumber\\
C^{\infty}(\Omega)&=\{x\in\R^{d_m}: {a^\infty_\ell}(\Omega)' x\le \bar h_{\ell,0} \text{ for all }\ell \in J^\infty\}\nonumber\\
C^{\infty}_J(\Omega)&=\{x\in C^{\infty}(\Omega): {a^\infty_\ell}(\Omega)' x= \bar h_{\ell,0} ~\forall\ell\in J \text{ and }{a^\infty_\ell}(\Omega)' x< \bar h_{\ell,0} ~\forall\ell\in J^\infty/ J\}\nonumber\\
V^{\infty}_J(\Omega)&=\left\{\sum_{\ell\in J} v_\ell a^\infty_\ell(\Omega) : v_\ell\in\R, v_\ell\ge 0\right\}\text{, and } \nonumber\\
K^{\infty}_J(\Omega)&=C^{\infty}_J(\Omega)+V^{\infty}_J(\Omega).\label{limit-objects}
\end{align}
Furthermore, for every $J\subseteq J^\infty$, let $P^\infty_J(\Omega)$ denote the projection onto $\text{span}(V^\infty_J(\Omega))$. 
There exists a $\kappa^\infty_J(\Omega)\in \text{span}(V^\infty_J(\Omega))$ such that for every $ x\in C^\infty_J(\Omega)$, 
\begin{align}
P^\infty_J(\Omega)  x=\kappa^\infty_J(\Omega).\label{kappa_J_dagger}
\end{align} 
This follows because for two $ x_1, x_2\in C^\infty_J(\Omega)$, and for any $v\in\text{span}(V^\infty_J(\Omega))$, $v'( x_1- x_2)=0$, which implies that $P^\infty_J(\Omega)(x_1- x_2)=0$. 

Let $h^\infty$ denote the vector formed from the elements of $\bar h_0$ that are finite. 
Let $J^\infty( x,\Omega)$ be the indices for the binding inequalities in the problem: 
\begin{equation}
\inf_{t: A^\infty(\Omega) t\le h^\infty} ( x-t)'( x-t). 
\end{equation}
Also let $t_{\infty}^{*}( x,\Omega)$ denote the unique minimizer. 
We can apply Lemma \ref{lem:partition} to the objects defined in (\ref{limit-objects}). 
This implies that 
\begin{itemize}
\item[($\text{a}^\infty$)] if $x\in K^\infty_J(\Omega)$ then $x-t_{\infty}^{*}(x,\Omega)\in V^\infty_J(\Omega)$ and $t_{\infty}^{*}( x,\Omega)\in C^\infty_J(\Omega)$, 
\item[($\text{b}^\infty$)] the set of all $K^\infty_J(\Omega)$ form a partition of $\R^{d_m}$, and 
\item[($\text{c}^\infty$)] for each $J\subseteq J^\infty$, we have $x\in K^\infty_J(\Omega)$ iff $J=J^\infty( x,\Omega)$. 
\end{itemize}

Let $r^\infty(x,\Omega)=\text{rk}(A^\infty_{J_\infty(x,\Omega)})$. 
When $r^\infty(x,\Omega)=1$, we can define 
\begin{equation}
\tau^\infty_j(x,\Omega)=\begin{cases}
\frac{\|a_{\bar j}^\infty(\Omega)\|(\bar h_{j, 0}-a_j^\infty(\Omega)'t^*_\infty(x,\Omega))}{\|a_j^\infty(\Omega)\|\|a_{\bar j}^\infty(\Omega)\|-a_j^\infty(\Omega)'a_{\bar j}^\infty(\Omega)}&\text{ if }\|a_j^\infty(\Omega)\|\|a_{\bar j}^\infty(\Omega)\|\neq a_j^\infty(\Omega)'a_{\bar j}^\infty(\Omega)\\
\infty&\text{ else }
\end{cases}, 
\end{equation}
where $\bar j\in J^\infty(x,\Omega)$ such that $a_{\bar j}^\infty(\Omega)\neq 0$. 
We also let $\tau^\infty(x,\Omega)=\inf_{j=J^\infty} \tau^\infty_j(x,\Omega)$, and $\beta^\infty(x,\Omega)=2\alpha\Phi(\tau^\infty(x,\Omega))$. 
When $r^\infty(x,\Omega)\neq 1$, let $\tau^\infty_j(x,\Omega)=0$, so that $\beta^\infty(x,\Omega)=\alpha$.

Before proceeding to the next step, consider $M^q_J(\Omega_0)$, which is a sequence of  projection matrices in $\R^{d_m}$ onto a space of dimension $d_m-r^q_J$. 
Since the space of such matrices is compact, we can find a subsequence, $n_q$, such that for all $J\subseteq\{1,...,d_{\bar A}\}$, $M^q_J(\Omega_0)\rightarrow M^N_{J}$, where $M^N_J$ is a projection matrix onto a subspace, $N_J$, of dimension $d_m-r^q_J$.\footnote{Recall that $r_J^q$ does not depend on $q$ due to the construction of the subsequence $\{n_q\}$.} 
Furthermore, for any sequence of positive definite matrices such that $\Omega_q\rightarrow\Omega_0$, we have $M^q_J(\Omega_q)\rightarrow M^N_J$. 
This follows because, if we let $E_q$ denote a $d_m\times r^q_J$ matrix whose columns form an orthonormal basis for $\text{span}(V^q_J(\Omega_0))$ (which is the range of $P^q_J(\Omega_0)$) then for any positive definite matrix, $\Omega$, the columns of ${\Omega^{1/2}}{\Omega_0^{-1/2}}E_q$ form a basis for $\text{span}(V^q_J(\Omega))$, which implies that 
\begin{align}
M^q_J(\Omega_q)&=I_{d_m}-{\Omega_q^{1/2}}{\Omega_0^{-1/2}}E_q(E'_q\Omega_0^{-1/2}\Omega_q{\Omega_0^{-1/2}}E_q)\inv E'_q\Omega_0^{-1/2}\Omega_q^{1/2}\nonumber\\
&=I_{d_m}-E_q(E'_qE_q)\inv E'_q+o(1)=M^q_J(\Omega_0)+o(1).
\end{align}

(C) Next, we invoke the almost sure representation theorem on the convergence in (\ref{convergence1}) and (\ref{variance1}).\footnote{See \cite{VaartWellner1996}, Theorem 1.10.3, for the a.s. representation theorem.} 
Then, we can treat the convergence in (\ref{convergence1}) and (\ref{variance1}) as holding almost surely.\footnote{This can be formalized by defining random variables, $\ddot X_q$, $\ddot X$, and $\ddot\Omega_q$, satisfying $\ddot X_q=_d X_q$, $\ddot X=_d X$, $\ddot \Omega_q=_d \widehat\Omega_q$, $\ddot X_q\rightarrow_{a.s.} \ddot X$, and $\ddot \Omega_q\rightarrow_{a.s.} \Omega_0$.} 
For the rest of the proof of part (a), consider $A^\infty(\Omega)$, $P^\infty_J(\Omega)$, $\kappa_J^\infty(\Omega)$, and the objects defined in (\ref{limit-objects}) and let the objects without the argument $(\Omega)$ denote the objects evaluated at $\Omega_0$. 
For example, $A^\infty = A^\infty(\Omega_0)$. 

We now construct an event, $\Xi\subseteq\R^{d_m}$, such that $\Pr(X\in\Xi)=1$. 
For every $L\subseteq J^\infty$, let 
\begin{align}
&V^{\infty}_{L +}=\{x\in V^\infty_L| \forall L'\subseteq L, \text{ if }r^q_{L'}<r^\infty_L \text{ then } M^N_{L'} x\neq 0\}.
\end{align} 
For each $L\subseteq J^\infty$ such that $r^\infty_L>0$, let 
\begin{align}
\Xi_L&=\{x\in K^\infty_L: P^\infty_Lx-\kappa^\infty_L\in V^\infty_{L+}\text{ and }(P^\infty_L x-\kappa^\infty_L)'(P^\infty_L x-\kappa^\infty_L)\neq \chi^2_{r^\infty_L, 1- \beta^\infty(x)}\}.
\end{align} 
Since $r^\infty_L>0$, $P^\infty_LX\sim N(0,P^\infty_L)$, which is absolutely continuous on $\text{span}(V^\infty_L)$, and therefore the probability that $P^\infty_L X-\kappa^\infty_L$ lies in any one of the finitely many subspaces, $null(M^N_{L'})=\{x\in\R^{d_m}: M_{L'}^Nx=0\}$, each with dimension $r^q_{L'}<r^\infty_L$, is zero. 
Also, $(P^\infty_L X-\kappa^\infty_L)'(P^\infty_L X-\kappa^\infty_L)$ is absolutely continuous because it can be written as the sum of $\textup{rk}(A^\infty_L)$ squared normal random variables. 
Also, $\chi^2_{r^\infty_L,1-\beta^\infty(X)}$ depends on $X$ only through $M^\infty_LX$, which is independent of $P^\infty_LX$. 
Therefore, for each fixed $M^\infty_LX$, the conditional probability that $(P^\infty_L X-\kappa^\infty_L)'(P^\infty_L X-\kappa^\infty_L)=\chi^2_{r^\infty_L,1-\beta^\infty(M^\infty_LX+\kappa^\infty_L)}$ is zero. 
This implies that the unconditional probability is also zero. 
Therefore, 
\begin{align}
\Pr(P^\infty_L X-\kappa^\infty_L\in V^\infty_L / V^\infty_{L+} \text{ or } (P^\infty_L X-\kappa^\infty_L)'(P^\infty_L X-\kappa^\infty_L)=\chi^2_{r^\infty_L, 1-\beta^\infty(X)})=0.\label{almost-sure}
\end{align} 

For $L\subseteq J^\infty$ such that $\textup{rk}(A^\infty_L)=0$, let $\Xi_L=K^\infty_L$. 
Then, let $\Xi=\cup_{L\subseteq J^\infty}\Xi_L$. 
Therefore, by property ($\text{b}^\infty$) and equation (\ref{almost-sure}), $\Pr(X\in \Xi)=1$. 

(D) We consider the set of all sequences such that $x_q\rightarrow x_\infty\in\Xi$ and $\Omega_q\rightarrow\Omega_0$. 
By the definition of $\Xi$ and the almost sure convergence of $(X_q,\widehat{\Omega}_q)$, these sequences occur with probability one. 
Fix such a sequence for the remainder of the proof of part (a). 
For this step, consider $A^q(\Omega)$, $P^q_J(\Omega)$, and the objects defined in (\ref{finite-objects}), and let the objects without the argument $(\Omega)$ denote the objects evaluated at $\Omega_q$. 

Below we show that for each sequence, 
\begin{equation}
\mathds{1}\{ \|x_q-t^*_q(x_q)\|^2\le \chi^2_{r^q(x_q),1-\beta^q(x_q)}\}\ge \mathds{1}\{ \|x_\infty-t_{\infty}^{ *}(x_\infty)\|^2\le \chi^2_{r^\infty (x_\infty),1-\beta^\infty(x_\infty)}\}\label{StepDObjective}
\end{equation}
eventually. 
Notice that by (\ref{NewTn}) and (\ref{qMinimizerDefinition}), the left hand side is equal to $\mathds{1}\{T_{n_q}(\theta_{n_q})\leq\chi^2_{\hat r, 1-\hat\beta}\}$. 
If (\ref{StepDObjective}) holds, then by the bounded convergence theorem, 
\begin{align}
&\underset{q\to\infty}{\textup{liminf}}\Pr{}_{F_{n_q}}(T_{n_q}(\theta_{n_q})\leq\chi^2_{\hat r, 1-\hat\beta})\geq \Pr{}(\|X-t_{\infty}^{ *}(X,\Omega_0)\|^2\le \chi^2_{r^\infty(X,\Omega_0),1-\beta^\infty(X,\Omega_0)}). \label{AlmostDone}
\end{align}
Also, 
\begin{equation}
\Pr{}(\|X-t_{\infty}^{*}(X,\Omega_0)\|^2\le \chi^2_{r^\infty(X,\Omega_0),1-\beta^\infty(X,\Omega_0)})\ge 1-\alpha \label{Done}
\end{equation}
by Theorem \ref{thm:normal-rcc}(a), which applies with $n=1$ and $\overline{m}_n(\theta)=X$ because $t_{\infty}^{ *}(X,\Omega_0)=P_{C^\infty}X$, where $P_{C^\infty}$ is the projection of $X$ onto $C^\infty=\{\mu\in\R^{d_m}: A^\infty(\Omega_0) \mu\le h^\infty\}$. 
Together, (\ref{AlmostDone}) and (\ref{Done}) imply (\ref{subsequencing}) for the given subsequence, $n_q$. 

To finish the proof of part (a), we prove (\ref{StepDObjective}). 
Let $L^\infty$ be the subset of $J^\infty$ for which $x_\infty\in K^\infty_{L^\infty}$. 
We show that 
\begin{equation}
1=\sum_{L\subseteq \{1,...,d_{\bar A}\}}\mathds{1}\{x_q\in K^q_L\}=\sum_{L\subseteq L^\infty: r^q_L\ge r^\infty_{L^\infty}}\mathds{1}\{x_q\in K^q_L\} \label{indicatorconvergence}
\end{equation} 
eventually. 
Property (b) above implies that the first equality holds at every $q$. 
Thus it is sufficient to show the second equality. 
Note that $t^*_q(x_q)\rightarrow t^*_\infty(x_\infty)$ by Lemma \ref{lem:argmin_converge}, using Lemma \ref{BExistence}(d) to verify the condition. 
For the second equality, it is sufficient to show that, for all $L\notin \{L\subseteq L^\infty: r^q_L\ge r^\infty_{L^\infty}\}$, $x_q\notin K^q_L$ eventually. 
Specifically, we consider three cases: (I) $L\not\subseteq J^\infty$, (II) $L\subseteq J^\infty$ but $L\not\subseteq L^\infty$, (III) $L\subseteq L^\infty$ but $r^q_L<r^\infty_{L^\infty}$.

\textup{(I)} Let $L\not\subseteq J^\infty$. 
Then, there exists a $\ell\in L$ such that $\bar h_{\ell, q}\rightarrow\infty$. 
Then ${a_\ell^q}' t^*_q(x_q)< \bar h_{\ell, q}$ eventually because $t^*_q(x_q)\rightarrow t_{\infty}^{*}(x_\infty)$. 
This implies that $t^*_q(x_q)\notin C^q_L$, and therefore by (a), $x_q\notin K^q_L$ eventually. 

\textup{(II)} Let $L\subseteq J^\infty$ but $L\not\subseteq L^\infty$. 
Then, there exists a $\ell\in L$ such that ${a^\infty_\ell}' t_{\infty}^{*}(x_\infty)<\bar h_{\ell,0}$. 
By the fact that ${a_\ell^q}' t^*_q(x_q)\rightarrow {a^\infty_\ell}' t_{\infty}^{*}(x_\infty)$ and $\bar h_{\ell, q}\rightarrow \bar h_{\ell,0}$, we have that ${a_\ell^q}' t^*_q(x_q)<\bar h_{\ell, q}$ eventually. 
This implies that $z^*_q(x_q)\notin C^q_L$, and therefore by property (a) above, $x_q\notin K^q_L$ eventually. 

\textup{(III)} Let $L\subseteq L^\infty$ such that $r^q_L<r^\infty_{L^\infty}$. 
This case is impossible if $r_{L^\infty}^\infty=0$. 
Thus we only need to consider $r_{L^\infty}^\infty>0$. 
Note that $x_\infty-t_{\infty}^{*}(x_\infty)=P^\infty_{L^\infty}x_\infty-\kappa^\infty_{L^\infty}$ by property ($\text{a}^\infty$) above. 
Also, by the definition of $\Xi$ we have $x_\infty\in \Xi_{L^\infty}$, which implies that $x_\infty-t_{\infty}^{*}(x_\infty)\in V^\infty_{L^\infty +}$, which in turn means that $M^N_L(x_\infty-t_{\infty}^{*}(x_\infty))\neq 0$. 
By the convergence, $M^q_L(x_q-t^*_q(x_q))\rightarrow M^N_L(x_\infty-t_{\infty}^{*}(x_\infty))$, we have that $M^q_L(x_q-t^*_q(x_q))\neq 0$ eventually. 
However, if $x_q\in K^q_L$, then by property (a) above, $x_q-t^*_q(x_q)\in V^q_L$, which implies that $M^q_L(x_q-t^*_q(x_q))=0$. 
This means that $x_q\notin K^q_L$ eventually. 
Therefore, (\ref{indicatorconvergence}) holds eventually. 

We next show that 
\begin{equation}
\limsup_{q\rightarrow\infty}\beta^q(x_q)\le \beta^\infty(x_\infty). \label{betaconv}
\end{equation}
It is sufficient to show that for every subsequence, there exists a further subsequence such that (\ref{betaconv}) holds. 
Thus, it is without loss of generality to suppose that $\beta^q$ converges. 
If $\beta^q\to \alpha$, then (\ref{betaconv}) holds simply because $\beta^\infty(x_\infty)\ge\alpha$. 
If $\lim_{q\to\infty}\beta^q>\alpha$, then for every $q$ large enough, there exists a $\bar j_q$ such that $a_{\bar j_q}^q\neq \mathbf{0}$ and ${a_{\bar j_q}^q}'x_q=\bar{h}_{\bar j_q,q}$. 
We can take a further subsequence so that $\bar j_q$ does not vary with $q$ (and denote it by $\bar j$). 
Then $\lim_{q\rightarrow\infty}a^q_{\bar j}=a^\infty_{\bar j}\neq 0$ because each row of $\bar A_q$ is either zero or belongs to the unit circle. 
Also, the fact that $\bar h_{\bar j,q}={a^q_{\bar j}}'x_q\rightarrow {a^\infty_{\bar j}}'x_\infty=\bar h_{\bar j,0}$ implies that $\bar j\in J^\infty$. 
Thus, $\bar j$ can be used to define $\tau_j^\infty(x_\infty)$. 

Take $j\in J^\infty$, and consider two cases. 
(i) For $j$ such that $\|a_{\bar j}^\infty\|\|a_j^\infty\|={a_j^\infty}'a_{\bar j}^\infty$, we have $\tau_j^\infty(t_\infty^\ast(x_\infty);\bar{A}_0\Omega^{1/2},\bar{h}_0)=\infty$. 
(ii) For $j $ such that $\|a_{\bar j}^\infty\|\|a_j^\infty\|\neq{a_j^\infty}'a_{\bar j}^\infty$, we have 
\begin{align}
\tau^q_j(x_q)&=\frac{\|a_{\bar j}^q\|(\bar h_{j, q}-{a_j^q}'t^*_q(x_q))}{\|a_j^q\|\|a_{\bar j}^q\|-{a_j^q}'a_{\bar j}^q}\nonumber\\
&\rightarrow \frac{\|a_{\bar j}^\infty\|(\bar h_{j, 0}-{a_j^\infty}'t^*_\infty(x_\infty)}{\|a_j^\infty\|\|a_{\bar j}^\infty\|-{a_j^\infty}'a_{\bar j}^\infty}=\tau^\infty_j(x_\infty), \label{betaconvergence2}
\end{align}
which uses $a_{\bar j}^q\rightarrow a_{\bar j}^\infty$, $a_j^q\rightarrow a_j^\infty$, $\bar h_{j, q}\rightarrow \bar h_{j, 0}$, and $t^*_q(x_q)\rightarrow t^*_\infty(x_\infty)$ by Lemma \ref{lem:argmin_converge} and Lemma \ref{BExistence}(d). 
Therefore, 
\begin{align}
\lim_{q\rightarrow\infty}\tau^q(x_q)=&\lim_{q\rightarrow\infty}\inf_{j\in\{1,...,d_{\bar A}\}} \tau^q_j(x_q)\nonumber\\
\le& \lim_{q\rightarrow\infty}\inf_{\{j\in J^\infty: \|a_{\bar j}^\infty\|\|a_j^\infty\|\neq{a_j^\infty}'a_{\bar j}^\infty\}} \tau^q_j(x_q)\label{BetaConvergenceFinal1}\\
=&\inf_{\{j\in J^\infty: \|a_{\bar j}^\infty\|\|a_j^\infty\|\neq{a_j^\infty}'a_{\bar j}^\infty\}} \tau^\infty_j(x_\infty)\nonumber\\
=&\inf_{j\in J^\infty} \tau^\infty_j(x_\infty)=\tau^\infty(x_\infty). \nonumber
\end{align}
This shows that (\ref{betaconv}) holds. 

We now verify (\ref{StepDObjective}). 
Notice that 
\begin{align}
&\mathds{1}\{ \|x_q-t^*_q(x_q)\|^2\le \chi^2_{r^q(x_q),1-\beta^q(x_q)}\}\nonumber\\
=&\sum_{J\subseteq \{1,...,d_{\bar A}\}} \mathds{1}\{ \|x_q-t^*_q(x_q)\|^2\le \chi^2_{r^q(x_q),1-\beta^q(x_q)}\}\mathds{1}\{x_q\in K^q_J\}\nonumber\\
=&\sum_{L\subseteq L^\infty: r^q_L\ge r^\infty_{L^\infty}} \mathds{1}\{\| x_q-t^*_q(x_q)\|^2\leq \chi^2_{r^q(x_q),1-\beta^q(x_q)}\}\mathds{1}\{x_q\in  K^q_L\}\nonumber\\
=&\sum_{L\subseteq L^\infty: r^q_L\ge r^\infty_{L^\infty}} \mathds{1}\{\|x_q-t^*_q(x_q)\|^2\leq \chi^2_{r^q_L,1-\beta^q(x_q)}\}\mathds{1}\{x_q\in  K^q_L\}\nonumber\\
\ge&\sum_{L\subseteq L^\infty: r^q_L\ge r^\infty_{L^\infty}} \mathds{1}\{\| x_q-t^*_q(x_q)\|^2\leq \chi^2_{r^\infty_{L^\infty},1-\beta^q(x_q)}\}\mathds{1}\{x_q\in  K^q_L\}\label{rankinequality}\\
\ge& \mathds{1}\{\| x_\infty-t_{\infty}^{*}(x_\infty)\|^2\leq \chi^2_{r^\infty_{L^\infty},1-\beta^\infty(x_\infty)}\}\sum_{L\subseteq L^\infty: r^q_L\ge r^\infty_{L^\infty}}\mathds{1}\{ x_q\in  K^q_L\}\label{tothelimit}\\
=& \mathds{1}\{\|x_\infty-t_{\infty}^{ *}(x_\infty)\|^2\leq \chi^2_{r^\infty_{L^\infty},1-\beta^\infty(x_\infty)}\}\mathds{1}\{x_\infty\in  K^\infty_{L^\infty}\}\nonumber\\
=& \sum_{J\subseteq J^\infty}\mathds{1}\{\| x_\infty-t_{\infty}^{ *}(x_\infty)\|^2\leq \chi^2_{r^\infty_J,1-\beta^\infty(x_\infty)}\}\mathds{1}\{x_\infty\in  K^\infty_{J}\}\nonumber\\
=& \sum_{J\subseteq J^\infty}\mathds{1}\{\| x_\infty-t_{\infty}^{ *}(x_\infty)\|^2\leq \chi^2_{r^\infty(x_\infty),1-\beta^\infty(x_\infty)}\}\mathds{1}\{x_\infty\in  K^\infty_{J}\}\nonumber\\
=& \mathds{1}\{\| x_\infty-t_{\infty}^{ *}(x_\infty)\|^2\leq \chi^2_{r^\infty(x_\infty),1-\beta^\infty(x_\infty)}\}, \label{finalinequality}
\end{align}
where: the first equality follows from property (b); the second equality follows from (\ref{indicatorconvergence}); the third equality follows from property (c); the first inequality follows because $r^q_L\ge r^\infty_{L^\infty}$; the second inequality must hold eventually because, when $r^\infty_{L^\infty}>0$, 
\begin{equation}
\|x_q-t^*_q(x_q)\|^2\rightarrow \|x_\infty-t_{\infty}^{*}(x_\infty)\|^2=\|P^\infty_{L^\infty} x_\infty-\kappa^\infty_{L^\infty}\|^2\neq\chi^2_{r_{L^\infty}^\infty, 1-\beta^\infty(x_\infty)} 
\end{equation}
and 
\begin{equation}
\liminf_{q\rightarrow\infty}\chi^2_{r^\infty_{L^\infty}, 1-\beta^q(x_q)}\ge \chi^2_{r^\infty_{L^\infty},1-\beta^\infty(x_\infty)}
\end{equation}
(the equality follows from property ($\text{a}^{\infty}$) above because $x_\infty\in K^\infty_{ L^\infty}$, the $\neq$ follows because $x\in \Xi_{L^\infty}$, and the $\ge$ follows from (\ref{betaconv})), and when $r^\infty_{L^\infty}=0$, $A^\infty_{L^\infty}=0$, we have $A^q_{L^\infty}=0$ eventually (because each row of $A^\infty$ either belongs to the unit circle or is zero), and therefore, $x_q\in K^q_L$ for $L\subseteq K^\infty$ implies $x_q-t^*_q(x_q)=0$ eventually; the fourth equality follows from (\ref{indicatorconvergence}) and $x_\infty\in K^\infty_{L^\infty}$; the fifth equality follows because all the terms with $J\neq L^\infty$ are zero; the sixth equality follows from ($\text{c}^{\infty}$); and the final equality follows from ($\text{a}^{\infty}$). 
This verifies (\ref{StepDObjective}), proving part (a). \medskip

\noindent{\underline{Next, we prove part (b)}}. 
Consider the sequence $\{F_n,\theta_n\}_{n=1}^\infty$. 
It is sufficient to show that for every subsequence, $n_m$, there exists a further subsequence, $n_q$, such that 
\begin{align}
\lim_{q\rightarrow\infty}\Pr{}_{F_{n_q}}(T_{n_q}(\theta_{n_q})\leq \chi^2_{\hat r, 1-\hat\beta})= 1-\alpha. \label{subsequencing-equality}
\end{align}
The proof follows that of part (a) with the following changes. 

(i) The augmentation of the inequalities with additional inequalities defined by $B_q$ and $h_q$ in (\ref{InequalityAugmentation}) is no longer needed. We can take $\bar A_q=\Lambda_q A(\theta_{n_q})D^{1/2}_q$ and $\bar h_q=g_{n_q}$. 

(ii) Note that for each $j\in\{1,...,d_A\}$, $g_{j,q}$ (which is equal to $\bar h_{j,q}$) is either zero or 
\begin{equation}
\sqrt{n_q}\frac{b_j(\theta_{n_q})-a_j(\theta_{n_q})'\E_{F_{n_q}}\overline{m}_n(\theta_{n_q})}{\|a_j(\theta_{n_q})'D_q\|}\to 0,
\end{equation}
by assumption. 
Thus, $\bar h_0=0_{d_A}$. 

(iii) Without $B_q$, (\ref{J_hat}) still holds without appealing to Lemma \ref{BExistence} by Assumption \ref{assu:para_space}(b). 

(iv) By Assumption \ref{assu:para_space}(b), for all $J\subseteq\{1,...,d_A\}$, $r_J^q= r_J^\infty$ for all $q$. 

(v) The appeal to Lemma \ref{BExistence} below (\ref{indicatorconvergence}) is replaced by Lemma \ref{lem:set_consistency} to get $t^*_q(x_q)\rightarrow t^*_\infty(x_\infty)$. 

(vi) The expression in (\ref{betaconv}) becomes 
\begin{equation}
\lim_{q\rightarrow\infty}\beta^q(x_q)=\beta^\infty(x_\infty). \label{betaconv_eq}
\end{equation}
To show this, we consider two cases. 
In the first case, we have $\textup{rk}(A_{L^\infty}^\infty)=0$. 
Then by (\ref{indicatorconvergence}), $x_q\in K_L^q$ for some $L$ such that $\textup{rk}(A_{L}^q)=0$ eventually. 
This implies that $\beta^q(x_q)=\alpha = {\beta}^\infty(x_\infty)$ eventually. 
In the second case, we have $\textup{rk}(A_{L^\infty}^\infty)\geq 1$. 
Then by (\ref{indicatorconvergence}), $x_q\in K_L^q$ for some $L$ such that $\textup{rk}(A_{L}^q)\geq1$ eventually. 
That is, for large enough $q$, there exists a $\bar j_q$ such that ${a_{\bar j_q}^q}'t^*_q(x_q)=\bar{h}_{\bar j_q, q}$. 
By a subsequencing argument, we can suppose $\bar j_q$ does not depend on $q$, and denote it by $\bar j$. 
If $\|a_{\bar j}^\infty\|\|a_j^\infty\|={a_j^\infty}'a_{\bar j}^\infty$, then $\tau_j^\infty(x_\infty)=\infty$ and $\tau_j^q(x_q)=\infty$, using the fact that $\|a_j^q\|\|a_{\bar j}^q\|={a_j^q}'a_{\bar j}^q$, which follows from the fact that $\text{rk}(A^q_{\{j,\bar j\}})=\text{rk}(A^\infty_{\{j,\bar j\}})$. 
If $\|a_{\bar j}^\infty\|\|a_j^\infty\|\neq{a_j^\infty}'a_{\bar j}^\infty$, then the convergence in (\ref{betaconvergence2}) continues to hold by appealing to Lemma \ref{lem:set_consistency} instead of Lemma \ref{BExistence}(d) to get $t^*_q(x_q)\rightarrow t^*_\infty(x_\infty)$. 
This implies that the inequality in (\ref{BetaConvergenceFinal1}) holds with equality because $J^\infty=\{1,...,d_{\bar A}\}$, and for all $j$ such that $\|a_{\bar j}^\infty\|\|a_j^\infty\|={a_{\bar j}^\infty}'a_j^\infty$, $\tau_j^q(x_q)=\infty$. 
Therefore, (\ref{betaconv_eq}) holds. 

(vii) Now (\ref{StepDObjective}) is satisfied with equality for the following reasons: (1) the inequality in (\ref{rankinequality}) holds with equality because $r^q_L=r^\infty_L$, and thus $r^\infty_L= r^\infty_{L^\infty}$ for all $L\subseteq L^\infty:r_L^q\geq r_{L^\infty}^\infty$, and (2) the inequality in (\ref{tothelimit}) holds with equality eventually using (\ref{betaconv_eq}). 

(viii) An appeal to Theorem \ref{thm:normal-rcc}(b) implies that equality holds in (\ref{Done}). The conditions for Theorem \ref{thm:normal-rcc}(b) are satisfied because $A^\infty\neq \mathbf{0}$ and $h^\infty=\mathbf{0}$. 

Combining these changes with the proof of part (a) proves (\ref{subsequencing-equality}), and therefore, part (b). \medskip

\noindent\underline{Finally, we prove part (c)}. 
We can apply the proof of part (b) twice. 
We can define $\tilde t^*_q(x)$, $\tilde r^q(x)$ and $\tilde\beta^q(x)$ in the same way as $t^*_q(x)$, $r^q(x)$ and $\beta^q(x)$ except with $A_J$ and $b_J$ replacing $A$ and $b$. 
Consider any sequence $\Omega_q\rightarrow\Omega_0$ and $x_q\rightarrow x_\infty\in\Xi$. 
The proof of part (b) applied to the original inequalities ($A$ and $b$) implies (\ref{StepDObjective}) holds with equality eventually. 
We note that the limiting objects are the same when applied to $A_J$ and $b_J$ because $J^\infty\subseteq J$. 
This is because the assumption for part (c) implies that the $j$th element of $g_q=\sqrt{n_q}\Lambda_q(b(\theta_{n_q})-A(\theta_{n_q})\E_{F_{n_q}}\overline{m}_{n_q}(\theta_{n_q}))$ diverges to $+\infty$ for all $j\notin J$. 
Therefore, when we apply the proof of part (b) to the reduced inequalities ($A_J$ and $b_J$), we get that 
\begin{equation}
\mathds{1}\{\|x_q-\tilde t^*_q(x_q)\|^2\le \chi^2_{\tilde r^q(x_q),1-\tilde \beta^q(x_q)}\}=\mathds{1}\{ \|x_\infty-t_{\infty}^{ *}(x_\infty)\|^2\le \chi^2_{r^\infty (x_\infty),1-\beta^\infty(x_\infty)}\}
\end{equation}
eventually as $q\rightarrow\infty$. 
Therefore, 
\begin{equation}
\mathds{1}\{\mathds{1}\{\|x_q-\tilde t^*_q(x_q)\|^2\le \chi^2_{\tilde r^q(x_q),1-\tilde \beta^q(x_q)}\}\neq\mathds{1}\{\|x_q- t^*_q(x_q)\|^2\le \chi^2_{r^q(x_q),1-\beta^q(x_q)}\}\}=0
\end{equation}
eventually. 
Therefore, by the bounded convergence theorem, 
\begin{equation}
\text{Pr}_{F_{n_q}}\left(\phi^{\text{RCC}}_{n_q}(\theta_{n_q},\alpha)\neq\phi^{\text{RCC}}_{n_q, J}(\theta_{n_q},\alpha)\right)\rightarrow 0. \label{adskdj}
\end{equation}
Since, for every subsequence, $n_m$, there exists a further subsequence, $n_q$, such that (\ref{adskdj}) holds, part (c) of Theorem \ref{thm:momineq-sz} follows. 
\end{proof}

\subsection{Proof of Auxiliary Lemmas \ref{lem:argmin_converge} - \ref{InactivityIsUnchanged}}

\begin{proof}[Proof of Lemma \textup{\ref{lem:argmin_converge}}] 
The assumption implies that there exists a sequence, $z^*_n$, such that, for large enough $n$, 
\begin{align}
A_nz_n^\ast\leq h_n \text{ and } z_n^\ast\to \mu^*(x_0,A_0,h_0) \text{ as }n\to\infty. \label{eq:set_consistency}
\end{align}
This implies that 
\begin{align}
\|x_n-\mu^*(x_n,A_n,h_n)\|^2\leq \|x_n-z_n^\ast\|^2\to \|x_0-\mu^*(x_0,A_0,h_0)\|^2.
\end{align}
Taking $\lim\sup$ on both sides, we get
\begin{align}
\lim\sup{}_{n\to\infty}\|x_n-\mu^*(x_n,A_n,h_n)\|^2\leq  \|x_0-\mu^*(x_0,A_0,h_0)\|^2.\label{lowerbd}
\end{align}

Now note that $\mu^*(x_n,A_n,h_n) = \arg\min_{z:A_nz\leq h_n}\|x_n-z\|^2$. 
This sequence of minimizers is necessarily bounded because otherwise (\ref{lowerbd}) cannot hold. 
Thus for any subsequence $\{n_m\}$ there is a further subsequence $\{n_q\}$ such that $\mu^*(x_{n_q},A_{n_q},h_{n_q})\to z_\infty$ for some $z_\infty\in \R^{d_x}$. 
Since $A_{n_q}\mu^*(x_{n_q},A_{n_q},h_{n_q})\leq h_{n_q}$, we have $A_0z_{\infty}\leq h_0$. 
Thus,
\begin{align}
\lim_{q\to\infty}\| x_{n_q}-\mu^*(x_{n_q},A_{n_q},h_{n_q})\|^2= \|x_0-z_\infty\|^2\geq  \|x_0-\mu^*(x_0,A_0,h_0)\|^2 .\label{upperbd0}
\end{align}
Since the subsequence is arbitrary, this implies that
\begin{align}
\text{liminf}_{n\to\infty}\|x_n-\mu^*(x_n,A_n,h_n)\|^2\geq   \|x_0-\mu^*(x_0,A_0,h_0)\|^2.\label{upperbd}
\end{align}
Combining (\ref{lowerbd}) and (\ref{upperbd}), we have $\text{lim}_{n\to\infty}\|x_n-\mu^*(x_n,A_n,h_n)\|^2=   \|x_0-\mu^*(x_0,A_0,h_0)\|^2$. 
This, (\ref{upperbd0}), and the uniqueness of $\arg\min_{z:A_0z\leq h_0}\|x_0-z\|^2$ together imply that 
\begin{equation}
\mu^*(x_n,A_n,h_n)\to z_\infty= \mu^*(x_0,A_0,h_0)\text{ as } n\to\infty,
\end{equation}
proving the lemma.
\end{proof}

\begin{proof}[Proof of Lemma \textup{\ref{lem:set_consistency}}] 
Let $a_{j,0}'$ denote the $j$th row of $A_{0}$ and let $a_{j,n}'$ denote the $j$th row of $A_n$. 
For part (ii) of the definition of convergence, let $n_q$ be a subsequence and $z_q$ be a sequence such that $z_q\in\textup{poly}(A_{n_q},h_{n_q})$ for all $q$ and $z_q\rightarrow z_0$ as $q\rightarrow\infty$. 
Then, 
\begin{equation}
a'_{j,0}z_0 = \lim_{q\rightarrow\infty}a'_{j,n_q} z_q\le \limsup_{q\rightarrow\infty}h_{j,n_q}=h_{j,0}, 
\end{equation}
showing that $z_0\in\textup{poly}(A_0,h_0)$. 

For part (i) of the definition of the convergence, let $z_0\in\textup{poly}(A_0,h_0)$, and let $J_0 = \{j=1,\dots,d_A: a_{j,0}'z_0 = h_{j,0}\}.$ 
If $J_0=\emptyset$, then  $z_n^\ast = z_0$ satisfies the requirement by $A_n\to A_0$ and $h_n\to h_0$. 
If $J_0\neq\emptyset$ but $\textup{rk}(A_{J_0,0})=0$, then $a_{j,0} = \mathbf{0}$ for all $j\in J_0$, which implies that $a_{j,n}=\mathbf{0}$ for all $j\in J_0$ by the rank condition stated in the lemma. 
Then, we can again let $z_n^\ast = z_0$ and $a_{j,n}'z_n^\ast = 0\leq h_{j,n}$ for all $j\in J_0$. 
Again, $\{z_n^\ast\}$ satisfies the requirement due to $A_n\to A_0$ and $h_n\to h_0$. 

Now suppose that $\textup{rk}(A_{J_0,0})>0$. 
The key for the next step is to partition $J_0$ into two subsets $J_0^\ast$ and $J_0^o$. 
We require the partition to satisfy the following conditions:
\begin{itemize}
\item[(i)] $J_0^\ast$ contains $\textup{rk}(A_{J_0,0})$ elements such that $\{a_{j^\ast,0}:j^\ast\in J_0^\ast\}$ has full rank, and for any element in $j^o\in J_0^o$, there exists a unique linear representation $a_{j^o,0} =\sum_{j^\ast\in J_0^\ast} w_{j^o,j^\ast} a_{j^\ast,0}$, where $w_{j^o,j^\ast}:j^\ast\in J_0^\ast$ are real-valued weights. 
\item[(ii)] The linear representation satisfies: for any $j^\ast\in J_0^\ast$ and $j^o\in J_0^o$ such that $w_{j^o,j^\ast}\neq 0$, we have $h_{j^\ast,0}\leq h_{j^o,0}$.  
\end{itemize}
Such a partition always exists. 
To see why, note that the existence of a partition satisfying (i) is guaranteed by linear algebra. 
The number of partitions satisfying (i) is finite because $J_0$ is a finite set. 
If we choose the partition to be one that minimizes $\sum_{j^\ast\in J_0^\ast}h_{j^\ast,0}$ among those satisfying (i), then the chosen partition also satisfies (ii). 

We note that for all $n$, $\textup{rk}(A_{J^*_0, n})=\textup{rk}(A_{J_0})$ implies that for every $j^o\in J^o_0$ and $j^*\in J^*_0$, there exist weights, $w_{j^o,j^*,n}$, such that 
\begin{equation}
a_{j^o,n}=\sum_{j^*\in J^*_0}w_{j^o,j^*,n}a_{j^*,n}. 
\end{equation}
Furthermore, we know that if $w_{j^o,j^*,n}\neq 0$, then $w_{j^o,j^*}\neq 0$. 
This follows because, otherwise, we would have 
\begin{equation}
\textup{rk}(I_{\{j^o\}\cup (J_0^\ast/ \{j^\ast\})}A_n)>\textup{rk}(I_{(J_0^\ast/ \{j^\ast\})}A_n) =\textup{rk}(I_{(J_0^\ast/ \{j^\ast\})}A_0) = \textup{rk}(I_{\{j^o\}\cup(J_0^\ast/ \{j^\ast\})}A_0), 
\end{equation}
contradicting the assumed rank condition. 

Let $A_{J_0^\ast,0}$ denote the submatrix of $A_0$ formed by the rows selected by $J_0^\ast$, and let $A_{J_0^\ast,n}$, $h_{J_0^\ast,0}$, and $h_{J_0^\ast,n}$ be defined analogously. 
Now let $D$ be a $(d_m -|J_0|)\times d_m$ matrix, the rows of which form an orthonormal basis for the orthogonal complement of the space spanned by $\{a_{j,0}: j\in J_0^\ast\}$. 
Then the matrix $\left(\begin{smallmatrix}A_{J_0^\ast,0}\\ D\end{smallmatrix}\right)$ is invertible, which implies that the matrix $\left(\begin{smallmatrix}A_{J_0^\ast,n}\\ D\end{smallmatrix}\right)$ is invertible for large enough $n$. 
Let ${h}^{\wedge}_{J_0^\ast,n}  = \min(h_{J_0^\ast,n},h_{J_0^\ast,0})$ where the minimum is taken element by element. 
Let
\begin{equation}
z_n^\dagger= \left(\begin{smallmatrix}A_{J_0^\ast,n}\\ D\end{smallmatrix}\right)^{-1}\left(\begin{smallmatrix}h^{\wedge}_{J_0^\ast,n}\\ Dz_0\end{smallmatrix}\right).
\end{equation}
It is easy to verify that 
\begin{align}
z_n^\dagger&\to  \left(\begin{smallmatrix}A_{J_0^\ast,0}\\ D\end{smallmatrix}\right)^{-1}\left(\begin{smallmatrix}h_{J_0^\ast,0}\\ Dz_0\end{smallmatrix}\right) = z_0, \text{ and } \label{eq:znconverge}\\
A_{J_0^\ast,n}z_n^\dagger &= h^{\wedge}_{J_0^\ast,n}\leq h_{J_0^\ast,n}.\label{eq:J0star}
\end{align}

If $a_{j,n}'z_n^\dagger\leq h_{j,n}$ for all $j\in J_0^o$ for large enough $n$, then (\ref{eq:set_consistency}) holds with $z_n^\ast = z_n^\dagger$ and we are done. 
Otherwise, let
\begin{equation}
\lambda_n =\left\{\begin{array}{ll}\min\left\{1, \min_{j\in J_0^o:h_{j,0}>0} \frac{h_{j,n}}{a_{j,n}'z_n^\dagger}\right\}&\text{if }\{j\in J_0^o:h_{j,0}>0\}\neq \emptyset
\\1&
\text{otherwise}\end{array}\right..
\end{equation}
This is well-defined for large enough $n$ since $a_{j,n}'z_n^\dagger\to a_{j,0}'z_0 = h_{j,0}$ and thus $a_{j,n}'z_n^\dagger\neq 0$ for large enough $n$. 
Also, by definition $\lambda_n\le 1$, and 
\begin{align}
\lambda_n\to \min_{j\in J_0^o:h_{j,0}>0} \frac{h_{j,0}}{a_{j,0}'z_0} = 1.\label{eq:lambda_converge}
\end{align}

Now let 
\begin{align}
z_n^\ast  = \lambda_n z_n^\dagger. \label{eq:z_nstar_def}
\end{align}
Then for any $j\in J^o_0$ such that $h_{j,0}>0$, we have 
\begin{align}
a_{j,n}'z_n^\ast \leq h_{j,n}. \label{eq: hpos}
\end{align}
For any $j\in J^o_0$ such that $h_{j,0}=0$, we have 
\begin{align}
a_{j,n}'z_n^\ast &=\lambda_n \sum_{j^\ast\in J_0^\ast}w_{j,j^\ast,n} a_{j^\ast,n}'z_n^\dagger \nonumber\\
&= \lambda_n \sum_{j^\ast\in J_0^\ast}w_{j,j^\ast,n} \min(h_{j^\ast,n},h_{j^\ast,0})\nonumber\\
&=0\le h_{j, n}, \label{eq:h0}
\end{align}
where the first equality follows by the definition of the weights, $w_{j,j^*,n}$, the second equality follows from the definition of $z^\dagger_n$, the third equality follows because, if $w_{j,j^*,n}\neq 0$, then $w_{j,j^*}\neq 0$, and therefore $0\le \min(h_{j^\ast,n},h_{j^\ast,0})\le h_{j^\ast,0}\le h_{j,0}=0$ by property (ii) of the partition.

Equations (\ref{eq:znconverge}), (\ref{eq:lambda_converge}), and (\ref{eq:z_nstar_def}) together imply that $z_n^\ast\to z_0$. 
This also implies that, for all $j\notin J_0$, $a_{j,n}'z_n^\ast - h_{j,n}\to a_{j,0}'z_0-h_{j,0}<0$ and thus, for large enough $n$, 
\begin{equation}
A_{\{1,\dots,d_A\}/ J_0,n}z_n^\ast < h_{\{1,\dots,d_A\}/ J_0,n}.
\end{equation}
This combined with equations (\ref{eq:J0star}), $\lambda_n\le 1$, and (\ref{eq:z_nstar_def})-(\ref{eq:h0}) implies that $A_nz_n^\ast\leq h_n$. 
Therefore, $\{z_n^\ast\}$ satisfies the requirement and the lemma is proved.
\end{proof}

\begin{proof}[Proof of Lemma \textup{\ref{BExistence}}] 
The proof of Lemma \ref{BExistence} makes use of three additional lemmas which are stated and proved at the end of this subsection. 
We use $b_{j,q}$ to denote the transpose of the $j^{\text{th}}$ row of $B_q$, and similarly for $a_{j,n_q}$, $a_{j,0}$, and $b_{j,0}$.
An equivalent way to state condition (d) is: 
\begin{enumerate}[label=(\roman*)]
\item[\textup{(i)}] for any further subsequence, $n_q$, and for every sequence $x_q\in \textup{poly}(A_{n_q},g_{n_q})\cap \textup{poly}(B_q,h_{q})$ such that $x_q\rightarrow x_0$, $x_0\in \textup{poly}(A_0,g_0)\cap \textup{poly}(B_0,h_0)$, and 
\item[\textup{(ii)}] for every $x_0\in \textup{poly}(A_0,g_0)\cap \textup{poly}(B_0,h_0)$, there exists $x_q\in \textup{poly}(A_{n_q},g_{n_q})\cap \textup{poly}(B_q,h_{q})$ such that $x_q\rightarrow x_0$.
\end{enumerate}

Before proving the lemma, we note that for any subsequence, $n_q$ such that $A_{n_q}\to A_0$ and $g_{n_q}\to g_0$, and for any $B_q\rightarrow B_0$ and $h_q\rightarrow h_0$, 
condition (d)(i) is satisfied. 
Specifically, let $x_q$ denote a sequence that belongs to $\textup{poly}(A_{n_q},g_{n_q})\cap \textup{poly}(B_q,h_{q})$ for all $q$, and such that $x_q\rightarrow x_0$. 
Then 
\begin{equation}
a'_{j,0}x_0 = \lim_{q\rightarrow\infty}a'_{j,n_q}x_q\le \lim_{q\rightarrow\infty} g_{j,n_q}=g_{j,0}. 
\end{equation}
 
Also, by the convergence of $h_q$, we have that 
\begin{equation}
b'_{j,0}x_0 = \lim_{q\rightarrow\infty} b'_{j,q}x_q\le \lim_{q\rightarrow\infty} h_{j,q}=h_{j,0}.
\end{equation}
Therefore, $x_0\in \textup{poly}(A_0,g_0)\cap \textup{poly}(B_0,h_0)$. 

We also note that for any $q$, $B_q$, and $h_q$ satisfying (b), condition (c) must also be satisfied. 
If not, then there exists a $q$, an $x\in \textup{poly}(A_{n_q},g_{n_q})$  and a $j'\in J(x;B_q,h_q)$ such that $b_{j',q}$ cannot be written as a linear combination of $a_{j,n_q}$ for $j\in J(x;A_{n_q},g_{n_q})$. 
This implies that there exists a $v$ such that $b'_{j',q}v>0$ and $v\perp a_{j,n_q}$ for all $j\in J(x;A_{n_q},g_{n_q})$. 
But then, $x+\alpha v\in \textup{poly}(A_{n_q},g_{n_q})$ for sufficiently small $\alpha$, at the same time that $b'_{j',q}(x+\alpha v)>h_q$. 
This contradicts the fact that $\textup{poly}(A_{n_q},g_{n_q})\subseteq \textup{poly}(B_q,h_q)$. 
Therefore, (c) holds. 

We now prove the lemma by finding a subsequence, $n_q$, and sequences $\{B_q\}$ and $\{h_q\}$ that satisfy conditions (a), (b), and (d)(ii). 
We first consider $A_n$ and $g_n$. 
By the compactness of the unit circle, let $n_q$ be a subsequence so that $A_{n_q}$ converges to some $A_0$. 
Also suppose $g_{n_q}$ converges along the subsequence to some vector $g_0\in(\R_{+}\cup\{+\infty\})^{d_A}$. 

Let $J^+_A$ denote the subset of $\{1,...,d_A\}$ for which $g_{j,0}>0$, and let $J^0_A$ denote the subset for which $g_{j,0}=0$. 
Consider $A_{J^0_A, 0}$, which defines a cone in $\R^{d_m}$: $\textup{poly}(A_{J^0_A,0}, 0)=\{x\in\R^{d_m}: A_{J^0_A,0}x\le 0\}$. 
Let $S$ denote the smallest linear subspace of $\R^{d_m}$ that contains this cone. 
Let the dimension of $S$ be $d_S$. 
Let $J^S_A$ be the subset of $J^0_A$ for which $a_{j, 0}\perp S$ for all $j\in J^S_A$. 
Let $J^N_A=\{1,...,d_A\}/ J^S_A$. 

Next, we define sequences $B_q$ and $h_q$ that satisfy conditions (a), (b), and (d)(ii) by induction on the dimension of $S$. 
If $d_S=0$, then no $B_q$ or $h_q$ is required. 
Condition (a) is satisfied by the above choice of the subsequence. 
Condition (b) is satisfied because $\textup{poly}(B_q,h_q)=\R^{d_m}$ for all $q$. 
Condition (d)(ii) is satisfied because $\textup{poly}(A_0,g_0)=\{0\}$, and then we can take $x_q=0$ for all $q$, which belongs to $\textup{poly}(A_{n_q},g_{n_q})$ and converges to $x_0=0\in \textup{poly}(A_0,g_0)$. 

If $d_S>0$, then suppose that the conclusion of Lemma \ref{BExistence} holds for all values of the dimension of $S$ less than $d_S$. 

Let $C_q=\textup{poly}(A_{J^S_A, n_q},g_{J^S_A, n_q})$. 
Let $C_q^S$ be the projection of $C_q$ onto $S$. 
That is, $C_q^S=\{P_S x: x\in C_q\}$, where $P_S$ denotes the projection onto $S$ and $M_S=I-P_S$. 
The fact that $C_q$ is a polyhedral set (defined by finitely many affine inequalities) implies by Theorem 19.3 in \cite{Rockafellar1970} that $C_q^S$ is also a polyhedral set. 
Therefore, there exists a $d_{B_1}\times d_m$ matrix of unit vectors in $S$, $B^1_q$ and a vector $h^1_q$ such that $C_q^S=\{y\in S: B^1_q y\le h^1_q\}$. 
We note that $C_q^S$ contains zero, so $h^1_q\ge 0$. 
In the special case of $d_S = d_m$, $C_q^s=C_q$ and we let $B_q^1$ be the matrix composed of all the non-zero rows of $A_{n_q}$ and let $h_q^1$ be the corresponding elements of $g_{n_q}$. 

Let $n_q$ be a further subsequence so that $B^1_q\rightarrow B^1_0$ and $h^1_q\rightarrow h^1_0$, where some of the elements of $h^1_0$ may be $+\infty$, in which case the convergence holds elementwise. 
We note that this construction satisfies conditions (a) and (b) because $\textup{poly}(A_{n_q},g_{n_q})\subseteq C_q\subseteq \textup{poly}(B^1_q,h^1_q)$ for all $q$, where the second subset holds because $B^1_q x=B^1_q M_S x+B^1_qP_Sx=B^1_qP_S x\le h^1_q$ for all $x\in C_q$ because the rows of $B^1_q$ belong to $S$ and $P_S x\in C^S_q$. 

Let $J^{+}_B$ denote the set of $j\in\{1,...,d_{B_1}\}$ for which $h^1_{j,0}>0$, and let $J^{0}_B$ denote the set for which $h^1_{j,0}=0$, where $h^1_{j,0}$ is the $j$th element of $h^1_0$. 
Consider $B^1_{J^{0}_B,0}$ and $A_{J^0_A,0}$, which together define a cone in $S$: $\{x\in S: B^1_{J^{0}_B,0}x\le 0 \text{ and } A_{J^0_A, 0}x\le 0\}$. 
As before, let $S^\dagger$ denote the smallest linear subspace of $S$ that contains this cone. 
Let $J^{S^\dagger}_B$ denote the set of all $j\in J^{0}_B$ for which $b^1_{j,0}\perp S^\dagger$. 
Also let $J^{S^\dagger}_A$ denote the set of all $j\in J^0_A$ for which $a_{j,0}\perp S^\dagger$. 
Let the dimension of $S^\dagger$ be $d_{S^\dagger}$. 

If $d_{S^\dagger}<d_S$, then the result follows by the induction assumption. 
In particular, if we let  
\[
\tilde A_q=\left[\begin{array}{c}A_{n_q}\\ B^1_{q}\end{array}\right] \text{ and } \tilde g_q=\left[\begin{array}{c}g_{n_q}\\ h^1_q\end{array}\right], 
\]
then the subspace, $\tilde S$, defined to be the smallest linear subspace containing $\textup{poly}(\tilde A_0, \tilde g_0)$, is equal to $S^\dagger$. 
Therefore, there exists a further subsequence, $n_q$, and another matrix of inequalities, $B^2_q$ and $h^2_q$ such that: (a) $B^2_q\rightarrow B^2_0$ and $h^2_q\rightarrow h^2_0$, (b) $\textup{poly}(\tilde A_q,\tilde g_q)\subseteq \textup{poly}(B^2_q,h^2_q)$ for all $q$ along the subsequence, and (d)(ii) $\textup{poly}(\tilde A_q,\tilde g_q)\cap \textup{poly}(B^2_q,h^2_q)\rightarrow \textup{poly}(\tilde A_0,\tilde g_0)\cap \textup{poly}(B^2_0,h^2_0)$ pointwise. 
It is easy to see that these conditions imply conditions (a), (b), and (d)(ii) for the original $A_n$ and $g_n$ along this subsequence, with 
\[
B_q=\left[\begin{array}{c}B^1_q\\B^2_q\end{array}\right] \text{ and } h_q=\left[\begin{array}{c}h^1_q\\h^2_q\end{array}\right],
\]
using the fact that $\textup{poly}(\tilde A_q,\tilde g_q)=\textup{poly}(A_{n_q},g_{n_q})\cap \textup{poly}(B^1_q,h^1_q)$. 

Therefore, we only need to show condition (d)(ii) in the case that $d_{S^\dagger}=d_S$. 
In this case, $S=S^\dagger$, and so $J^{S^\dagger}_B=\emptyset$ and $J^{S^\dagger}_A=J^S_A$. 
Fix $x_0\in \textup{poly}(A_0,g_0)\cap \textup{poly}(B^1_0,h^1_0)$. 
We show that for every $\epsilon>0$ there exists a $Q$ such that for all $q\ge Q$ there exists a $y_q\in \textup{poly}(A_{n_q},g_{n_q})\cap \textup{poly}(B^1_q,h^1_q)$ such that $\|y_q-x_0\|\le 2\epsilon$. 
If true, then this can be used to construct a sequence satisfying $y_q\rightarrow x_0$, establishing condition (d)(ii). 

Fix $\epsilon>0$. 
By Lemma \ref{SlackExistence}, there exists a point, $\tilde x$, in $S$ that satisfies ${b^1_{j,0}}' \tilde x<h^1_{j,0}$ for all $j\in \{1,...,d_{B_1}\}$, and $a'_{j,0}\tilde x<g_{j,0}$ for all $j\in J^N_A$. 
There exists a $\lambda\in(0,1)$ small enough that $x^\dagger=\lambda\tilde x+(1-\lambda)x_0\in \bar B(x_0,\epsilon)$, where $\bar B(x_0,\epsilon)$ denotes the closed ball of radius $\epsilon$ around $x_0$. 
Note that $x^\dagger$ satisfies $a'_{j,0}x^\dagger<g_{j,0}$ for all $j\in J^N_A$ and ${b^1_{j,0}}'x^\dagger<h_{j,0}$ for all $j\in \{1,...,d_{B_1}\}$. 
Therefore, there exists a $\delta\in (0,\epsilon)$ and a $Q$ such that for all $q\ge Q$, and for all $x\in \bar B(x^\dagger,\delta)$, ${b^1_{j,q}}' x<h^1_{j,q}$ for all $j\in \{1,...,d_{B_1}\}$, and $a'_{j,n_q}x<g_{j,n_q}$ for all $j\in J^N_A$.  
Notice that, for all $q\ge Q$, $x^\dagger\in C^S_q=\{y\in S: B^1_q y\le h^1_q\}$, which means that there exists a $y_q\in C_q$ such that $x^\dagger=P_S y_q$. 
By Lemma \ref{CloseToS} applied to $K=\{x^\dagger\}$ (where the condition is satisfied because, by Lemma \ref{InequalityRepresentationofS}, $S=\{x\in\R^{d_x}: A_{J^S,0}x\le 0\}$), there exists a larger $Q$ such that for all $q\ge Q$, $y_q\in \bar B(x^\dagger,\delta)$. 
Therefore, $\|y_q-x_0\|\le 2\epsilon$. 
\end{proof}

\begin{proof}[Proof of Lemma \ref{InactivityIsUnchanged}]
Fix $x\in\R^{d_m}$. 
The fact that $\textup{poly}(A,g)\subseteq\textup{poly}(B,h)$ implies that $\mu^*(x;A,g)=\mu^*(x;[A;B],[g,h])$. 
Denote the common value by $\mu^*$. 

If there does not exist a $\bar j\in J(x;A,g)$ such that $a_{\bar j}\neq 0$, then $x=\mu^*$ and $a'_j x<g_j$ for all $j\notin J(x;A,g)$. 
Suppose, to reach a contradiction, that there does exist a $\bar j\in J(x;[A;B],[g;h])$ such that $b_{\bar j-d_A}\neq 0$. 
Then, there would exist a point, $y$, very close to $x$ (say, $y=x+\epsilon b_{\bar j-d_A}$ for some $\epsilon>0$) such that $y\notin \textup{poly}(B,h)$ but $y\in\textup{poly}(A,g)$. 
This contradicts the assumption that $\textup{poly}(A,g)\subseteq\textup{poly}(B,h)$. 
Therefore, there does not exist a $\bar j\in J(x;[A;B],[g;h])$ such that $b_{\bar j-d_A}\neq 0$. 
This implies that, in this case, $\tau_j(x,A,g)=0$ for all $j\in\{1,...,d_A\}$ and $\tau_j(x;[A;B],[g;h])=0$ for all $j\in\{1,...,d_A+d_B\}$. 
Therefore, $\tau(x;A,g)=\tau(x;[A;B],[g;h])$. 

Suppose there does exist a $\bar j\in J(x;A,g)$ such that $a_{\bar j}\neq 0$. 
Then, the same $\bar j$ can be used to define $\tau_j(x;[A;B],[g;h])$ because $J(x;A,g)\subseteq J(x;[A;B],[g;h])$. 

We show that for every $j=1,...,d_B$, $\tau_{j+d_A}(x;[A;B],[g;h])\ge \tau(x;A,g)$. 
The result holds trivially if $\|b_j\|\|a_{\bar j}\|=b'_ja_{\bar j}$ because then $\tau_{j+d_A}(x;[A;B],[g;h])=\infty$. 
Suppose, to reach a contradiction, that $\tau_{j+d_A}(x;[A;B],[g;h])<\tau(x;A,g)$. 
Let $\tau^*=\tau_{j+d_A}(x;[A;B],[g;h])$, and consider two cases. 

(i) If $\tau^*=0$, then for some $\epsilon>0$, the point $t^*=\mu^*+\epsilon (I_{d_m}-a_{\bar j}a'_{\bar j}\|a_{\bar j}\|^{-2})b_j$ belongs to $\textup{poly}(A,g)$ but not $\textup{poly}(B,h)$. 
To see that $t^*\in\textup{poly}(A,h)$, note that for all $\ell\in J(x;A,g)$, the fact that $\tau(x;A,g)>0$ implies that $a_\ell$ is collinear with $a_{\bar j}$. 
Then, $a'_\ell t^*=a'_\ell \mu^*=g_\ell$. 
For all $\ell\notin J(x;A,g)$, $a'_\ell\mu^*<g_\ell$, so $\epsilon$ can be chosen small enough that $a'_\ell t^*<g_\ell$ for all $\ell\notin J(x;A,g)$. 
To see that $t^*\notin \textup{poly}(B,h)$, note that 
\begin{equation}
b'_jt^*=b'_j\mu^*+\epsilon\|b_j\|^2-\epsilon(b'_j a_{\bar j})\|a_{\bar j}\|^{-2}=h_j+\epsilon\|b_j\|^2-\epsilon(b'_j a_{\bar j})^2\|a_{\bar j}\|^{-2}>h_j,
\end{equation}
where the second equality follows because $\tau^*=0$ and $b_j$ is not collinear with $a_{\bar j}$ (so $b'_j\mu^*=h_j$), and the inequality follows because $(b'_j a_{\bar j})^2<\|a_{\bar j}\|^2\|b_j\|^2$. 
This contradicts the assumption that $\textup{poly}(A,g)\subseteq\textup{poly}(B,h)$. 
Therefore, in this case, $\tau^*\ge \tau(x;A,g)$. 

(ii) If $\tau^*>0$, then let $t^*=\mu^*+\tau^*\left(\frac{b_j}{\|b_j\|}-\frac{a_{\bar j}}{\|a_{\bar j}\|}\right)$. 
We show that $t^*$ belongs to the interior of $\textup{poly}(A,g)$ but is on the boundary of $\textup{poly}(B,h)$. 
Note that for every $\ell\in\{1,...,d_A\}$, 
\begin{equation}
a'_\ell t^*=a'_\ell\mu^*+\tau^*\left(\frac{a'_\ell b_j}{\|b_j\|}-\frac{a'_\ell a_{\bar j}}{\|a_{\bar j}\|}\right)=a'_\ell\mu^*+\frac{\|a_{\bar j}\|(h_j-b'_j\mu^*)}{\|a_{\bar j}\|\|b_j\|-b'_ja_{\bar j}}\left(\frac{a'_\ell b_j}{\|b_j\|}-\frac{a'_\ell a_{\bar j}}{\|a_{\bar j}\|}\right). \label{InactivityUnchanged1}
\end{equation}
When $a_\ell$ is collinear with $a_{\bar j}$, the right hand side of (\ref{InactivityUnchanged1}) is less than $g_\ell$ because $a'_\ell\mu^*\le g_\ell$, $h_j>b'_j\mu^*$ (because $\tau^*>0$), and $\|a_{\bar j}\|\|b_j\|>a'_{\bar j}b_j$ (because $b_j$ is not collinear with $a_{\bar j}$).  
When $a_\ell$ is not collinear with $a_{\bar j}$, the right hand side of (\ref{InactivityUnchanged1}) is less than $g_\ell$ because 
\begin{align}
&\|a_{\bar j}\|(h_j-b'_j\mu^*)\left(\frac{a'_\ell b_j}{\|b_j\|}-\frac{a'_\ell a_{\bar j}}{\|a_{\bar j}\|}\right)-(g_\ell - a'_\ell\mu^*)(\|a_{\bar j}\|\|b_j\|-b'_ja_{\bar j})\nonumber\\
<&(h_j-b'_j\mu^*)\left(\|a_{\bar j}\|\left(\frac{a'_\ell b_j}{\|b_j\|}-\frac{a'_\ell a_{\bar j}}{\|a_{\bar j}\|}\right)-(\|a_{\bar j}\|\|a_\ell\|-a'_\ell a_{\bar j})\right)\nonumber\\
\le&0, 
\end{align}
where the first inequality follows because
\begin{equation}
(g_\ell - a'_\ell\mu^*)(\|a_{\bar j}\|\|b_j\|-b'_ja_{\bar j})>(h_j-b'_j\mu^*)(\|a_{\bar j}\|\|a_\ell\|-a'_\ell a_{\bar j})
\end{equation}
(by the assumption that $\tau^*<\tau(x;A,g)\le\tau_\ell(x,A,g)$), and the second inequality follows because $b'_j\mu^*\le h_j$ and $\|a_\ell\|\|b_j\|\ge a'_\ell b_j$. 
This shows that $t^*$ is on the interior of $\textup{poly}(A,g)$. 
We also show that $t^*$ is on the boundary of $\textup{poly}(B,h)$ by calculating that $b'_j t^*=h_j$. 
By a similar calculation to above, we see that 
\begin{align}
&(b'_j t^*-h_j)(\|b_j\|\|a_{\bar j}\|-b'_ja_{\bar j})\nonumber\\
=&\|a_{\bar j}\|(h_j-b'_j\mu^*)\left(\|b_j\|-\frac{b'_ja_{\bar j}}{\|a_{\bar j}\|}\right)+(b'_j\mu^*-h_j)(\|b_j\|\|a_{\bar j}\|-b'_ja_{\bar j})=0. 
\end{align}
This implies that there exists a point, $y$, very close to $t^*$ (say $y=t^*+\epsilon b_j$ for some $\epsilon>0$) such that $y\notin \textup{poly}(B,h)$ but $y\in\textup{poly}(A,g)$. 
This contradicts the assumption that $\textup{poly}(A,g)\subseteq\textup{poly}(B,h)$. 
Therefore, $\tau_{j+d_A}(x;[A;B],[g;h])\ge\tau(x;A,g)$ for all $j=1,...,d_B$. 
\end{proof}

\begin{lemma}
\label{SlackExistence}
Let $A$ be a $d_A\times d_m$ matrix. Let $g$ be nonnegative. Let $J^+$ denote the subset of $\{1,...,d_A\}$ such that $g_j>0$, and let $J^0$ denote the subset of $\{1,...,d_A\}$ such that $g_j=0$. Let $S$ denote the smallest linear subspace containing $\textup{poly}(A_{J^0},0)=\{x\in\R^{d_m}: A_{J^0}x\le 0\}$. Let $J^S$ be the subset of $J^0$ for which $A_{J^S}\perp S$. Let $J^N=\{1,...,d_A\}/ J^S$. There exists a $\tilde x\in S$ such that $a'_{j}\tilde x<g_j$ for all $j\in J^N$. 
\end{lemma}

\begin{proof}[Proof of Lemma \textup{\ref{SlackExistence}}]
First, let $M>\max_{j\in J^+}\|a_j\|$, and let $\epsilon\in (0,\min_{j\in J^+}\{g_{j}\}/M)$. 
Then, for all $\tilde x\in \bar B(0,\epsilon)$, $a'_{j}\tilde x<g_{j}$ for all $j\in J^+$, where $\bar B(x,\epsilon)$ denotes the closed ball of radius $\epsilon$ around $x$. 
Also, for every $j\in J^N\cap J^0$, $\{x\in S: a'_{j}x=0\}$ defines a subspace of $S$. 
We note that for all $j\in J^N\cap J^0$, $\{x\in S: a'_{j}x=0\}$ is a proper subset of $S$, because otherwise $j$ would belong to $J^S$. 
By the definition of $S$, $S\cap \textup{poly}(A_{J^N\cap J^0},0)$ is not contained within any of these subspaces. 
In particular, for each $j\in J^N\cap J^0$, we can find a $\tilde x_j$ and a neighborhood, $N_j$, (relatively open in $S$) that belongs to $S\cap \textup{poly}(A_{J^N\cap J^0, 0},0)/\{x\in S: a'_{j}x=0\}$. 
Indeed, we can consider $j\in J^N\cap J^0$ sequentially, and define each neighborhood to be a subset of the previous one. 
Therefore, the final $\tilde x_j$ must belong to $S\cap \textup{poly}(A_{J^N\cap J^0, 0},0)$ and satisfy $a'_{j}\tilde x<0$ for all $j\in J^N\cap J^0$. 
Take $\tilde x=\lambda \tilde x_j$, where $\lambda>0$ is small enough that $\tilde x\in \bar B(0,\epsilon)$. 
Then, $\tilde x$ satisfies $a'_{j}\tilde x<g_{j}$ for all $j\in J^N$. 
\end{proof}

\begin{lemma}\label{CloseToS}
Let $A_n \to A_0$ and $g_n\to 0$, where $g_n\geq 0$ for all $n$. 
Suppose $S=\{x\in {\R}^{d_m}:A_0x\leq 0\}$ is a linear subspace of ${\R}^{d_m}$. 
Let $S^\perp$ denote the orthogonal subspace to $S$ in $\R^{d_m}$. 
Let $P_S x$ denote the projection of $x\in \R^{d_m}$ onto $S$ and let $M_Sx$ denote $x-P_Sx$. 
Then, for every $K\subseteq S$, compact, and for every $\varepsilon>0$, we have
\begin{equation}
\{x\in \textup{poly}(A_n,g_n):P_S x\in K, \|M_Sx\|\geq \varepsilon\}=\emptyset
\end{equation}
eventually as $n\to\infty$. 
\end{lemma}
\begin{proof}[Proof of Lemma \textup{\ref{CloseToS}}]
Suppose that the conclusion of the lemma is not true. 
Then there exists a sequence $\{x_n\in \textup{poly}(A_n,g_n)\}$ and a subsequence $n_m$ such that $P_Sx_{n_m}\in K$ and $\|M_Sx_{n_m}\|\geq \varepsilon$ for all $m\ge 1$. 
Define the unit vector $x_{n_m}^\perp = M_Sx_{n_m}/\|M_Sx_{n_m}\|$. 
Then, by the compactness of $K$ and the unit circle, there exists a further subsequence $n_q$ such that $P_Sx_{n_q}\to x^S$ and $x_{n_q}^\perp\to x^\perp$  for some $x^S\in S $ and $x^\perp\in S^\perp$ as $q\to\infty$. 

Because $x^\perp\in S^\perp$ and $x^\perp \neq 0$, we know that $x^\perp\notin S=\{x\in\R^{d_m}: A_0 x\le 0\}$, and therefore there exists a $j$ such that
\begin{align}
a'_{j,0} x^\perp>0.\label{perp}
\end{align}
Also, since $x^S\in S$, $a'_{j,0}x^S\leq 0$. Since $S$ is a linear subspace, we have $a'_{j,0}(-x^S)\leq 0$ as well. 
This shows that $a'_{j,0}x^S =0$ (and more generally, $S=\{x\in \R^{d_m}: A_0 x=0\}$). 

Now consider
\begin{align}
a'_{j,n_q}x_{n_q} - g_{j,n_q} &=a'_{j,n_q}P_Sx_{n_q} +a'_{j,n_q}M_Sx_{n_q} - g_{j,n_q}\nonumber\\
&=o(1)+a'_{j,0}x^S + \|M_Sx_{n_q}\| (o(1)+a'_{j,0}x^\perp) - o(1)\nonumber\\
&=o(1)+ \|M_Sx_{n_q}\| (o(1)+a'_{j,0}x^\perp).
\end{align}
By (\ref{perp}), $o(1)+a'_{j,0}x^\perp>0$ eventually. 
This, combined with $\|M_S x_{n_q}\|\ge\epsilon$ implies that 
\begin{align}
a'_{j,n_q}x_{n_q} - g_{j,n_q} >0 
\end{align}
eventually. 
This contradicts the definition of the sequence $x_{n}$ which requires that $x_n\in\textup{poly}(A_n,g_n)$ for all $n$. 
\end{proof}

\begin{lemma}
\label{InequalityRepresentationofS}
Let $A$ be a matrix. 
Let $S$ be the smallest linear subspace containing $C=\textup{poly}(A,0)$. 
Let $J=\{j: a_j\perp S\}$. 
Then, $S=\textup{poly}(A_J,0)$. 
\end{lemma}
\begin{proof}[Proof of Lemma \textup{\ref{InequalityRepresentationofS}}]
First, notice that if $x\in S$, then $x\perp a_j$ for all $j\in J$, and therefore, $A_J x=0$, so $x\in \textup{poly}(A_J,0)$. 

To go the other way, let $x\in \textup{poly}(A_J,0)$. 
Lemma \ref{SlackExistence} implies that there exists an $\tilde x\in S$ such that $a'_j \tilde x<0$ for all $j\in J^c$, where $J^c=\{1,...,d_A\}/ J$. 
Consider $y=x+M\tilde x$ for $M$ large. 
We note that $A_J y = A_J x+MA_J\tilde x\le 0$ since $x\in \textup{poly}(A_J,0)$ and $\tilde x\in S\subseteq \textup{poly}(A_J,0)$. 
We also note that for every $j\in J^c$, $a'_j y=a'_j x + M a'_j \tilde x\rightarrow -\infty$ as $M$ diverges. 
Thus, there exists an $M$ large enough that $y\in \textup{poly}(A,0)$. 
This implies that $y\in S$ because $\textup{poly}(A,0)\subseteq S$. 
This also implies that $x=y-M\tilde x\in S$ because $S$ is a linear subspace. 
\end{proof}

\section{Proof of Lemmas \ref{lem:dual}-\ref{lem:rank} in Section \ref{sec:sub}}\label{sec:proof-dualrank}
\begin{proof}[Proof of Lemma \textup{\ref{lem:dual}}] 
Theorem 4.2 of Kohler (1967) shows that the conclusion of the lemma holds if $H(C)$ is a sufficient set of extreme vectors in the cone $\{h\geq 0:h'C=0\}$. 
Here a vector is an extreme vector in $\{h\geq 0:h'C=0\}$ if the rows of $C$ corresponding to non-zero elements of the vector have rank exactly one less than the number of non-zero elements of the vector. 
If $h$ is an extreme vector, then $\{\lambda h:\lambda\geq 0\}$ is an extreme ray of $\{h\geq 0:h'C=0\}$. 
A sufficient set of extreme vectors is a set formed by taking exactly one non-zero vector from each extreme ray.

Now it suffices to show that the set of vertices of the polyhedron $\{h\geq 0:h'C=\mathbf{0},\mathbf{1}'h=1\}$ is a sufficient set of extreme vectors of $\{h\geq 0:h'C=0\}$. 
Let $h_0$ be an extreme vector in the cone $\{h\geq 0:h'C=0\}$. 
Without loss of generality suppose that  $\mathbf{1}'h_0=1$. We show by contradiction that there do not exist $h^{\ast}\neq h_0$ and $h^{\dagger}\neq h_0$ in the polyhedron $\{h\geq 0:h'C=\mathbf{0},\mathbf{1}'h=1\}$ such that $h_0=\lambda h^\ast +(1-\lambda )h^{\dagger}$ for some $\lambda\in (0,1)$. 
Suppose the contrary. 
Then the zero elements in $h_0$ must correspond to zero elements in $h^\ast$ and $h^\dagger$. 
Let $h_{0,+}$ be the subvector of $h_0$ without the zero elements. 
Let $h^\ast_{+}$ be subvector of $h^\ast $ corresponding to the positive elements of $h_0$. Similarly define $h^\dagger_{+}$. 
Then we must have $h^\ast_{+}\neq h_{0,+}$ and $h^\dagger_{+}\neq h_{0,+}$, $h_{0,\neq0} = \lambda h^\ast_{+}+(1-\lambda)h^\dagger_{+}$, and also $h_{+}^\ast\neq \mathbf{0}$ and $h_{+}^\dagger\neq \mathbf{0}$. 
Let $C_{+}$ denote the rows of $C$ corresponding to the positive elements of $h_0$. 
Then we have 
\begin{equation}
\begin{pmatrix}h_{0,+}&h^\ast_{+}&h^\dagger_{+}\end{pmatrix}'C_{+} =0.
\end{equation}
The rank of the matrix $\begin{pmatrix}h_{0,+}&h^\ast_{\neq0}&h^\dagger_{+}\end{pmatrix}$ is at least two. 
This contradicts the premise that $C_{+}$ is only rank-deficient by one (since $h_0$ is an extreme vector of the cone $\{h\geq 0,h'C=0\}$). 
Therefore, $h_0$ is a vertex of the polyhedron $\{h\geq 0:h'C=\mathbf{0},\mathbf{1}'h=1\}$. 

The above shows that all extreme vectors of the cone $\{h\geq 0:h'C=\mathbf{0}\}$ satisfying the normalization $\mathbf{1}'h=1$ are vertices of the polyhedron $\{h\geq 0: h'C=\mathbf{0},\mathbf{1}'h=1\}$. 
This proves that the set of vertices of the polyhedron is a sufficient set of extreme vectors of the cone. 
The result then follows from Theorem 4.2 of Kohler (1967). 
\end{proof}

\begin{proof}[Proof of Lemma \textup{\ref{lem:rank}}] 
Denote $B_Z$, $C_Z$, and $d_Z$ by $B$, $C$, and $d$. 
Let $h_1',\dots,h_{m_1}'$ be all the rows of $H(C)$ orthogonal to $B\hat{\mu}-d$. 
Then by definition, $A_{\widehat{J}} = [B'h_1,\dots,B'h_{m_1}]'$, and thus $\textup{rk}(A_{\widehat{J}}) = \textup{rk}(B'h_1,\dots,B'h_{m_1})$. 
Since $h_1,\dots,h_{m_1}\in \{h\geq 0:h'C=0,h'(B\hat{\mu}-d)=0\}$, we have $B'h_1,\dots,B'h_{m_1}\in \{B'h:h\geq 0,h'C=0,h'(B\hat{\mu}-d)=0\}$. 
This implies that $\textup{rk}(A_{\widehat{J}}) \leq\textup{rk}(\{B'h:h\geq 0,h'C=0,h'(B\hat{\mu}-d)=0\})$. 

Next, suppose that $\tilde{h}_1,\dots,\tilde{h}_{m_2}\in \{h\geq 0:h'C=0,h'(B\hat{\mu}-d)=0\}$ such that $\text{rk}(B'\mathcal{H})= \textup{rk}(B'\tilde{h}_1,\dots,B\tilde{h}_{m_2})$. 
By the definition of $H(C)$,  $\tilde{h}_1,\dots,\tilde{h}_{m_2}$ must all be linear combinations of the rows of $H(C)$. 
In fact, they must all be linear combinations of $h_1,\dots,h_{m_1}$ defined in the first part of the proof because any other row (say, $h_\ast$) of $H(C)$ must satisfy the strict inequality $h_\ast'(B\hat{\mu}-d)>0$ (since they correspond to the inactive inequalities). 
Consequently, $B'\tilde{h}_1,\dots,B'\tilde{h}_{m_2}$ must be  linear combinations of $B'h_1,\dots,B'h_{m_1}$. 
This implies that $\text{rk}(B'\mathcal{H})\leq \textup{rk}({A}_{\widehat{J}})$. 
Therefore, the lemma is proved.
\end{proof}

\section{Asymptotic Validity of the Subvector Tests}\label{sec:asymsub}
\subsection{General Conditions for Asymptotic Validity}\label{app:gensub}

We fix the realization of $\{Z_i\}_{i=1}^n$ and denote it by $z$.\footnote{Technically, $z$ and the objects that are defined given $z$, including $\Theta_0(F_z)$ depend on $n$ as well. We keep this dependence implicit for simplicity.} 

Let ${\cal F}_{z}$ be a collection of distributions $F_{z}$. 
The following high-level assumption is sufficient for the uniform asymptotic validity of the sCC and the sRCC tests. 
This assumption is the conditional version of Assumption \ref{assu:para_space}. 
\begin{assumption}\label{assu:para_space_sub} 
The given sequence $\{(F_{z, n}, \theta_n):F_{z,n}\in{\cal F}_z, \theta_n\in\Theta_0(F_{z,n})\}_{n=1}^\infty$ satisfies, for every subsequence, $n_m$, there exists a further subsequence, $n_q$, and there exists a sequence of positive definite $d_m\times d_m$ matrices, $\{D_q\}$, such that:

\textup{(a)} Under the sequence $\{F_{z,n_q}\}_{q=1}^\infty$,
\begin{equation}
\sqrt{n_q}D_{q}^{-1/2}(\overline{m}_{n_q}(\theta_{n_q})-\E_{F_{z,n_q}}\overline{m}_{n_q}(\theta_{n_q}))\to_d N(\mathbf{0},\Omega), \label{weakconvergence_sub}
\end{equation}
for a positive definite correlation matrix $\Omega$, and 
\begin{equation}
\|D_{q}^{-1/2}\widehat\Sigma_{n_q}(\theta_{n_q})D_{q}^{-1/2}-\Omega\|\rightarrow_p 0. \label{varianceconsistency_sub}
\end{equation}

\textup{(b)} Let $A(\theta)$ and $b(\theta)$ be defined in Lemma \textup{\ref{lem:dual}}. 
$\Lambda_q A(\theta_{n_q}) D_q \rightarrow \bar A_0$ for some $d_A\times d_m$ matrix $\bar A_0$, and for every $J\subseteq\{1,...,d_A\}$, $\textup{rk}(I_JA(\theta_{n_q})D_q)=\textup{rk}(I_J\bar A_0)$, where $\Lambda_q$ is the diagonal $d_A\times d_A$ matrix whose $j$th diagonal entry is one if $e'_j A(\theta_{n_j})=0$ and $\|e'_j A(\theta_{n_q})D_q\|\inv$ otherwise. 
\end{assumption}
The following corollary of Theorem \ref{thm:momineq-sz} shows the uniform asymptotic validity of the sRCC test.

\begin{corollary}\label{cor:momineq-sz} \textup{(a)}
Suppose Assumption \textup{\ref{assu:para_space_sub}(a)} holds for all sequences $\{(F_{z,n},\theta_n) : F_{z,n}\in {\cal F}_z,\theta_n\in\Theta_0(F_{z,n})\}_{n=1}^n$. Then, 
\begin{align*}
\underset{n\to\infty}{\textup{limsup}}\sup_{F_{z\in {\cal F}_{z}}}\sup_{\theta\in\Theta_0(F_{z})}\E_{F_{z}}(\phi^{\textup{sRCC}}_n(\theta,\alpha) |z)\le \alpha.
\end{align*}

Next consider a sequence $\{(F_{z,n},\theta_n):F_{z,n}\in {\cal F}_z, \theta_n\in\Theta_0(F_{z,n})\}_{n=1}^\infty$ satisfying Assumption \textup{\ref{assu:para_space_sub}(a,b)}. 

\textup{(b)}  If, along any further subsequence, for all $j=1,...,d_A$, $\sqrt{n_q}e'_j\Lambda_q(A(\theta_{n_q})\E_{F_{z,n_q}}\overline{m}_{n_q}(\theta_{n_q})-b(\theta_{n_q}))\rightarrow 0$, and if $\bar A_0\neq 0_{d_A\times d_m}$, then 
\[
\lim_{n\rightarrow\infty}\E_{F_{z,n}}\phi^{\textup{sRCC}}_n(\theta_n,\alpha)=\alpha.
\]

\textup{(c)} If, for $J\subseteq\{1,...,d_A\}$, along any further subsequence, $\sqrt{n_q}e'_j\Lambda_q(A(\theta_{n_q})\E_{F_{z,n_q}}\overline{m}_{n_q}(\theta_{n_q})-b(\theta_{n_q}))\rightarrow -\infty$ as $q\rightarrow\infty$, for all $j\notin J$, then 
\[
\lim_{n\rightarrow\infty}\textup{Pr}_{F_{z,n}}(\phi^{\textup{RCC}}_n(\theta_n,\alpha)\neq\phi^{\textup{RCC}}_{n,J}(\theta_n,\alpha))=0.
\]
\end{corollary}
\noindent\textbf{Remark.} Corollary \ref{cor:momineq-sz} follows from Theorem \ref{thm:momineq-sz} because $\Theta_0(F_{z, n})$ has the equivalent representation
\begin{align}
\Theta_0(F_{z, n}) = \{\theta\in\Theta:A(\theta)E_{F_{z,n}}[\overline{m}_n(\theta)|z]\leq b(\theta)\},
\end{align}
by Lemma \ref{lem:dual}. 
There is one subtle point: $A(\theta) = H(C_z(\theta))B_z(\theta)$ might change dimension because $H(C_z)$ might change dimension with $C_z$, and $C_z$ might change with the sample size. 
But this does not cause a problem because the dimension of $H(C_z)$ and thus that of $A(\theta)$ is bounded by a function of $k$ and $p$ which does not change with the sample size.\footnote{This fact is know as the McMullen's upper bound theorem. See e.g. Section 8.4 of  \cite{Ziegler1995}.} 
Due to this boundedness, for any subsequence of $\{n\}$ we can always find a further subsequence along which the dimension of $A(\theta)$ does not change. 
Then the problem falls into the framework of Theorem \ref{thm:momineq-sz}. 

\subsection{Primitive Conditions under i.i.d. Sampling}\label{app:iidsub}
Now we assume that  $\{W_i\}_{i=1}^n$ is an i.i.d. sample unconditionally and derive primitive conditions for Assumption \ref{assu:para_space_sub}(a). 
Let the conditional distribution of $W_i$ given $Z_i=z_i$ be represented by the mapping: $F_{|}:z_i\mapsto F_{|z_i}$. 
Let ${\cal F}_{|}$ denote a collection of $F_|$ and let ${\cal F}_{z} =\{\times_{i=1}^nF_{|z_i}:F_|\in {\cal F}_|\}$, where $\times_{i=1}^nF_{|z_i}$ denotes the joint distribution whose marginal distributions are independent $F_{|z_i}$. 
The following assumption is sufficient for (\ref{weakconvergence_sub}) in Assumption \ref{assu:para_space_sub}. 
In the assumption $\sigma^2_{j|z}(\theta) := n^{-1}\sum_{i=1}^n\textup{Var}_{F_{|z_i}}(m_j(W_i,\theta)|z_i)$ and 
\begin{equation}
D_{|z}(\theta) = \textup{diag}(\sigma^2_{1|z}(\theta),\dots,\sigma^2_{d_m|z}(\theta)).
\end{equation}
Let $\textup{eig}_{\min}(V)$ denote the minimum eigenvalue of a matrix $V$.

\begin{assumption}\label{assu:PWZ} 
There exists an $M_0<\infty$ and an $\epsilon_0>0$ such that for all $F_|\in{\cal F}_|$, the following hold. 
\begin{enumerate}
\item[\textup{(a)}] $\sigma^2_{j|z}(\theta) >0$ for all $j=1,\dots,d_m$, $\theta\in\Theta$, and for all $n$. 
\item[\textup{(b)}] $n^{-1}\sum_{i=1}^n\E_{F_{|z_i}}((m_j(W_i,\theta)/\sigma_{j|z}(\theta))^{4}|z_i) <M_0$ for all $j$, all $\theta\in\Theta$, and for all $n$. 
\item[\textup{(c)}] $\textup{eig}_{\min}(n^{-1}\sum_{i=1}^n[\textup{Var}_{F_{|z_i}}(D_{|z}^{-1/2}(\theta)m(W_i,\theta)|z_i)])>\epsilon_0$ for all $\theta\in\Theta$ and for all $n$. 
\end{enumerate}
\end{assumption}

\noindent\textbf{Remark.} 
Part (b) requires $m(W_i,\theta)$ to have finite 4th moment conditional on $Z_i=z_i$. 
This is used both to derive the asymptotic normality of $\overline{m}_n(\theta)$ using the Lindeberg-Feller central limit theorem under the sequence of $F_{z,n}$, and to show the consistency of the average conditional variance estimator $\widehat{\Sigma}_n(\theta)$. 
Part (c) requires that the average conditional variance of $m(W_i,\theta)$ to be invertible uniformly over $\theta$ and $F_{|}\in{\cal F}_|$. 
This is required since we use the quasi-likelihood ratio statistic which involves inverting an estimator of the average conditional variance. 

When the nearest neighbor matching variance estimator in (\ref{condVar_nn}) is used, the following additional assumption is used for consistency. 
\begin{assumption}\label{assu:PZ}
\begin{enumerate}
\item[\textup{(a)}]  $\{z_i\}_{i=1}^\infty$ is a bounded sequence of distinct values.\medskip

\item[\textup{(b)}] $\Sigma_{Z,n}\to \Sigma_Z$ where $\Sigma_Z$ is finite positive definite matrix.
\item[\textup{(c)}] There exist $M_g>0$ and $M_V>0$ such that for all $\theta\in \Theta$ and $F_{|}\in {\cal F}_{|}$ the conditional mean and variance, $\E_{F_{|z_i}}[D_{|z}(\theta)^{-1/2}m(W_i,\theta)|z_i]$ and $Var_{F_{|z_i}}(D_{|z}(\theta)^{-1/2}m(W_i,\theta)|z_i)$, are Lipschitz continuous in $z_i$ with Lipschitz constants $M_g$ and $M_V$, respectively. 
\end{enumerate}
\end{assumption}
\noindent\textbf{Remark.} 
The boundedness part of part (a) is used to show that  $z_i$ and its nearest neighbor get close to each other on average as $n\to\infty$. 
This can be guaranteed by pre-normalizing $Z_i$ before applying the matching estimator. 
For example, if the raw conditioning variable $\tilde{Z}_i$ is supported in $(0,\infty)$, one can let $Z_i = \Phi(\tilde{Z}_i)$ where $\Phi(\cdot)$ is the standard normal cumulative distribution function. 
This and the Lipschitz continuity in part (c) together ensure that the nearest neighbor provides the correct information about the conditional variance in the limit. 
The distinct value part of part (a) ensures that each point can be the nearest neighbor of at most uniformly bounded number of other points. 
This holds with probability one if $Z_i$ has no probability mass on any single point. 
It can be made to hold by adding a tiny continuous noise to $Z_i$ when $Z_i$ has repeated values. 
The noise should be set small enough to be a tie breaker only in the nearest neighbor calculation. 
Part (b) of the assumption can be established for a probability-one set of $\{Z_i\}$ values by the strong law of large numbers. 

The following theorem verifies Assumption \ref{assu:para_space_sub}(a). 

\begin{theorem}\label{lem:sub-normal} 
\begin{enumerate}
\item[\textup{(a)}] Assumption \textup{\ref{assu:PWZ}} implies \textup{(\ref{weakconvergence_sub})} in Assumption \textup{\ref{assu:para_space_sub}} for all sequences $\{(F_{z,n},\theta_n):F_{z,n}\in{\cal F}_z,\theta_n\in\Theta_0(F_{z,n})\}_{n=1}^\infty$. 
\item[\textup{(b)}] If $\{z_i\}_{i=1}^n$ contains at least two instances of each value eventually as $n\rightarrow\infty$, and Assumption \textup{\ref{assu:PWZ}} holds, then \textup{(\ref{varianceconsistency_sub})} holds for $\widehat{\Sigma}_n(\theta)$ defined in \textup{(\ref{condVar1})}, for all sequences $\{(F_{z,n},\theta_n):F_{z,n}\in{\cal F}_z,\theta_n\in\Theta_0(F_{z,n})\}_{n=1}^\infty$. 
\item[\textup{(c)}] If Assumptions \textup{\ref{assu:PWZ}} and \textup{\ref{assu:PZ}} hold, then \textup{(\ref{varianceconsistency_sub})} holds for $\widehat{\Sigma}_n(\theta)$ defined in \textup{(\ref{condVar_nn})}, for all sequences $\{(F_{z,n},\theta_n):F_{z,n}\in{\cal F}_z,\theta_n\in\Theta_0(F_{z,n})\}_{n=1}^\infty$. 
\end{enumerate}
\end{theorem}

\begin{proof}[Proof of Theorem \textup{\ref{lem:sub-normal}}]
(a) Let $\{(F_{z,n},\theta_n):F_{z,n}\in{\cal F}_z,\theta_n\in\Theta_0(F_{z,n})\}$ be an arbitrary sequence. 
Let $F_{|z_i, n}$ denote the conditional distribution of $W_i$ given $Z_i=z_i$ implied by $F_{z,n}$.  
Let $\sigma_{j|z,n}^2(\theta)$ and $D_{|z, n}(\theta)$ be defined just like $\sigma_{j|z}^2(\theta)$ and $D_{|z}(\theta)$ except with $F_{|z_i}$ replaced by $F_{|z_i,n}$. 
Let $D_n = D_{|z,n}(\theta_n)$. 
Then $D_n$ is positive definite for every $n$ by Assumption \ref{assu:PWZ}(a). 

Let $\Omega_n = D_n^{-1/2}n^{-1}\sum_{i=1}^n\textup{Var}_{F_{|z_i,n}}(m(W_i,\theta_n)|z_i)D_n^{-1/2}$. 
Algebra shows that the square of the $(j,\ell)$th element of $\Omega_n$ is bounded by
\begin{align}
2n^{-1}\sum_{i=1}^n\E_{F_{|z_i,n}}\left[\left(\frac{m_j(W_i,\theta_n)}{\sigma_{j|z,n}(\theta_n)}\right)^4|z_i\right]+2n^{-1}\sum_{i=1}^n\E_{F_{|z_i,n}}\left[\left(\frac{m_\ell(W_i,\theta_n)}{\sigma_{j|z,n}(\theta_n)}\right)^4|z_i\right],
\end{align}
which is bounded by $4M_0$ by Assumption \ref{assu:PWZ}(a). 
Thus $\textup{vec}(\Omega_n)\in [0,4M_0]^{d_m^2}$ which is a compact set. 
This implies that a subsequence $n_q$ can be found for any subsequence of $\{n\}$ such that $\Omega_{n_q}\to \Omega_\infty$. 
Furthermore, Assumption \ref{assu:PWZ}(c) implies that $\Omega_\infty$ is positive definite.

It remains to verify the Lindeberg condition for the Lindeberg-Feller central limit theorem (CLT) along the subsequence $\{n_q\}$. 
Let $a$ be an arbitrary real vector on the unit sphere in $\R^{d_m}$. 
Let 
\begin{equation}
\hat{m}_{n,i}(\theta) = a'D_n^{-1/2}(m(W_i,\theta)-\E_{F_{|z_i,n}}[m(W_i,\theta)|z_i]).
\end{equation}
Let
\begin{align}
s_q^2 = {n_q}^{-1}\sum_{i=1}^{n_q}\E_{F_{|z_i, n_q}}\left[\hat{m}_{n_q,i}(\theta_{n_q})^2|z_i\right].
\end{align}
For an arbitrary $\varepsilon>0$, consider the derivation, 
\begin{align}
&\sum_{i=1}^{n_q} n_q^{-1}s_q^{-2}\E_{F_{|z_i,n_q}}[\hat{m}_{n_q,i}(\theta_{n_q})^2 1\{n_q^{-1}s_q^{-2}\hat{m}_{n_q,i}(\theta_{n_q})^2>\varepsilon\}|z_i]\nonumber\\
&\leq  n_q^{-2}s_q^{-4}\varepsilon^{-1}\sum_{i=1}^{n_q}\E_{F_{|z_i,n_q}}[\hat{m}_{n_q,i}(\theta_{n_q})^4|z_i]\nonumber\\
&\leq 16n_q^{-2}s_q^{-4}\varepsilon^{-1}\sum_{i=1}^{n_q}\E_{F_{|z_i,n_q}} [(a'D_{n_q}^{-1/2}m(W_i,\theta_{n_q}))^4|z_i]\nonumber\\
&\leq 16n_q^{-2}s_q^{-4}\varepsilon^{-1}\sum_{i=1}^{n_q}\E_{F_{|z_i,n_q}} [\|D_{n_q}^{-1/2}m(W_i,\theta_{n_q})\|^4|z_i]\nonumber\\
& =O(n_q^{-1}s_q^{-4}\varepsilon^{-1})\nonumber\\
&\to 0, \text{ as }q\to\infty,
\end{align}
where the first inequality holds because $1(x>\varepsilon)\leq \frac{x}{\varepsilon}$ for any $x\geq 0$, the second inequality holds because $E[(X-E(X))^4]\leq 16 E[X^4]$, the third inequality holds by the Cauchy-Schwarz inequality and $\|a\|=1$, the equality holds by Assumption \ref{assu:PWZ}(b), and the convergence holds because $s_q^2\to a'\Omega_\infty a$ by the definition of the subsequence $\{n_q\}$. 
Therefore, the Lindeberg condition holds and the CLT applies, proving part (a). 

(b) Note that $\widehat{\Sigma}_n(\theta)$ is the weighted average of the standard sample variance estimator within subsamples with same $z_i$ values. 
Thus, by standard argument, we have
\begin{align}
\E_{F_{z,n}}[\widehat{\Sigma}_n(\theta_n)|z]& = \sum_{\ell\in\mathcal{Z}}\frac{n_\ell}{n}\textup{Var}_{F_{|\ell,n}}(m(W_i,\theta_n)|\ell)=\frac{1}{n}\sum_{i=1}^n\textup{Var}_{F_{|z_i,n}}(m(W_i,\theta_n)|z_i),
\end{align}
where the second equality holds by rearranging terms. 
Thus,
\begin{align}
\E_{F_{z,n}}[D_n^{-1/2}\widehat{\Sigma}_n(\theta_n)D_n^{-1/2}|z] = \Omega_n.
\end{align}
Also by standard calculation, the $(j,j')$ element of $D_n^{-1/2}\widehat{\Sigma}_n(\theta_n)D_n^{-1/2}$ has a conditional variance given $z$:
\begin{align}
\frac{1}{n^2}\sum_{i=1}^n\textup{Var}_{F_{|z_i,n}}\left(\frac{m_j(W_i,\theta_n)m_{j'}(W_i,\theta_n)}{\sigma_{j|z,n}(\theta_n)\sigma_{j'|z,n}(\theta_n)}|z_i\right) + \frac{1}{n^2}\sum_{i=1}^n\frac{\omega_{j|z_i,n}^2(\theta_n)\omega_{j'|z_i,n}^2(\theta)+\omega_{jj'|z_i,n}(\theta_n)^2}{n_{z_i}-1},\label{varbound}
\end{align}
where $\omega_{j|z_i,n}(\theta) =\textup{Var}_{F_{|z_i,n}}\left(\frac{m_j(W_i,\theta)}{\sigma_{j|z,n}(\theta)}|z_i\right)$ and $\omega_{jj'|z_i,n}(\theta) = \textup{Cov}_{F_{|z_i,n}}\left(\frac{m_j(W_i,\theta)}{\sigma_{j|z,n}(\theta)},\frac{m_{j'}(W_i,\theta)}{\sigma_{j'|z,n}(\theta)}|z_i\right)$. 
By standard algebraic manipulation, we have
\begin{align}
\textup{Var}_{F_{|z_i,n}}\left(\frac{m_j(W_i,\theta_n)m_{j'}(W_i,\theta_n)}{\sigma_{j|z,n}(\theta_n)\sigma_{j'|z,n}(\theta_n)}|z_i\right)&\leq \frac{1}{2}(M_{ji}+M_{j'i})\text{, and}\nonumber\\
\omega_{j|z_i,n}^2(\theta_n)\omega_{j'|z_i,n}^2(\theta)+\omega_{jj'|z_i,n}(\theta_n)^2&\leq M_{ji}+M_{j'i}
\end{align}
where $M_{ji}  = \E_{F_{|z_i,n}}\left[\left(\frac{m_j(W_i,\theta_n)}{\sigma_{j|z,n}(\theta_n)}\right)^4|z_i\right]$. 
Therefore, by Assumption \ref{assu:PWZ} and the additional assumption that $n_{z_i}\geq 2$ for all $i$, we have that the expression in (\ref{varbound}) is bounded by $\frac{1}{n}(M_0+2M_0)$ which converges to zero as $n\to\infty$. 
This proves part (b). 

(c) First, we prove that
\begin{align}
n^{-1}\sum_{i=1}^n\|z_i-z_{\ell_Z(i)}\|^2\to 0.\label{nnZ}
\end{align}
To begin, define $\tilde{z}_i=\Sigma_{Z,n}^{-1/2}z_i$. 
By Assumption \ref{assu:PZ}(b), $\Sigma_{Z,n}^{-1/2}\to \Sigma_Z^{-1/2}$ as $n\to\infty$ and this limit is finite. 
Thus, $\Sigma_{Z,n}^{-1/2}$ is uniformly bounded over all large enough $n$. 
This and Assumption \ref{assu:PZ}(a) together imply that the elements of the array $\{\tilde{z}_1,\dots,\tilde{z}_n\}_{n\geq 1}$ are chosen from a bounded set. 
Then Lemma 1 of \cite{AbadieImbens2008} applies directly and implies that
\begin{align}
n^{-1}\sum_{i=1}^n\|\tilde{z}_i-\tilde{z}_{\ell_Z(i)}\|^2\to 0.\label{nnZ2}
\end{align}
Consider the derivation
\begin{align}
n^{-1}\sum_{i=1}^n\|z_i-z_{\ell_Z(i)}\|^2 &= n^{-1}\sum_{i=1}^n(\tilde{z}_i-\tilde{z}_{\ell_Z(i)})'\Sigma_{Z,n}(\tilde{z}_i-\tilde{z}_{\ell_Z(i)})\nonumber\\
&\leq n^{-1}\sum_{i=1}^n\|\tilde{z}_i-\tilde{z}_{\ell_Z(i)}\|^2 \textup{eig}_{\max}(\Sigma_{Z,n})\nonumber\\
&\to 0,
\end{align}
where $\textup{eig}_{\max}(\cdot)$ stands for maximum eigenvalue and the convergence holds by (\ref{nnZ2}) and Assumption \ref{assu:PZ}(b). 
This proves (\ref{nnZ}).\medskip

Next consider an arbitrary unit vector $a$ in $\R^{d_m}$, let
\begin{align}
s_{n,i}^2(\theta) = a'D_n^{-1/2}(m(W_i,\theta) - m(W_{\ell_Z(i)},\theta))(m(W_i,\theta) - m(W_{\ell_Z(i)},\theta))'D_n^{-1/2}a.
\end{align}
Then $a'D_n^{-1/2}\widehat{\Sigma}_n(\theta)D_n^{-1/2}a = \frac{1}{2n}\sum_{i=1}^n s_{n,i}^2(\theta_n)$. 
Since $a$ is arbitrary, it suffices to show that for any subsequence of $\{n\}$ there exists a further subsequence $\{n_q\}$ such that
\begin{align}
\frac{1}{2n_q}\sum_{i=1}^{n_q} s_{n_q,i}^2(\theta_{n_q})\to_p a'\Omega_\infty a.\label{aSa}
\end{align}
as $q\to\infty$. 

Let $\hat{m}_{n,i}(\theta)$ be defined in the proof of part (a). 
Then
\begin{align}
&\E_{F_{z,n}}[s_{n,i}^2(\theta_n)|z] \nonumber\\&= \E_{F_{z,n}}[a'(\hat{m}_{n,i}(\theta_n)-\hat{m}_{n,\ell_Z(i)}(\theta_n)+\Delta_{ni})(\hat{m}_{n,i}(\theta_n)-\hat{m}_{n,\ell_Z(i)}(\theta_n)+\Delta_{ni})'a|z]\nonumber\\
&=a'\E_{F_{|z_i,n}}[\hat{m}_{n,i}(\theta_n)\hat{m}_{n,i}(\theta_n)'|z_i]a+a'\E_{F_{|z_{\ell_Z(i)},n}}[\hat{m}_{n,\ell_Z(i)}(\theta_n)\hat{m}_{n,\ell_Z(i)}(\theta_n)'|z_{\ell_Z(i)}]a+a'\Delta_{ni}\Delta_{ni}'a\nonumber\\
&=2a'D_n^{-1/2}\textup{Var}_{F_{|z_i,n}}[m(W_i,\theta_n)|z_i]D_n^{-1/2}a +a'\Delta^V_{ni}a+a'\Delta_{ni}\Delta_{ni}'a,\label{si}
\end{align}
where $\Delta_{ni} = \E_{F_{|z_i,n}}[D_n^{-1/2}m(W_i,\theta_n)|z_i] - \E_{F_{|z_{\ell_Z(i)},n}}[D_n^{-1/2}m(W_{\ell_Z(i)},\theta_n)|z_{\ell_Z(i)}]$, and $$\Delta_{ni}^V = \textup{Var}_{F_{|z_i,n}}[D_n^{-1/2}m(W_i,\theta_n)|z_i] - \textup{Var}_{F_{|z_{\ell_Z(i)},n}}[D_n^{-1/2}m(W_{\ell_Z(i)},\theta_n)|z_{\ell_Z(i)}].$$
By Assumption \ref{assu:PZ}(c) we have,
\begin{align}
\|\Delta_{ni}\|\leq C_g\|z_i-z_{\ell_Z(i)}\|\text{ and }\|\Delta_{ni}^V\|\leq C_V\|z_i-z_{\ell_Z(i)}\|.
\end{align}
Thus, 
\begin{align}
n^{-1}\sum_{i=1}^na'\Delta_{ni}\Delta_{ni}'a\leq n^{-1}\sum_{i=1}^n \|\Delta_{ni}\|^2\leq n^{-1}\sum_{i=1}^n C_g^2\|z_i-z_{\ell_Z(i)}\|^2 \to 0,\label{ada}
\end{align}
and
\begin{align}
n^{-1}\sum_{i=1}^n a'\Delta_{ni}^Va&\leq n^{-1}\sum_{i=1}^n\|\Delta_{ni}^V\|\leq n^{-1}\sum_{i=1}^nC_V\|z_i-z_{\ell_Z(i)}\|\nonumber\\
&\leq C_V\sqrt{n^{-1}\sum_{i=1}^n\|z_i-z_{\ell_Z(i)}\|^2}\to 0.\label{adva}
\end{align}
For an arbitrary subsequence of $n$, consider a further subsequence $\{n_q\}$ such that $\Omega_n\to \Omega_\infty$. 
Such a further subsequence always exists by the proof of part (a). 
Then as $q\to\infty$,
\begin{align}
n_q^{-1}\sum_{i=1}^{n_q} 2a'\textup{Var}_{F_{|z_i, n_q}}[D_{n_q}^{-1/2}m(W_i,\theta_{n_q})a|z_i]a \to 2a'\Omega_\infty a.\label{ava}
\end{align}
Combining (\ref{si}), (\ref{ada}), (\ref{adva}), and (\ref{ava}), we have
\begin{align}
\E_{F_{z,n_q}}[a'D_{n_q}^{-1/2}\widehat{\Sigma}_{n_q}(\theta_{n_q})D_{n_q}^{-1/2}a|z]=\frac{1}{2n_q}\sum_{i=1}^{n_q} \E_{F_{z,n_q}}[s_{n_q,i}^2(\theta_n)|z]\to a'\Omega_\infty a.
\end{align}

Now it suffices to show that 
\begin{align}
\E_{F_{z,n}}\left[\left(n^{-1}\sum_{i=1}^n \left(s_{n,i}^2(\theta_n) - \E_{F_{z,n}}[s_{n,i}^2(\theta_n)|z]\right)\right)^2|z\right]\to 0.\label{sEsconv}
\end{align}
Let $\varepsilon_i(\theta) = a'\hat{m}_{n,i}(\theta)$ and $\sigma_i^2(\theta) = a'\textup{Var}_{F_{|z_i,n}}(D_n^{-1/2}m(W_i,\theta|z_i))a= \E_{F_{|z_i,n}}[\varepsilon_i^2(\theta)|z_i]$. 
Consider that
\begin{align}
n^{-1}\sum_{i=1}^n \left(s_{n,i}^2(\theta_n) - \E_{F_{z,n}}[s_{n,i}^2(\theta_n)|z]\right) &=n^{-1}\sum_{i=1}^n (\varepsilon_i^2(\theta_n) -\sigma_i^2(\theta_n))\nonumber\\
&+n^{-1}\sum_{i=1}^n(\varepsilon_{\ell_Z(i)}^2(\theta_n)-\sigma_{\ell_Z(i)}^2(\theta_n))\nonumber\\
&+2n^{-1}\sum_{i=1}^n(a'\Delta_{ni})\varepsilon_i(\theta_n)\nonumber\\
&-2n^{-1}\sum_{i=1}^n(a'\Delta_{ni})\varepsilon_{\ell_Z(i)}(\theta_n)\nonumber\\
&+2n^{-1}\sum_{i=1}^n\varepsilon_i(\theta_n)\varepsilon_{\ell_Z(i)}(\theta_n).\label{sEs}
\end{align}
Clearly, all the summands on the right-hand side have conditional expectation zero. 
Now we show that the conditional variance (which then is the conditional second moment) of each of them converges to zero. 

For the first summand on the right-hand side of (\ref{sEs}), consider that 
\begin{align}
\E_{F_{z,n}}\left[\left(n^{-1}\sum_{i=1}^n (\varepsilon_i^2(\theta_n) -\sigma_i^2(\theta_n))\right)^2|z\right]& = \frac{1}{n^2}\sum_{i=1}^n\textup{Var}_{F_{|z_i,n}}(\varepsilon_i^2(\theta_n)|z_i)\nonumber\\
&\leq \frac{1}{n^2}\sum_{i=1}^n\E_{F_{|z_i,n}}[\varepsilon_i^4(\theta_n)|z_i]\nonumber\\
&\leq \frac{16}{n^2}\sum_{i=1}^n\E_{F_{|z_i,n}}[(a'D_n^{-1/2}m(W_i,\theta_n))^4|z_i]\nonumber\\
&\leq  \frac{16}{n^2}\sum_{i=1}^n\E_{F_{|z_i,n}}[\|D_n^{-1/2}m(W_i,\theta_n)\|^4|z_i]\nonumber\\
&\to 0,\label{sEs1}
\end{align}
where the convergence holds by Assumption \ref{assu:PWZ}(b). 
For the second summand on the righ-hand side of (\ref{sEs}), consider that
\begin{align}
&\E_{F_{z,n}}\left[\left(n^{-1}\sum_{i=1}^n (\varepsilon_{\ell_Z(i)}^2(\theta_n) -\sigma_{\ell_Z(i)}^2(\theta_n))\right)^2|z\right]\nonumber\\
&=\frac{1}{n^2}\sum_{i=1}^n \E_{F_{z,n}}[\left(\varepsilon_{\ell_Z(i)}^2(\theta_n) -\sigma_{\ell_Z(i)}^2(\theta_n))\right)^2|z]\nonumber\\
&+\frac{2}{n^2}\sum_{i=1}^n\sum_{j=i+1}^n\E_{F_{z,n}}[\left(\varepsilon_{\ell_Z(i)}^2(\theta_n) -\sigma_{\ell_Z(i)}^2(\theta_n)\right)\left(\varepsilon_{\ell_Z(j)}^2(\theta_n) -\sigma_{\ell_Z(j)}^2(\theta_n)\right)|z]\nonumber\\
&\leq \frac{\overline{L}+2\overline{L}^2}{n^2}\sum_{i=1}^n\E_{F_{|z_i,n}}[\left(\varepsilon_i^2(\theta_n) -\sigma_i^2(\theta_n))\right)^2|z_i] \to 0,\label{sEs2}
\end{align}
where $\overline{L}$ is the maximum number of times a $j$ is $\ell_Z(i)$ for some $i$. 
This number is bounded by $3^{d_z}-1$ which does not depend on $n$ (see e.g. \cite{ZegerGersho1994}). 
The convergence holds by (\ref{sEs1}).

For the third summand in (\ref{sEs}), consider that
\begin{align}
\E_{F_{z,n}}\left[\left(n^{-1}\sum_{i=1}^na'\Delta_{ni}\varepsilon_i(\theta_n)\right)^2|z\right] &= \frac{1}{n^2}\sum_{i=1}^n(a'\Delta_{ni})^2\E_{F_{|z_i,n}}[\varepsilon_i^2(\theta_n)|z_i]\nonumber\\
&\leq \frac{C_g\overline{B}}{n^2}\sum_{i=1}^n\E_{F_{|z_i,n}}[\varepsilon_i^2(\theta_n)|z_i]\nonumber\\
&\leq \frac{C_g \overline{B}}{n^2}\sum_{i=1}^n(1+\E_{F_{|z_i,n}}[\varepsilon_i^4(\theta_n)|z_i])\nonumber\\
&\to 0,\label{sEs3}
\end{align}
where $\overline{B}$ is the maximum distance of two points in the sequence $\{z_i\}_{i=1}^n$ which is bounded by Assumption \ref{assu:PZ}(a), the first inequality holds by Assumption \ref{assu:PZ}(c), the second inequality holds by $x^2\leq (\max(1,|x|))^2\leq \max\{1,x^4\}\leq 1+x^4$ and the convergence holds by (\ref{sEs1}). 

For the fourth summand in (\ref{sEs}), consider that
\begin{align}
\E_{F_{z,n}}\left[\left(n^{-1}\sum_{i=1}^na'\Delta_{ni}\varepsilon_{\ell_Z(i)}(\theta_n)\right)^2|z\right] &= \frac{1}{n^2}\sum_{i=1}^n(a'\Delta_i(\theta_n))^2\E_{F_{|z_{\ell_Z(i)},n}}[\varepsilon_{\ell_Z(i)}^2(\theta_n)|z_{\ell_Z(i)}]\nonumber\\
&\leq \frac{C_g\overline{B}}{n^2}\sum_{i=1}^n\E_{F_{z_{\ell_Z(i)},n}}[\varepsilon_{\ell_Z(i)}^2(\theta_n)|z_{\ell_Z(i)}]\nonumber\\
&\leq \frac{C_g \overline{B}}{n^2}\sum_{i=1}^n(1+\E_{F_{|z_{\ell_Z(i)},n}}[\varepsilon_{\ell_Z(i)}^4(\theta_n)|z_{\ell_Z(i)}])\nonumber\\
&\leq \frac{C_g \overline{L}\overline{B}}{n^2}\sum_{i=1}^n(1+\E_{F_{|z_i,n}}[\varepsilon_i^4(\theta_n)|z_i])\nonumber\\
& \to 0,\label{sEs4}
\end{align}
where $\overline{L}$ is number discussed below (\ref{sEs2}).

For the fifth summand on the righ-hand side of (\ref{sEs}), consider that
\begin{align}
&\E_{F_{z,n}}\left[\left(n^{-1}\sum_{i=1}^n\varepsilon_i(\theta_n)\varepsilon_{\ell_Z(i)}(\theta_n)\right)^2|z\right]\nonumber\\
&=\frac{1}{n^2}\sum_{i=1}^n\E_{F_{z,n}}[\varepsilon_i^2(\theta_n)\varepsilon_{\ell_Z(i)}^2(\theta_n)|z]\nonumber\\
&+\frac{2}{n^2}\sum_{i=1}^n\sum_{j=i+1}^n\E_{F_{z,n}}[\varepsilon_i(\theta_n)\varepsilon_{\ell_Z(i)}(\theta_n)\varepsilon_j(\theta_n)\varepsilon_{\ell_Z(j)}(\theta_n)|z]\nonumber\\
&\leq \frac{1}{n^2}\sum_{i=1}^n\E_{F_{z,n}}[\varepsilon_i^2(\theta_n)\varepsilon_{\ell_Z(i)}^2(\theta_n)|z]+\frac{\overline{L}}{n^2}\sum_{i=1}^n\E_{F_{z,n}}[\varepsilon_i^2(\theta_n)\varepsilon_{\ell_Z(i)}^2(\theta_n)|z]\nonumber\\
&\leq\frac{1+\overline{L}}{2n^2} \sum_{i=1}^n\E_{F_{z,n}}[\varepsilon_i^4(\theta_n)+\varepsilon_{\ell_Z(i)}^4(\theta_n)|z]\nonumber\\
&\to 0,\label{sEs5}
\end{align}
where the first inequality holds because $\E_{F_{z,n}}[\varepsilon_i(\theta_n)\varepsilon_{\ell_Z(i)}(\theta_n)\varepsilon_j(\theta_n)\varepsilon_{\ell_Z(j)}(\theta_n)|z]$ is nonzero only when $j= \ell_Z(i)$ and $\ell_Z(j)=i$ and this occurs at most $\overline{L}$ times for each $i$, the second inequality holds by $2xy\leq x^2+y^2$, and the convergence holds by (\ref{sEs1}) and the last two lines of (\ref{sEs4}). 

Combining (\ref{sEs})-(\ref{sEs5}), we have that (\ref{sEsconv}) holds, which then proves part (c).
\end{proof}

\section{Numerical Details for Section \ref{sub:mc-subv}}
\subsection{Calculation of the Identified Set}\label{sub:IDset}
Let $Y_U$ denote $\log(s_{N,i}+2/N) - \log(1-s_{N,i}+\underline{s})$ and let $Y_L$ denote $\log(s_{N,i}+\underline{s})-\log(1-s_{N,i}+2/N)$. 
The value $\theta_0$ should satisfy: There exists $\delta = (\delta_1,\delta_2)'\in \R^2$ such that for $z=(z_c,z_e)'\in \{0,1\}^2$,
\begin{align}
\E[Y_U|z]- \E[X|z]\theta_0\geq \delta_1+\delta_2 z_c\geq \E[Y_L|z]-\E[X|z]\theta_0.\label{bounding}
\end{align}
The identified set for $\theta_0$ can be solved via two linear programming problems once $\E[Y_U|z]$, $\E[Y_L|z]$, and $\E[X|z]$ are calculated. 
Note that 
\begin{equation}
\E[X|z] = \E[1\{2z_e+\varepsilon\geq 0\}] =\Phi(2z_e).
\end{equation}
We also need to calculate $\E[Y_U|z]$ and $\E[Y_L|z]$. 
Let
\begin{align}
\ell(s,N,c)=\sum_{i=0}^N\left(\begin{smallmatrix}N\\ i\end{smallmatrix}\right)s^i(1-s)^{N-i}\log(i+c).
\end{align}
Then
\begin{align}
\E[Y_U|z] &= \E[\ell(s^\ast,N,2) - \ell(1-s^\ast,N,N\underline{s})|z]\nonumber\\
\E[Y_L|z] &= \E[\ell(s^\ast,N,N\underline{s}) - \ell(1-s^\ast,N,2)|z],
\end{align}
where $s^\ast = \frac{\exp(-1\{2z_e+\varepsilon>0\}-z_c+\varepsilon)}{1+\exp(-1\{2z_e+\varepsilon>0\}-z_c+\varepsilon)}.$ 
The conditional expectations can then be calculated by simulating a large number of $\varepsilon$ draws. 
After obtaining, $\E[Y_U|z]$ and $\E[Y_L|z]$, we use linear programming based on (\ref{bounding}) to calculate the upper and the lower bound for $\theta_0$, and we find them to be $[-1.203,-0.757]$. 

\subsection{An Algorithm for Calculating Projected Confidence Sets}

We provide the details for the algorithm used to calculate the Proj-U and Proj-C tests. 
First regress $(\psi_i^U(\theta)+\psi_i^L(\theta))/2$ on $Z_{ci}$ via ordinary least squares to obtain an initial value $\delta_{init1}$, and use this as the starting value for a derivative-free local minimization algorithm (e.g. \textup{fminsearch} in Matlab) to search for a local minimizer $\delta_{\min 1}$. 
If $T_n(\theta,\delta_{\min 1})-\textup{cv}_n(\theta,\delta_{\min 1},1-\alpha)\leq 0$, stop and let $\phi_n^{\textup{Proj}}(\theta,\alpha) = 0$. 
Otherwise, draw a random starting value from $N(\delta_{init1},I)$ and use the local minimization algorithm to find another local minimizer. 
Continue a maximum of 4 times with initial value always drawn from $N(\delta_{init1},I)$, and set $\phi_n^{\textup{Proj}}(\theta,\alpha) = 1$ only if { none} of the local minimizers (say $\delta_\min$) makes $T_n(\theta,\delta_{\min})-\textup{cv}_n(\theta,\delta_{\min},1-\alpha)\leq 0$. 

\bibliography{references}

\end{document}